%% file: main.tex
\documentclass[orivec]{lmcs}


\usepackage{color}
\usepackage[table,dvipsnames]{xcolor}
\usepackage{amsmath}
\usepackage{amssymb} 
\usepackage{wasysym}
\usepackage{graphicx}
\usepackage{xspace}
\usepackage{subfig}
\usepackage{stmaryrd}
\usepackage{enumitem}
\usepackage{tikz}
\usepackage{pgfplotstable}
\usepackage{pgfplots}
\pgfplotsset{compat=newest}
\usepackage{booktabs}
\usepackage[normalem]{ulem}
\usepackage[cal=rsfso]{mathalpha}
\usepackage{lineno}

\usetikzlibrary{automata, positioning, arrows, shapes.misc,decorations.markings}
\tikzset{%
  ->,
  >=stealth',
  node distance=2cm,
  every state/.style={%
    thick, minimum size=4.5mm, inner sep=1.5pt, label distance=4mm},
  every edge/.style={%
    draw, semithick},
  initial text={}, % sets the text that appears on the start arrow
}

\usepackage[T1]{fontenc}

\input{DLmacros}

\begin{document}

%\linenumbers
%\mainmatter

%\title{Minimisation of Spatial Models\\ using Branching Bisimilarity}
\title[Spatial Model Checking via Minimised Models and Branching Bisimilarity]{Spatial Model Checking of Images\\ via Minimised Models and Branching Bisimilarity}

\titlecomment{{\lsuper*}Research partially funded by the Italian MUR Projects PRIN 2017FTXR7S, “IT- MaTTerS”, PRIN 2020TL3X8X “T-LADIES”, and Next Generation EU - MUR Project PNRR PRI ECS00000017 “THE - Tuscany Health Ecosystem”.}

\thanks{The authors are listed in alphabetical order as they equally contributed to the work presented in this paper.
The present paper has been produced without the help of any AI system.}	%optional

% affiliations are numbered automatically with a, b, c (see below)
% use the optional argument to indicate the affiliation(s) of each author
% omit the argument if there is only one author, or only one affiliation

\author[Ciancia]{Vincenzo Ciancia\lmcsorcid{0000-0003-1314-0574}}[a]
\author[Groote]{Jan~Friso~Groote\lmcsorcid{0000-0003-2196-6587}}[b]
\author[Latella]{Diego~Latella\lmcsorcid{0000-0002-3257-9059}}[c]
\author[Massink]{Mieke~Massink\lmcsorcid{0000-0001-5089-002X}}[a]
\author[de Vink]{Erik P. de Vink\lmcsorcid{0000-0001-9514-2260}}[b]

% affiliation 1 (automatically numbered a)
\address{Istituto di Scienza e Tecnologie dell'Informazione ``A. Faedo'', Consiglio Nazionale delle Ricerche, Pisa, Italy}	%optional
\email{Vincenzo.Ciancia@cnr.it, Mieke.Massink@cnr.it}  %optional

% affiliation 2 (automatically numbered b)
\address{Eindhoven University of Technology, Eindhoven, the Netherlands} %optional
\email{J.F.Groote@tue.nl, E.P.d.Vink@tue.nl} %optional

% affiliation 3 (automatically numbered c)
\address{Formerly with Istituto di Scienza e Tecnologie dell'Informazione ``A. Faedo'', Consiglio Nazionale delle Ricerche, Pisa, Italy. Retired}	%optional
\email{diego.latella@actiones.eu} %optional

%% required for running head on odd and even pages, use suitable
%% abbreviations in case of long titles and many authors:

%%%%%%%%%%%%%%%%%%%%%%%%%%%%%%%%%%%%%%%%%%%%%%%%%%%%%%%%%%%%%%%%%%%%%%%%%%%

%% the abstract has to PRECEDE the command \maketitle:
%% be sure not to issue the \maketitle command twice!

\begin{abstract}
Spatial models are of increasing interest in traditional computer science domains and beyond. Spatial minimisation procedures are crucial for efficient model checking of such models that are often large in size. For the recent notion of spatial bisimilarity for quasi-discrete closure models, called ``Compatible Paths'' (\cop) bisimilarity, an effective minimisation method is proposed, and shown to be correct.
Reasoning about space represented by quasi-discrete closure models involves two different conditional reachability modalities: a forward reachability, similar to that used in temporal logic, and a backward modality, representing the fact that a point can be reached from another point, under certain conditions. The core of our minimisation method is the encoding of closure models as labelled transition systems, enabling minimisation algorithms for branching bisimilarity to compute \cop{} equivalence classes. A prototype toolchain, \voxminx{,} is proposed to validate the minimisation method. \voxminx\ preserves the relationship between equivalence classes and sets of pixels in the original image. Experimental validation of the toolchain via benchmark examples demonstrates a promising speed-up in model checking of spatial properties for models of realistic size. 
\end{abstract}

\maketitle

%% start the paper here:

\input{Introduction}

\input{Preliminaries}

\input{CoPabisimilarity}
\input{Translation}

\input{FeasibilityStudy}

\input{Conclusions}

\section*{Acknowledgment}
Research partially supported by bilateral project between CNR (Italy)
and SRNSFG (Georgia) ``Model Checking for Polyhedral Logic''
(\#CNR-22-010); European Union -- Next GenerationEU -- National Recovery
and Resilience Plan (NRRP), Investment~1.5 Ecosystems of Innovation,
Project “Tuscany Health Ecosystem” (THE), CUP: B83C22003930001;
European Union -- Next-GenerationEU -- National Recovery and Resilience
Plan (NRRP) – MISSION~4 COMPONENT~2, INVESTMENT N.~1.1, CALL PRIN 2022
D.D.\ 104 02-02-2022 – (Stendhal) CUP N.~B53D23012850006; MUR project
PRIN 2020TL3X8X ``T-LADIES''; CNR project "Formal Methods in Software
Engineering 2.0", CUP B53C24000720005; Shota Rustaveli National
Science Foundation of Georgia grant
\#FR\nobreakdash-22\nobreakdash-6700.

% \section{Acknowledgements}
% We thank the anonymous reviewers for their valuable suggestions for
% improvement of this work.

%\vfill \mbox{}

%\pagebreak

\bibliographystyle{alpha}
\bibliography{abbr,QdCM2LTS2WAY,xref}

\appendix

\input{AppendixExample}

\end{document}

%% file: DLmacros.tex
\makeatletter
\DeclareRobustCommand{\cev}[1]{%
  \mathpalette\do@cev{#1}%
}
\newcommand{\do@cev}[2]{%
  \fix@cev{#1}{+}%
  \reflectbox{$\m@th#1\vec{\reflectbox{$\fix@cev{#1}{-}\m@th#1#2\fix@cev{#1}{+}$}}$}%
  \fix@cev{#1}{-}%
}
\newcommand{\fix@cev}[2]{%
  \ifx#1\displaystyle
    \mkern#23mu
  \else
    \ifx#1\textstyle
      \mkern#23mu
    \else
      \ifx#1\scriptstyle
        \mkern#22mu
      \else
        \mkern#22mu
      \fi
    \fi
  \fi
}

\makeatother

\newcommand{\RED}[1]{\textcolor{red}{#1}}

\newcommand{\BROWN}[1]{\textcolor{brown}{#1}}

\newcommand{\MAGENTA}[1]{\textcolor{magenta}{#1}}

\newcounter{dgnot} % Notes from Diego
\newenvironment{dgnot}[1][]{\refstepcounter{dgnot}\par\medskip
   \noindent \textbf{\RED{NfD~\thedgnot.}  #1} \rmfamily}{\medskip}

\newcounter{mknot} % Notes from Mieke
\newenvironment{mknot}[1][]{\refstepcounter{mknot}\par\medskip
   \noindent \textbf{\MAGENTA{NfM~\themknot.}  #1} \rmfamily}{\medskip}

\newcounter{vcnot} % Notes from Vincenzo
\newenvironment{vcnot}[1][]{\refstepcounter{vcnot}\par\medskip
   \noindent \textbf{\BROWN{NfV~\thevcnot.}  #1} \rmfamily}{\medskip}

\newcounter{eknot} % Notes from Erik
\newenvironment{eknot}[1][]{\refstepcounter{vcnot}\par\medskip
   \noindent \textbf{\MAGENTA{NfE~\thevcnot.}  #1} \rmfamily}{\medskip}

\newcommand{\calA}{\mathcal{A}}

\newcommand{\calC}{\mathcal{C}}

\newcommand{\calL}{\mathcal{L}}
\newcommand{\calM}{\mathcal{M}}

\newcommand{\calP}{\mathcal{P}}

\newcommand{\calS}{\mathcal{S}}

\newcommand{\calV}{\mathcal{V}}

\newcommand{\SET}[1]{\{#1\}}
\newcommand{\ZET}[2]{\lbrace \, {#1} \mid {#2} \, \rbrace}
\newcommand{\pws}{\mathbf{\calP}}

\newcommand{\nats}{\mathbb{N}}

\newcommand{\cs}{CS}
\newcommand{\qdcs}{QdCS}

\newcommand{\closure}{\calC}

\newcommand{\pthlen}{\mathtt{len}}

\newcommand{\cm}{CM}
\newcommand{\cmc}{CMC}

\newcommand{\cop}{CoPa}
\newcommand{\qdcm}{QdCM}
\newcommand{\ap}{{\tt AP}}
\newcommand{\model}{\calM}
\newcommand{\peval}{\calV}

\newcommand{\eqsign}{\rightleftharpoons}

\newcommand{\copbis}{\eqsign_{\mathtt{\cop}}}

\newcommand{\bis}{\mathrel{\,%
  \raisebox{.3ex}{$\underline{\makebox[.7em]{$\leftrightarrow$}}$}\,}}

\newcommand{\irl}{{\tt IRL}}

\newcommand{\icrl}{{\tt ICRL}}
\newcommand{\slcs}{{\tt SLCS}}

\newcommand{\voxlogica}{{\tt VoxLogicA}}
\newcommand{\imgql}{{\tt ImgQL}}
\newcommand{\graphlogica}{{\tt GraphLogicA}}
\newcommand{\grql}{{\tt GrQL}}
\newcommand{\mCRLT}{{\tt mCRL2}}
\newcommand{\voxmin}{{\tt VoxMin}}
\newcommand{\voxminx}{{\tt VoxMinX}}
\newcommand{\form}{\Phi}

\newcommand{\lneg}{\neg}
\newcommand{\liand}{\bigwedge}

\newcommand{\lstothru}{\vec{\zeta}}
\newcommand{\lsfromthru}{\cev{\zeta}}

\newcommand{\icrleq}{\simeq_{\icrl}}

\newcommand{\closedefi}{}

\newcommand{\lts}{LTS}
\newcommand{\LTS}{\calS}
\newcommand{\ltsT}{\mathtt{LTS}_{\mkern1mu \mathrm{gen}}}
\newcommand{\ltsTS}{\mathtt{LTS_{\mkern1mu \mathrm{sym}}}}
\newcommand{\act}{\mathtt{Act}}
\renewcommand{\act}{\mathcal{A}}
\newcommand{\Mtrans}[1]{\stackrel{#1}{\longrightarrow}}
\newcommand{\pijl}[1]{\Mtrans{#1}}
\newcommand{\eD}{\mathtt{dr}}
\newcommand{\eC}{\mathtt{cv}}
\newcommand{\ch}{\mathtt{ch}}

\newcommand{\sem}[1]{[\![#1]\!]}

\newcommand{\bisb}{\bis_b}
\newcommand{\cevp}[1]{\cev{#1}\mkern1mu{}'\mkern-2mu}

\newcommand{\cevx}{\cev{x}}
\newcommand{\cevxp}{\cevp{x}}
\newcommand{\cevy}{\cev{y}}

\newcommand{\depth}[1]{|\mkern1mu{#1}\mkern1mu|}
\newcommand{\eqV}{=_{\mkern1mu \peval}}

\newcommand{\ie}{\emph{i.e.}}

\newcommand{\mopB}{\mathop{\mkern-2mu \textit{B}}}
\newcommand{\neqV}{\neq_{\mkern1mu \peval}}
\newcommand{\pijlch}{\pijl{\ch}}
\newcommand{\pijlp}{\pijl{p}}
\newcommand{\pijlster}[1]{\mathbin{\stackrel{#1}{\longrightarrow}\!\!{}^\ast}}
\newcommand{\pijlstertau}{\pijlster{\tau}}
\newcommand{\pijltau}{\pijl{\tau}}
\newcommand{\vecp}[1]{\vec{#1}\mkern2mu{}'\mkern-2mu}

\newcommand{\vecx}{\vec{x}}
\newcommand{\vecxp}{\vecp{x}}
\newcommand{\vecy}{\vec{y}}
\newcommand{\vecyp}{\vecp{y}}
\newcommand{\xbar}{\bar{x}}
\newcommand{\xbarp}{\bar{x}'}

\makeatletter
\def\mathcenterto#1#2{\mathclap{\phantom{#1}\mathclap{#2}}\phantom{#1}}
\let\old@widetilde\widetilde
\def\widetildeto#1#2{\mathcenterto{#2}{\old@widetilde{\mathcenterto{#1}{#2\,}}}}
\let\old@widehat\widehat
\def\widehatto#1#2{\mathcenterto{#2}{\old@widehat{\mathcenterto{#1}{#2\,}}}}
\makeatother

\newcommand{\wtC}{\widetildeto{X}{\mathcal{C}} \mkern-2mu}
\newcommand{\wtM}{\widetildeto{X}{\mathcal{M}}}
\newcommand{\wtV}{\mkern1mu \widetildeto{X}{\peval} \mkern-2mu}
\newcommand{\wtX}{\widetilde{X}}
\newcommand{\wtpi}{\widetilde{\pi}}

%% file: Introduction.tex
\section{Introduction}
\label{sec:introduction}

Spatial model checking consists in the automatic verification of
properties, expressed in a suitable spatial logic, on each point of a
suitable spatial model.  
In~\cite{Ci+14} the Spatial Logic for Closure Spaces (\slcs{}) was introduced and further developed in~\cite{Ci+16}.
Closure spaces, or \v{C}ech closure spaces~\cite{Cec66}, are a generalisation of topological spaces suitable to model many kinds of spatial objects, ranging from topological objects in continuous spaces, such as Euclidean spaces, to discrete spatial objects, such as general and regular graphs. 
The latter are particularly useful to represent digital images. 
Closure spaces (\cs{}) and the sub-class of quasi-discrete closure spaces, \qdcs{s} for short, form a convenient theoretical framework because of their generality and relative simplicity. 
A practical demonstration of this is the tool \voxlogica{}, a recently developed spatial model checker that can efficiently check \slcs{} properties of large digital images represented as \emph{symmetric} quasi-discrete closure models---\qdcm{s}, i.e.\ models with \qdcs{s} as underlying spaces~\cite{Be+19,Be+19a,Be+21}.

Spatial and spatio-temporal model checking have been successfully employed, in the past years, in a variety of application areas, ranging from Collective Adaptive Systems \cite{Ci+16a,Ci+18} to signals~\cite{Ne+18}, images \cite{Ci+16,Ha+15,Ba+20} and polyhedra~\cite{Be+21a}, just to mention a few. 
These  methods for spatial analysis are enjoying an increasing interest in computer science and beyond, also in unexpected  domains such as medical imaging \cite{Be+19,Be+21}. 
Medical images are obtained from diagnostic instruments such as magnetic resonance images (MRI), computer tomography scans, positron emission tomography, or dermoscopic images. 
Such images usually consist of millions of pixels, in 2D, or voxels (volumetric pixels) in 3D images.

\begin{figure}%[h!]
\begin{center}
\includegraphics[height=0.15\textheight]{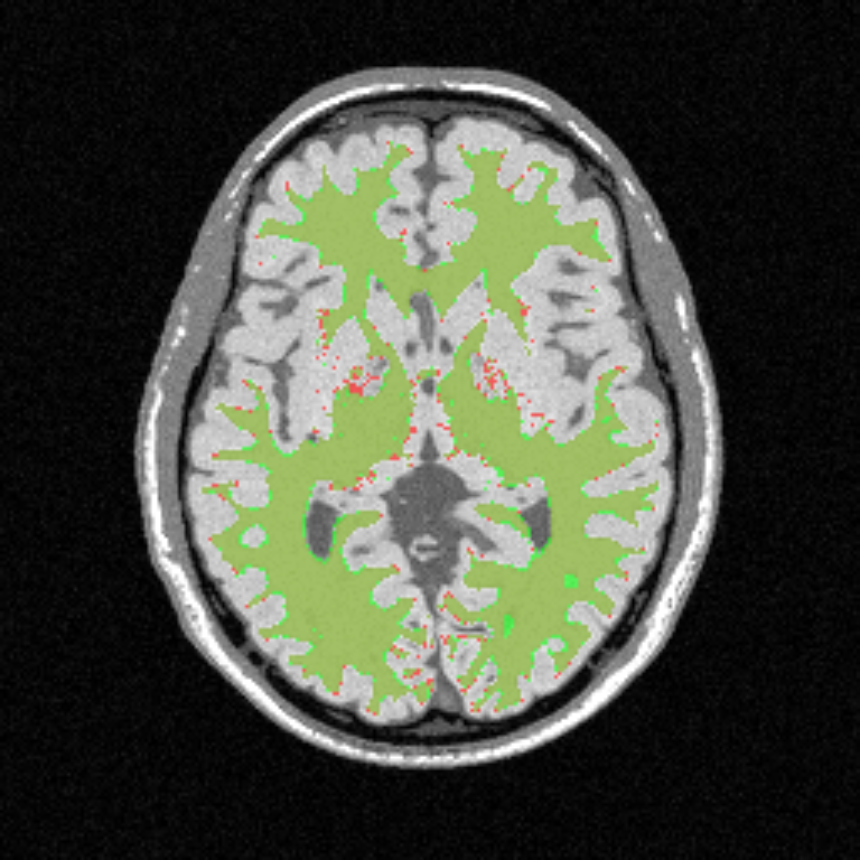}\quad
\includegraphics[height=0.15\textheight]{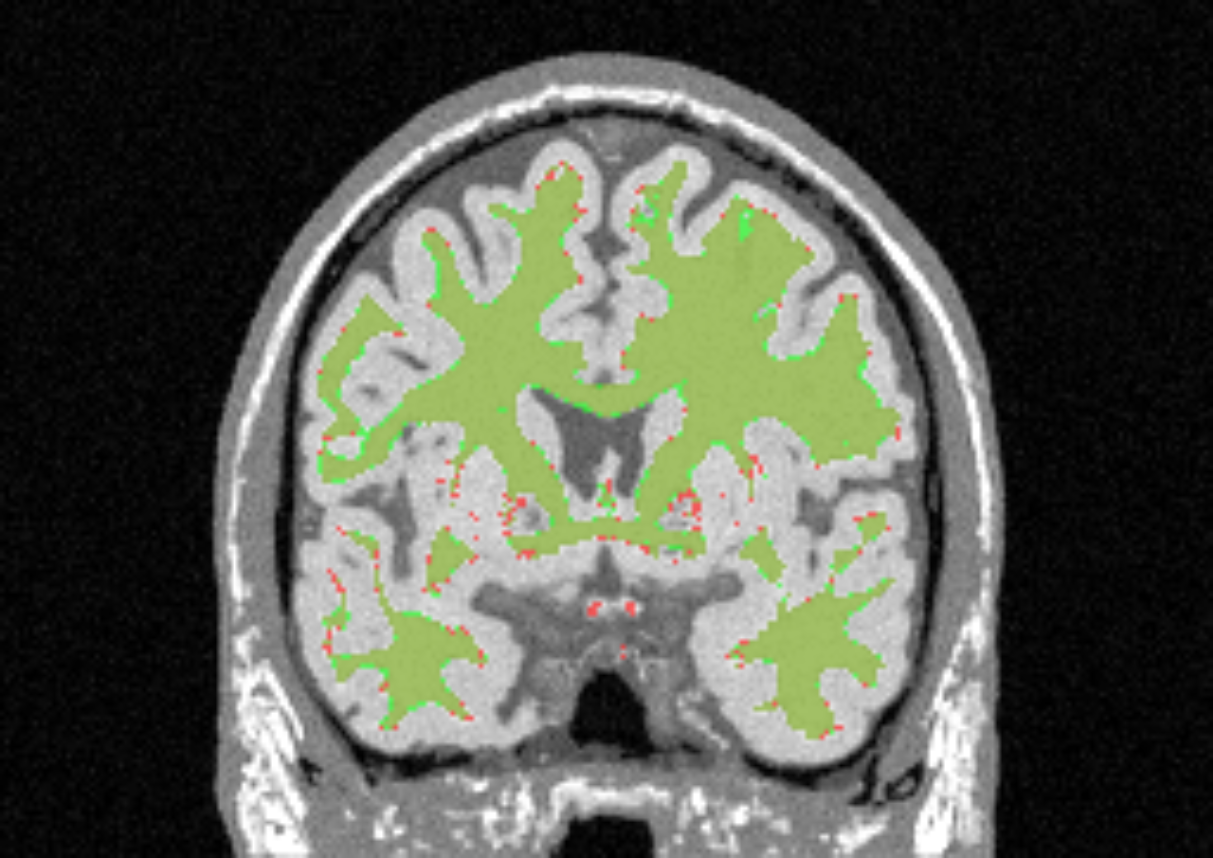}\quad
\includegraphics[height=0.15\textheight]{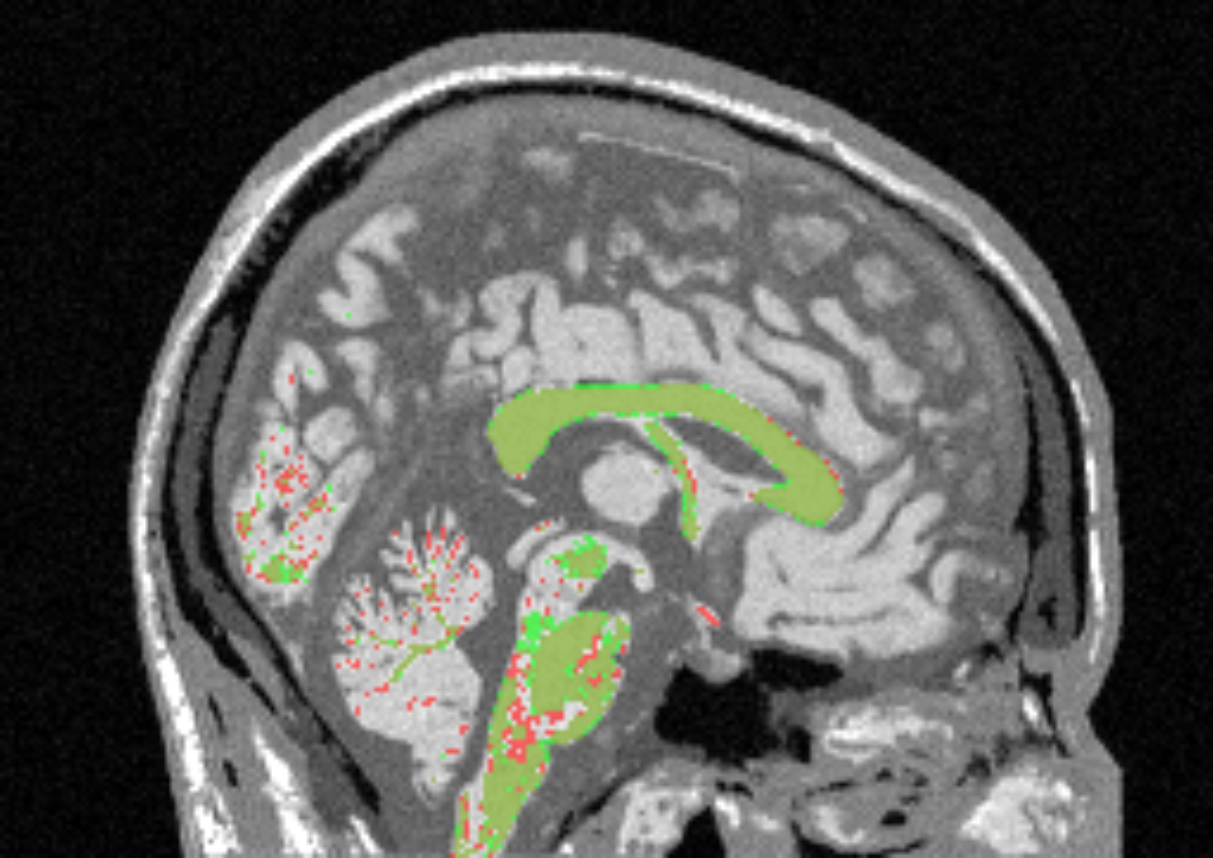}
\end{center}
\vspace*{-0.5\baselineskip}
\caption{Cross section of a dataset element of BrainWeb~\cite{Au+06} pat04 MRI at slice $(x,y,z)=(129,147,78)$, (axial view left, coronal view middle, sagittal view right): \voxlogica{} analysis of the segmentation of white matter, shown as a green overlay on top of a red overlay representing the ground truth.} \label{fig:brain}
\end{figure}

For example, in our earlier work~\cite{Be+19a}, the 3D MRI image of a healthy brain shown in Figure~\ref{fig:brain} consists of circa 12M~voxels (i.e.\ $256 \times 256 \times 181$) requiring approximately 10 seconds to analyse using \voxlogica{} on a desktop computer.\footnote{In that work we used an Intel Core I9~9900K processor (with 8~cores) and 32GB of~RAM\@.} 
\voxlogica{} checks a logical specification for \emph{every} point in the model exploiting parallel execution, memoization, and state-of-the-art imaging libraries~\cite{Be+19}. 
\voxlogica{}  internally creates a \qdcm{} of the pixel-based, or voxel-based, image that it is given as input. 
Such models take the form of a symmetric regular 8-adjacency graph, in which nodes are representing pixels (voxels) and edges between nodes are representing the adjacency between pixels represented by the nodes. 
Both adjacency in north-south direction and in east-west direction are considered, as well as diagonal adjacency. 
Figure~\ref{fig:img} shows a simple pixel-based image and Figure~\ref{fig:QdCMsym} its related \qdcm{} (except for the self-loops, that have been omitted to avoid clutter in the figure).

\begin{figure}
\centering
\subfloat[][]
{
  \input{pict/img_grid02.tex}
  \label{fig:img}
}\qquad \quad\quad\quad
\centering
\subfloat[][]
{
  \input{pict/grid02.tex}
  \label{fig:QdCMsym}
}\quad
\caption{Figure~\ref{fig:img}: A pixel based image with green and red pixels. Figure~\ref{fig:QdCMsym}: A finite, symmetric \qdcm{} in the shape of a regular graph with nodes representing pixels and edges representing the 8-adjacency relation between the pixels of the image in Figure~\ref{fig:img}.}\label{fig:VLExample}
\end{figure}

A detailed description of \voxlogica{}, its application to the contouring of high grade glioblastoma, a malignant type of brain tumour tissue, and to the identification of grey and white matter in the brain can be found in~\cite{BelmonteCM25}. 
In the same work also a hybrid contouring method has been proposed, combining the symbolic spatial model checking method with {\sf nnU-Net}~\cite{IsenseeJKPM21,IsenseeWUBRMJ24}, a sub-symbolic deep learning method. 
Such a hybrid approach may provide a human explainable method for the contouring while gaining in precision enhancing the results with the deep learning approach. 

A way to increase the time and space efficiency of spatial model checking is to
exploit suitable model minimisation algorithms based on \emph{spatial}
bisimilarity. 
For that purpose several spatial bisimilarities have been proposed in~\cite{CLMdV22,CLMdV25}. 
In particular, \cop-bisimilarity, based on a notion of ``path-compatibility'' is promising. 
The notion of path compatibility essentially requires that two paths, in order to be compatible, have to be both composed of a (non-empty) sequence of an equal number of non-empty adjacent ``zones'', such that each point in one zone of one path must be related, by the bisimulation relation, to every point in the corresponding zone of the other path (see the illustration in Figure~\ref{fig:Zones}).

In~\cite{CLMdV22,CLMdV25}, a logical characterisation of \cop-bisimilarity has
been given. 
More precisely, Infinitary Compatible Reachability Logic (\icrl) has been defined that is a modal logic with infinitary conjunction and two modalities, $\lstothru$ and~$\lsfromthru$, expressing conditional forward and backward reachability, respectively.  
Given two \icrl{} formulas $\form_1$ and~$\form_2$, a point~$x$ satisfies $\lstothru \form_1 [\form_2]$ if~$x$ satisfies~$\form_1$, or~$x$ satisfies~$\form_2$ and there is a path from~$x$ to a point~$y$ satisfying~$\form_1$ where all the points on the path between~$x$ and~$y$ satisfy~$\form_2$. 
Similarly for $\lsfromthru \form_1 [\form_2]$, which is satisfied by~$x$ if~$x$ satisfies $\form_1$, or ~$x$ satisfies~$\form_2$ and there is a path from a point~$y$ satisfying~$\form_1$, to~$x$, where all the points on the path between~$y$ and~$x$ satisfy~$\form_2$.\footnote{Note that,
  different from the context of classical temporal logics, in the
  context of space, and in particular when dealing with notions of
  directionality (e.g. one way roads, public area gates), it is
  important to be able to distinguish between the concept of
  “reaching” and that of “being reached”. The interested reader is
  referred to~\cite{CLMdV22,CLMdV25} for a discussion on the issue.}

Building on our previous work~\cite{CianciaGLMV23}, this paper includes two original contributions, one of a more theoretical nature and the other more practical in kind.

\paragraph{Theoretical contribution.} 
The paper introduces an encoding of finite Closure Models (\cm{}s), a sub-class of \qdcm{}s, into Labelled Transition Systems (LTS) that preserves \cop-bisimilarity. 
More precisely, two points in the input \cm{} are \cop-bisimilar if
and only if the states they are mapped to by the encoding are branching bisimilar \cite{GlW96,Gr+17,JGKW20,GJ25}. 
Thus, given a finite \cm, the encoding makes it possible to effectively compute the minimal model with respect to \cop-bisimilarity via the composition of the encoding and a very efficient minimisation algorithm for branching bisimilarity, proposed in~\cite{Gr+17,JGKW20,GJ25} and implemented as part of the \mCRLT\ tool set~\cite{Bu+19}. 
In addition, detailed correctness proofs, including those concerning the encoding, are provided.

\paragraph{Practical contribution.}
To validate the approach, we apply it to digital images, which are a
special case of (symmetric) finite closure models. A digital image is usually composed of a large number of pixels (i.e.\ points in the closure space representing the image) and, therefore, digital images represent  an interesting benchmark for the minimisation procedure. 
In particular, a prototype experimental implementation of the encoding has been developed which, in turn, is part of a more complex toolchain, \voxminx, used for the analysis of digital images. 
The toolchain comprises the implementation of the encoding developed in the theoretical part and produces a minimal model, exploiting the branching bisimulation minimisation procedure provided by the tool set \mCRLT~\cite{Bu+19}. 
The resulting minimal model is suitable for model checking with \graphlogica, a spatial model checker for the finitary version of \icrl, that takes decorated graphs as input, representing general finite closure models.  
The minimal model produced by \voxminx\ maintains the relationship between the equivalence classes, that are represented by states of the minimal model, and the respective sets of pixels of the original digital image. 
In this way, the spatial model checking results can directly be visualised on the original image, for example by highlighting the pixels that satisfy a spatial property of interest. 
\voxminx\ extends the preliminary toolchain that was introduced in our previous work~\cite{CianciaGLMV23}. 
The \voxminx\ toolchain is evaluated on a benchmark of images at various resolutions. 
In this evaluation the model checking performance of \voxminx\ is compared to that of  \voxlogica,  since \voxlogica\ is an optimised spatial model checker for non-minimised digital images, and therefore it forms a state-of-the-art basis for comparison. 
The evaluation shows that a considerable speed-up can be obtained in model checking time, especially when comparing the times for larger images, suggesting interesting directions for future research and applications.

\paragraph{Related work.} Qualitative reasoning about spatial entities~\cite{CohnR2008} has been, and still is, a very active area of research in which the theory of topology and closure spaces play a important role. 
Prominent examples of that area are the region connection calculi, such as~RCC8D\@.
An embedding of the latter in the collective variant of \slcs\ was presented in~\cite{Ci+19b}.
Our work is mainly inspired by spatial logics (see~\cite{HBSL07} for
an extensive overview), with seminal work dating back to Tarski and McKinsey in the forties of the previous century.
The work on \emph{spatial model checking} for logics with reachability originated in \cite{Ci+16}, which includes a comparison to the work of Aiello on spatial \emph{until} operators (see e.g.~\cite{A02}). 
In~\cite{A03},
Aiello envisaged practical applications of topological logics with \emph{until} to minimisation of images. The present paper builds on and extends that vision. 
Bisimilarity for spatial logics with reachability is a relatively new subject. 
In~\cite{Li+20}, a
bisimulation relation that is correct with respect to \slcs{} has been presented. 
Such a definition has not yet been proved complete and is aimed at characterising the logic including the \emph{near} operator, therefore, not quotienting up-to reachability, as done in the present paper. 
The work in~\cite{Be+21a,Ci+23c,Be+24} and that in~\cite{LoQ23} introduce
bisimulation relations that characterise spatial logics with reachability in polyhedral models and in simplicial complexes, respectively. 
It will be interesting future work to apply the minimisation techniques we present also to such relevant classes of models. 
First results in that direction can be found in~\cite{BezhanishviliBCGJLMV26} where weak simplicial bisimilarity and related minimisation procedures are presented for use in polyhedral model checking.

In the Computer Science literature, other kinds of spatial logics have been proposed that typically describe situations in which modal operators are interpreted \emph{syntactically} against the structure of agents in a process calculus. 
We refer to~\cite{CaG00,CaCC03} for some classical examples. 
Along the same lines, a recent example is given in~\cite{TPGN15}, concerning model checking of security aspects in cyber-physical systems, in a spatial context based on the idea of bigraphical reactive systems introduced by Milner~\cite{Mil09}. 
A bigraph consists of two graphs: A place graph, i.e. a forest defined over a set of nodes which is intended to represent entities and their locality in terms of a containment structure, and a link graph, a hypergraph composed over the same set of nodes representing arbitrary linking among those entities. 
The \qdcs\ models that are the topic of the present paper, instead, address space from a topological point of view rather than as a containment structure for spatial entities.

There is also active research on the improvement of the performance of branching bisimularity minimisation algorithms. A recent contribution can be found in~\cite{MaL26}.

The structure of the paper is as follows. Section~\ref{sec:Preliminaries} recalls relevant  concepts and introduces notation.  Section~\ref{sec:COPAbisimilarity} recalls \cop-bisimilarity for \qdcm{s}. In Section~\ref{sec:translation} the encoding of 
finite \qdcm{s} into \lts{s} is presented, together with the correctness results. 
Section~\ref{sec:FisStudy} describes a feasibility study and experimental evaluation of \voxminx\  applying it to three families of representative benchmark examples. 
Appendix~\ref{apx:running} shows a small running example for the transformations and minimisation performed by the \voxminx\ toolchain.

%% file: pict/img_grid02.tex
%% file grid02.tex used in CoPa section %%

\scriptsize
\scalebox{0.5}{
\begin{tikzpicture}[baseline=1,
  <->, >=stealth', semithick, shorten >=0.5pt, shorten <=0.5pt,
  every state/.style={
    %circle, minimum size=12pt, inner sep=0.5pt, draw},
    rectangle, minimum size=42pt, inner sep=0.5pt, draw},
  ]
  
  \node [state, fill=green!50] (x11) at (1.0,1.0) {\phantom{21}};
  \node [state, fill=green!50] (x12) at (1.0,2.5) {\phantom{16}};
  \node [state, fill=green!50] (x13) at (1.0,4.0) {\phantom{11}};
  \node [state, fill=green!50] (x14) at (1.0,5.5) {\phantom{6}};
  \node [state, fill=green!50] (x15) at (1.0,7.0) {\phantom{1}};

  \node [state, fill=green!50] (x21) at (2.5,1.0) {\phantom{22}};
  \node [state, fill=red!50  ] (x22) at (2.5,2.5) {\phantom{17}};
  \node [state, fill=red!50  ] (x23) at (2.5,4.0) {\phantom{12}};
  \node [state, fill=red!50  ] (x24) at (2.5,5.5) {\phantom{7}};
  \node [state, fill=green!50] (x25) at (2.5,7.0) {\phantom{2}};

  \node [state, fill=green!50] (x31) at (4.0,1.0) {\phantom{23}};
  \node [state, fill=red!50  ] (x32) at (4.0,2.5) {\phantom{18}};
  \node [state, fill=red!50  ] (x33) at (4.0,4.0) {\phantom{13}};
  \node [state, fill=red!50  ] (x34) at (4.0,5.5) {\phantom{8}};
  \node [state, fill=green!50] (x35) at (4.0,7.0) {\phantom{3}};

  \node [state, fill=green!50] (x41) at (5.5,1.0) {\phantom{24}};
  \node [state, fill=red!50  ] (x42) at (5.5,2.5) {\phantom{19}};
  \node [state, fill=red!50  ] (x43) at (5.5,4.0) {\phantom{14}};
  \node [state, fill=red!50  ] (x44) at (5.5,5.5) {\phantom{9}};
  \node [state, fill=green!50] (x45) at (5.5,7.0) {\phantom{4}};

  \node [state, fill=green!50] (x51) at (7.0,1.0) {\phantom{25}};
  \node [state, fill=green!50] (x52) at (7.0,2.5) {\phantom{20}};
  \node [state, fill=green!50] (x53) at (7.0,4.0) {\phantom{15}};
  \node [state, fill=green!50] (x54) at (7.0,5.5) {\phantom{10}};
  \node [state, fill=green!50] (x55) at (7.0,7.0) {\phantom{5}};

\end{tikzpicture}
}

%% file: pict/grid02.tex
%% file grid02.tex used in CoPa section %%

%\scriptsize

\scalebox{0.5}{
\begin{tikzpicture}[baseline=1,
  <->, >=stealth', semithick, shorten >=0.5pt, shorten <=0.5pt,
  every state/.style={
    circle, minimum size=20pt, inner sep=0.5pt, draw},
  ]
  
  \node [state, fill=green!50] (x11) at (1.0,1.0) {21};
  \node [state, fill=green!50] (x12) at (1.0,2.5) {16};
  \node [state, fill=green!50] (x13) at (1.0,4.0) {11};
  \node [state, fill=green!50] (x14) at (1.0,5.5) {6};
  \node [state, fill=green!50] (x15) at (1.0,7.0) {1};

  \node [state, fill=green!50] (x21) at (2.5,1.0) {22};
  \node [state, fill=red!50  ] (x22) at (2.5,2.5) {17};
  \node [state, fill=red!50  ] (x23) at (2.5,4.0) {12};
  \node [state, fill=red!50  ] (x24) at (2.5,5.5) {7};
  \node [state, fill=green!50] (x25) at (2.5,7.0) {2};

  \node [state, fill=green!50] (x31) at (4.0,1.0) {23};
  \node [state, fill=red!50  ] (x32) at (4.0,2.5) {18};
  \node [state, fill=red!50  ] (x33) at (4.0,4.0) {13};
  \node [state, fill=red!50  ] (x34) at (4.0,5.5) {8};
  \node [state, fill=green!50] (x35) at (4.0,7.0) {3};

  \node [state, fill=green!50] (x41) at (5.5,1.0) {24};
  \node [state, fill=red!50  ] (x42) at (5.5,2.5) {19};
  \node [state, fill=red!50  ] (x43) at (5.5,4.0) {14};
  \node [state, fill=red!50  ] (x44) at (5.5,5.5) {9};
  \node [state, fill=green!50] (x45) at (5.5,7.0) {4};

  \node [state, fill=green!50] (x51) at (7.0,1.0) {25};
  \node [state, fill=green!50] (x52) at (7.0,2.5) {20};
  \node [state, fill=green!50] (x53) at (7.0,4.0) {15};
  \node [state, fill=green!50] (x54) at (7.0,5.5) {10};
  \node [state, fill=green!50] (x55) at (7.0,7.0) {5};

  %% south-north %%
  \draw (x11) edge (x12) ;  \draw (x12) edge (x13) ;
  \draw (x13) edge (x14) ;  \draw (x14) edge (x15) ;

  \draw (x21) edge (x22) ;  \draw (x22) edge (x23) ;
  \draw (x23) edge (x24) ;  \draw (x24) edge (x25) ;

  \draw (x31) edge (x32) ;  \draw (x32) edge (x33) ;
  \draw (x33) edge (x34) ;  \draw (x34) edge (x35) ;

  \draw (x41) edge (x42) ;  \draw (x42) edge (x43) ;
  \draw (x43) edge (x44) ;  \draw (x44) edge (x45) ;

  \draw (x51) edge (x52) ;  \draw (x52) edge (x53) ;
  \draw (x53) edge (x54) ;  \draw (x54) edge (x55) ;

  %% south-west north-east %%
  \draw (x14) edge (x25) ;

  \draw (x13) edge (x24) ;
  \draw (x24) edge (x35) ;

  \draw (x12) edge (x23) ;
  \draw (x23) edge (x34) ;
  \draw (x34) edge (x45) ;

  \draw (x11) edge (x22) ;
  \draw (x22) edge (x33) ;
  \draw (x33) edge (x44) ;
  \draw (x44) edge (x55) ;

  \draw (x21) edge (x32) ;
  \draw (x32) edge (x43) ;
  \draw (x43) edge (x54) ;

  \draw (x31) edge (x42) ;
  \draw (x42) edge (x53) ;

  \draw (x41) edge (x52) ;

  %% west-east %%
  \draw (x11) edge (x21) ; \draw (x21) edge (x31) ;
  \draw (x31) edge (x41) ; \draw (x41) edge (x51) ;

  \draw (x12) edge (x22) ; \draw (x22) edge (x32) ;
  \draw (x32) edge (x42) ; \draw (x42) edge (x52) ;

  \draw (x13) edge (x23) ; \draw (x23) edge (x33) ;
  \draw (x33) edge (x43) ; \draw (x43) edge (x53) ;

  \draw (x14) edge (x24) ; \draw (x24) edge (x34) ;
  \draw (x34) edge (x44) ; \draw (x44) edge (x54) ;

  \draw (x15) edge (x25) ; \draw (x25) edge (x35) ;
  \draw (x35) edge (x45) ; \draw (x45) edge (x55) ;

  %% north-west south-east %%
  \draw (x12) edge (x21) ;

  \draw (x13) edge (x22) ;
  \draw (x22) edge (x31) ;

  \draw (x14) edge (x23) ;
  \draw (x23) edge (x32) ;
  \draw (x32) edge (x41) ;

  \draw (x15) edge (x24) ;
  \draw (x24) edge (x33) ;
  \draw (x33) edge (x42) ;
  \draw (x42) edge (x51) ;

  \draw (x25) edge (x34) ;
  \draw (x34) edge (x43) ;
  \draw (x43) edge (x52) ;

  \draw (x35) edge (x44) ;
  \draw (x44) edge (x53) ;

  \draw (x45) edge (x54) ;
\end{tikzpicture}
}

%% file: Preliminaries.tex
 \section{Preliminaries}\label{sec:Preliminaries}

We first introduce some relevant concepts and notation, in particular recalling \lts{s}, branching bisimilarity~\cite{GlW96,Gr+17,JGKW20}, (quasi-discrete) closure spaces, and closure models and paths therein.

Given a set $X$, 
$\pws(X)$ denotes the powerset of $X$.
The set of natural numbers is denoted by $\nats$. For $n,m \in \nats$ we often use the interval notation 
$[m,n]$ denoting the set $\ZET{i \in \nats}{m \leqslant i \leqslant n}$,
$[m,n)$ denoting the set $\ZET{i \in \nats}{m \leqslant i < n}$, and similarly for
$(m,n]$ and~$(m,n)$.
 
In the sequel, branching bisimilarity~\cite{GlW96,Gr+17,JGKW20} of states of \lts{s} plays a central role. Below we recall the relevant definitions.

\begin{defi}[Labelled Transition System - \lts{}]
  \label{def:lts}
  A \emph{Labelled Transition System}, \lts{} for short, is a tuple
  $(S, \act, {\rightarrow})$ where $S$ and~$\act$ are non-empty sets of
  \emph{states}, and \emph{actions}, respectively, and relation
  ${\rightarrow} \, \subseteq \, S \times \act \times S$ is the \emph{transition relation}.  \closedefi
\end{defi}

\noindent
As usual, we distinguish an action $\tau \in \act$ that models a ``silent move'' in the~\lts.
Moreover, we call the elements of the relation~$\rightarrow$ \emph{transitions}, and we write $s \pijl{\alpha} s'$ whenever $(s,\alpha,s')\in {\rightarrow}$.
  A \emph{computation} in the~\lts{} is an alternating sequence  $s_0 \pijl{\alpha_1} s_1 \, \cdots \, s_{n-1} \pijl{\alpha_n} s_n$   of states and actions where~$n \geqslant 0$ and $s_{i{-}1}   \pijl{\alpha_i} s_i$ for $i=1, \ldots, n$.
  We have occasion to write $s \pijlstertau s'$ if with respect to   the above   situation we have $s = s_0$, $\alpha_i = \tau$ for $i=1, \ldots, n$, and~$s' = s_n$.

\begin{defi}[Branching bisimilarity -- $\bis_b$]\label{def:dbs}
  Given an \lts{} $\LTS ={(S,\act,{\rightarrow})}$, a symmetric
  relation $B \subseteq S \times S$ is a \emph{branching bisimulation
    for $\LTS$} iff, for $s,t,s' \in S$ and $\alpha \in \act$,
  whenever $s \mopB t$ and $s \pijl{\alpha} s'$, it holds that
  \begin{enumerate} [label=(\roman*)]
  \item $s' \mopB t$ and~$\alpha = \tau$, or 
  \item $s \mopB \bar{t}$,
  $s' \mopB t'$ and $t \pijlstertau \bar{t}$,
  $\bar{t} \pijl{\alpha} t'$ for some $\bar{t},t' \in S$
  \end{enumerate}
  Two states $s, t \in S$ are called \emph{branching bisimilar}
  in~$\LTS$ if $s \mopB t$ for some branching bisimulation~$B$
  for~$\LTS$. Notation,~$s \,\bis_b^{\mkern1mu \LTS}\, t$. \closedefi
\end{defi}

\noindent
From now on, for readability, we omit the superscript 
${\LTS}$ in $\bis_b^{\mkern1mu \LTS}$, when this does not cause confusion.

Our  framework  for modelling space is based on the notion of a \emph{\v{C}ech closure space}~\cite{Cec66}, \cs{} for short, that provides a convenient common framework for the study of several different kinds of spatial models, including models of both discrete and continuous space~\cite{SmW07}. We briefly recall definitions and results on
\cs{s}, that are relevant for this paper --- most of which are borrowed from~\cite{Gal03} (see also~\cite{CLMdV22,CLMdV25,Ci+16}).

\begin{defi}[Closure Space -- \cs]\label{def:ClosureSpace}
  A \emph{closure space} is a pair $(X,\closure)$
  where $X$ is a 
  set (of \emph{points}) and \emph{$\closure : \pws(X) \to \pws(X)$} is the \emph{closure operator}, i.e., a function
     satisfying the
  following axioms:
  \begin{enumerate} [label=(\roman*)]
\item $\closure(\emptyset)=\emptyset$,
\item $A \subseteq \closure(A)$ for all  $A \subseteq X$, and
\item $\closure(A_1 \cup A_2) = \closure(A_1) \cup \closure(A_2)$ for
  all $A_1,A_2\subseteq X$.
  \end{enumerate} 
\closedefi
\end{defi}

\noindent
It is worth pointing out that \cs{s} are a generalisation of
topological spaces. In fact, the latter coincide with \cs{s} that
satisfy the \emph{idempotence} axiom, i.e.,
$\closure(\closure(A)) = \closure(A)$ for all $A \subseteq X$.

\begin{defi}[Quasi-discrete closure space -- \qdcs]
  \label{def:QDClosureSpace}
  A \emph{quasi-discrete closure space} is a \cs{} $(X,\closure)$ such
  that 
for each $A \subseteq X$ it holds that
    $\closure(A) = \bigcup_{x\in A}\closure(\SET{x})$.
    \closedefi
\end{defi}

\noindent
Thus, the closure operator in a \qdcs\ is determined by its
  value for the singleton. For brevity, we have occasion to write
  $\closure(x)$ instead of~$\closure(\{x\})$.

Given a relation $R \subseteq X \times X$, define the function
$\closure_{R}: \pws(X) \to \pws(X)$ as follows: for all~$A \subseteq X$,
we put $ \closure_R(A) = A \cup \ZET{x \in X}{\exists
  \mkern1mu a \in A : \,a \mathop{\mbox{$R$}} x}.  $ It is easy to see that, for any~$R$,
 $\closure_{R}$ satisfies all the axioms of
Definition~\ref{def:ClosureSpace} and so $(X, \closure_{R})$ is a~\cs. 
An example of the result of applying the closure operator~$\closure_R$ induced by a relation~$R$ to a set~$A$ is shown in Figure~\ref{fig:CloExample}.

\begin{figure}
\centering
\subfloat[][]
{
  \input{pict/PhiEV}
  \label{fig:Phi}
}\qquad \quad\quad\quad
\centering
\subfloat[][]
{
  \input{pict/NPhiEV}
  \label{fig:NPhi}
}\quad
\caption{Figure~\ref{fig:Phi}: A finite \qdcs{} $(X,\closure_R)$. The
  arrows represent the relation $R$ underlying~$\closure_R$. The
  points of the set $A \subseteq X$ are shown in white,
  remaining points are shown in black. Figure~\ref{fig:NPhi}:
 Points in~$\closure_R(A) \backslash A$
  are shown in grey.}
\label{fig:CloExample}
\end{figure}

The following theorem is a standard result in the theory of
\cs{s}~\cite{Gal03}.
\begin{thm}
  A \cs{} $(X, \closure)$ is quasi-discrete if and only if there is a
  relation $R \subseteq X \times X$ such that $\closure =
  \closure_{R}$.\qed
\end{thm}

\noindent
The above theorem implies that any graph coincides with a \qdcs{} where the relation underlying the closure operator is exactly the edge relation of the graph.
We prefer to treat graphs as
\qdcs{s} since in this way we can formulate key definitions at the
level of closure spaces leading to a uniform treatment for graphs
and other kinds of models for space (e.g. topological spaces)~\cite{SmW07}.
Furthermore, if $X$~is finite, any closure space~$(X,\closure)$ is quasi-discrete:
  For~$A \subseteq X$ it holds that
  \begin{displaymath}
    \closure(A) =
    \closure( \textstyle \bigcup \ZET{\{x\}}{x \in A} ) =
    \textstyle \bigcup \ZET{\closure(\{x\})}{x \in A}
  \end{displaymath}
  by axiom~(iii) of Definition~\ref{def:ClosureSpace} and finiteness of~$A$.
In the sequel, we consider only \emph{finite} \cs{s},
hence only \qdcs{s}.
and often refrain
from explicitly writing the subscript~$R$ in~$\closure_R$, when this
does not cause confusion. Finally, we say that $(X,\closure)$ is
a \emph{symmetric}  \qdcs\ if $\closure = \closure_R$ for a symmetric relation~$R$. 
In such a case, it holds that $y \in \closure(\{x\})
  \Leftrightarrow x \in \closure(\{y\})$ for $x,y \in X$.

In the context of the present paper, \emph{paths} over \cs{s} play an
important role. Following the tradition in topology, in the theory of~\cs{s}
paths are defined as continuous functions from an appropriate index space to the \cs{} at hand.
For finite \cs{s}, it is sufficient to consider 
finite paths.

\begin{defi}[Finite path]
A (finite) path in a  \cs{} $(X,\closure)$ is a 
function $\pi : [0,\ell] \to X$, for some $\ell \in \nats$,
such that $\pi(i{+}1) \in \closure(\SET{\pi(i)})$ for $i=0,\ldots,\ell{-}1$. 
We call~$\ell$ the \emph{length} of~$\pi$ and we denote it by~$\pthlen(\pi)$.
\closedefi
\end{defi}

\noindent
In the sequel, we will say that a path $\pi : [0,\ell] \to X$ in a finite \cs{} $(X,\calC)$ 
is a {\em path from} $x \in X$ if $\pi(0)=x$ and it is a {\em path to} $x$ if
$\pi(\ell)=x$. Furthermore, we will use the notation $(x_i)_{i=0}^{\ell}$ for
the path $\pi : [0,\ell] \to X$ such that $\pi(i)=x_i$ for all $i \in [0,\ell]$.

\begin{rem}\label{rem:OnPaths}
  It is worth pointing out that the notion of path in a \qdcs{} is
  similar to that of a path in a graph or of a computation in an \lts, but
  it is \emph{not} the same. In particular, due to axiom~(ii) of
  closure operator~$\closure$ and the requirement
  $\pi(i{+}1) \in \closure({\pi(i)})$, paths in \cs{s} allow
  \emph{stuttering};
  in other words, for \qdcs\ $(X,\closure)$, $x\in X$,
  and path~$\pi$, it may happen that $\pi(i)=\pi(i{+}1)=x$, for
  $i<\pi(\pthlen(\pi))$ \emph{even when} $(x,x)$~is \emph{not} an
  element of the relation $R\subseteq X \times X$
  underlying~$\closure$.
  This is different for a path $\ldots n_1n_2 \ldots$ in a
  graph~$(N,E)$, where in order for nodes $n_1$ and~$n_2$ in~$N$ to be
  adjacent, it is required that $(n_1,n_2)$~is an element of the edge
  relation~$E$.  A similar issue arises when comparing paths in
  \qdcs{s} with traces in \lts{s}. In fact, for \lts{}
  $(S,\act,{\rightarrow})$, two states $s_1$ and~$s_2$ can be adjacent
  in a computation $\cdots s_1\pijl{\alpha}s_2 \cdots$ \emph{only} if
  $(s_1,\alpha,s_2) \in {\rightarrow}$, and this holds also if~$s_2=s_1$.
\end{rem}

\noindent
We assume a set~$\ap$ of \emph{atomic proposition letters} is given and 
 introduce the notion of closure {\em model} (\cm{} for short).
\begin{defi}[Closure model -- \cm{}]
  \label{def:ClosureModel}
  A \emph{closure model} is a tuple
  $\model = (X,\closure, \peval)$, with $(X,\closure)$ a \cs, and
  $\peval: \ap \to \pws(X)$ the  
  \emph{valuation function},
  assigning to each $p\in \ap$ the set of points where $p$
  holds.
  \closedefi
\end{defi}

\noindent
All  definitions for \cs{s} also apply to \cm{s};
thus, a \emph{quasi-discrete closure model} (\qdcm{} for short) is a
\cm{} $\model=(X, \closure,\peval)$ where $(X,\closure)$~is a
\qdcs.  Similarly, a {\em symmetric} \qdcm{} is a \qdcm{} $(X,\closure,\peval)$ where $(X,\closure)$~is a symmetric \qdcs. For a closure model $\model=(X,\closure,\peval)$ we may
write $x \in \model$ when~$x \in X$. Similarly, we speak of paths in~$\model$ meaning paths in $(X,\closure)$. 
Finaly, with respect to~$\model$, for $x,y \in X$, we have
  $x \eqV y$ if ${x \in \peval(p)} \Leftrightarrow {y \in \peval(p)}$
  for all~$p \in \ap$, and $x \neqV y$ if not $x \eqV y$.

In the sequel, for a logic~$\calL$, a formula $\form \in \calL$, and a
model $\model=(X,\closure,\peval)$ we let
$\sem{\form}^{\model}_{\calL}$ denote the set
$\ZET{x\in X}{\model,x \models_{\calL} \form}$ of all the points in
$\model$ that satisfy $\form$, where $\models_{\calL}$ is the
satisfaction relation for $\calL$. For the sake of readability, we
refrain from writing the subscript~$\calL$ when this does not
cause confusion.

%% file: pict/PhiEV.tex
%% file PhiEV

\scalebox{0.90}{%
  \begin{tikzpicture}[>=stealth,
    every state/.style = {minimum size=2.0mm, inner sep=0pt},
  scale=0.5]
  \node [state, fill=black!80] (a2) at (0.2,3) {} ;
  \node [state, fill=black!80] (a1) at (-0.1,1.7) {} ;
  \node [state, fill=black!1] (b2) at (1,2.0) {} ;
  \node [state, fill=black!80] (b1) at (1,0.3) {} ;
  \node [state, fill=black!80] (c5) at (1.4,4) {} ;
  \node [state, fill=black!1] (c4) at (1.8,3.2) {} ;
  \node [state, fill=black!1] (c3) at (2,2.2) {} ;
  \node [state, fill=black!1] (c2) at (2.1,1) {} ;
  \node [state, fill=black!80] (c1) at (1.9,0) {} ;
  \node [state, fill=black!1] (d4) at (3.5,3.1) {} ;
  \node [state, fill=black!1] (d3) at (3.1,2.1) {} ;
  \node [state, fill=black!1] (d2) at (3.3,0.9) {} ;
  \node [state, fill=black!80] (d1) at (3.6,0.1) {} ;
  \node [state, fill=black!80] (e3) at (4.5,3) {} ;
  \node [state, fill=black!80] (e2) at (4.1,1.8) {} ;
  \node [state, fill=black!80] (e1) at (4.7,1.1) {} ;
  \draw (a1) edge (b2) ;
  \draw (b2) edge (a2) ;
  \draw (b2) edge (c3) ;
  \draw (c4) edge (d4) ;
  \draw (c3) edge (c4) ;
  \draw (c3) edge (d3) ;
  \draw (c3) edge (c2) ;
  \draw (c2) edge (b1) ;
  \draw (c1) edge (c2) ;
  \draw (c1) edge (d1) ;
  \draw (d4) edge (e3) ;
  \draw (d3) edge (d4) ;
  \draw (d3) edge (e2) ;
  \draw (d2) edge (c3) ;
  \draw (d2) edge (d1) ;
  \draw (e2) edge (e1) ;
  \draw (e1) edge (d2) ;

\end{tikzpicture}
} %% scalebox

%% file: pict/NPhiEV.tex
%% file NPhiEV

\scalebox{0.90}{%
  \begin{tikzpicture}[>=stealth,
    every state/.style = {minimum size=2.0mm, inner sep=0pt},
  scale=0.5]
  \node [state, fill=black!30] (a2) at (0.2,3) {} ;
  \node [state, fill=black!80] (a1) at (-0.1,1.7) {} ;
  \node [state, fill=black!1] (b2) at (1,2.0) {} ;
  \node [state, fill=black!30] (b1) at (1,0.3) {} ;
  \node [state, fill=black!80] (c5) at (1.4,4) {} ;
  \node [state, fill=black!1] (c4) at (1.8,3.2) {} ;
  \node [state, fill=black!1] (c3) at (2,2.2) {} ;
  \node [state, fill=black!1] (c2) at (2.1,1) {} ;
  \node [state, fill=black!80] (c1) at (1.9,0) {} ;
  \node [state, fill=black!1] (d4) at (3.5,3.1) {} ;
  \node [state, fill=black!1] (d3) at (3.1,2.1) {} ;
  \node [state, fill=black!1] (d2) at (3.3,0.9) {} ;
  \node [state, fill=black!30] (d1) at (3.6,0.1) {} ;
  \node [state, fill=black!30] (e3) at (4.5,3) {} ;
  \node [state, fill=black!30] (e2) at (4.1,1.8) {} ;
  \node [state, fill=black!80] (e1) at (4.7,1.1) {} ;
  \draw (a1) edge (b2) ;
  \draw (b2) edge (a2) ;
  \draw (b2) edge (c3) ;
  \draw (c4) edge (d4) ;
  \draw (c3) edge (c4) ;
  \draw (c3) edge (d3) ;
  \draw (c3) edge (c2) ;
  \draw (c2) edge (b1) ;
  \draw (c1) edge (c2) ;
  \draw (c1) edge (d1) ;
  \draw (d4) edge (e3) ;
  \draw (d3) edge (d4) ;
  \draw (d3) edge (e2) ;
  \draw (d2) edge (c3) ;
  \draw (d2) edge (d1) ;
  \draw (e2) edge (e1) ;
  \draw (e1) edge (d2) ;

\end{tikzpicture}
} %% scalebox

%% file: CoPabisimilarity.tex
\section{\cop-Bisimilarity for \qdcm{}}
\label{sec:COPAbisimilarity}

In~\cite{CLMdV22,CLMdV25} several notions of spatial bisimilarity for closure
models have been investigated. In particular, \cm-bisimilarity, and
its refinement for \qdcm{s} \cmc-bisimilarity, are a fundamental
starting point for the study of spatial bisimilarity because of their
strong links to topo-bisimilarity. However, \cm\
  and~\cmc\ are rather
fine-grained relations for reasoning about general properties of
space, since they are directly based on the closure
operator.\footnote{Or its dual operator called `interior'.} For
instance, with reference to the model of Figure~\ref{fig:FiveByFive},
where all black points satisfy only atomic proposition~$b$ while the
grey ones satisfy only~$g$, the point at the center of the model is
\emph{not} \cmc-bisimilar to any other black point. This is because
\cmc-bisimilarity is based on the fact that points reachable ``in one
step'' --- i.e. contained in the closure --- are taken into consideration.
This, in turn, gives bisimilarity a flavour of ``counting'', that
goes against the idea that, for instance in Figure~\ref{fig:FiveByFive}, all black points in the
model are to be considered spatially equivalent. In fact, they are all
black and all can reach black or grey points. Furthermore, they could
be considered equivalent to the black point of a smaller model
consisting of just one black and one grey point mutually
connected---that would, in fact, be a ``minimal''  --- in a sense that will be made clear later in the paper --- representation of the
closure model.

\begin{figure}
\def\samplesz{2.5cm}
\centering
\subfloat[][]
{
\includegraphics[height=2.3cm]{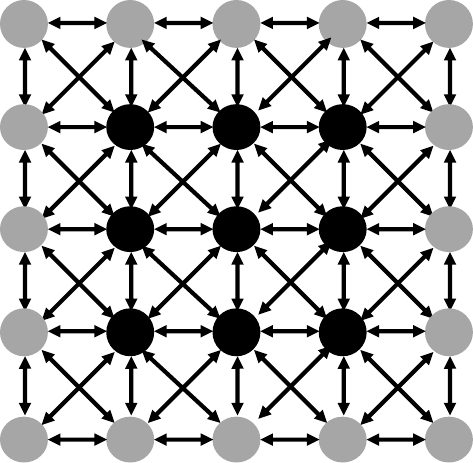}
\label{fig:FiveByFive}
}\quad \quad \quad
\centering
\subfloat[][]
{
\includegraphics[height=2.3cm]{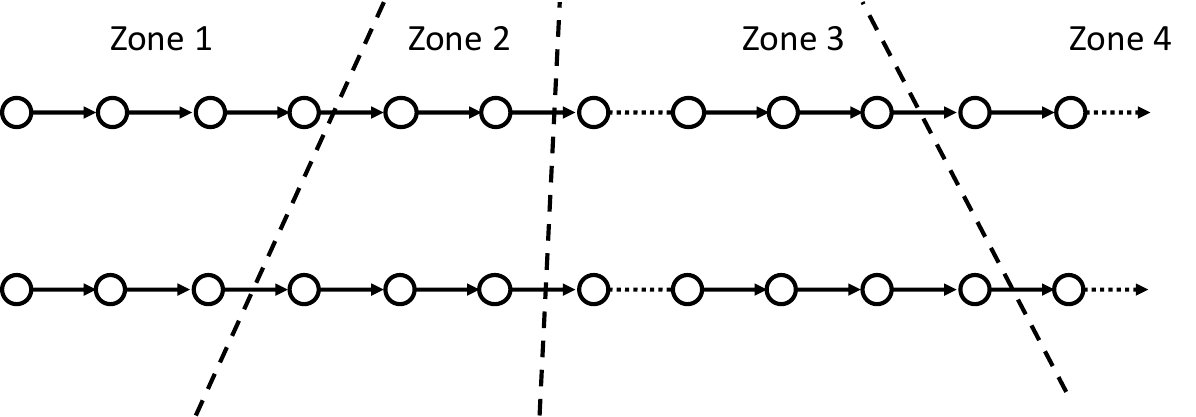}
\label{fig:Zones}
}

\caption{A model (a); zones in paths (b).}
\end{figure}

\noindent
In order to relax the ``counting'' capability of
\cmc-bisimilarity for~\qdcs{s} as
mentioned, a weaker notion of bisimilarity has been introduced
in~\cite{CLMdV22,CLMdV25} that is based on paths, instead of membership of closures, together with a notion of ``compatibility'' between relevant paths that essentially requires each of them be composed of a non-empty sequence of non-empty, adjacent ``zones''. 
More precisely, both paths under consideration in a transfer condition should share the same structure, as follows (see Figure~\ref{fig:Zones}):
\begin{itemize}
\item both paths are composed by a sequence of (non-empty) ``zones'';
\item the number of zones should be the same in both paths, \emph{but}
\item the length of sub-path in ``corresponding'' zones can be different, \emph{as well as}
 the length of each of the two paths;
\item \emph{each} point in a zone of a path should be related by the bisimulation to \emph{every} point in the corresponding zone of the other path.
\end{itemize}
This notion of compatibility gives rise to the notion of
\emph{Compatible Path bisimulation}, more briefly \cop-bisimulation, recalled below for \qdcm{s}.

\begin{defi} [path-compatibility]
  \label{def:path-compatability}
  Given \qdcm~$\model = (X, \closure, \peval)$ and relation
  $B \subseteq {X \mathop\times X}$, two paths
  $\pi = (x_i)_{i{=}0}^\ell$ and~$\varrho = (y_j)_{j{=}0}^m$ in~$X$
  are \emph{compatible with respect to~$B$} if $N > 0$ and two
  monotone surjections $f : [0,\ell] \to [1,N]$ and
  $g : [0,m] \to [1,N]$ exist such that $x_i \mopB y_j$ for all
  indices $0 \leqslant i \leqslant \ell$ and
  $0 \leqslant j \leqslant m$ with $f(i) = g(j)$.
  \closedefi
\end{defi}

\noindent
The functions $f$ and~$g$ are referred to as \emph{matching functions}
for $\pi$ and~$\varrho$.
The choice for $N$, $f$, and~$g$ is not unique.
The minimal number $N > 0$ for which matching functions exist is
defined to be the number of zones for the paths $\pi$ and~$\varrho$.

\begin{defi} [\cop-bisimilarity]
  \label{def:CoPabisimilarity}
  A symmetric relation $B \subseteq X \times X$ is a
  \emph{\cop-bisimulation relation} for \qdcs{}
  $\model = (X, \closure, \peval)$ if, whenever $x,y \in X$ satisfy
  $x \mopB y$, the following holds:
  \begin{enumerate}
  \item $x \eqV y$;
  \item \label{fwd-transfer} for every path~$\pi$ from~$x$ exists a
    $B$-compatible path~$\varrho$ from~$y \mkern1mu$;
  \item \label{bwd-transfer} for every path~$\pi$ to~$x$ exists a
    $B$-compatible path~$\varrho$ to~$y$.
  \end{enumerate}
  Two points $x,y \in X$ are called \emph{\cop-bisimilar} in $\model$
  if $x \mopB y$ for some \cop-bisimulation relation~$B$ for
  $\model$. Notation, $x \copbis y$.
  \closedefi
\end{defi}

\noindent
For a model~$\model = (X, \closure, \peval)$, it
  is immediate from the definition that \cop-bisimilarity
${\copbis} \subseteq {X \times X}$ in~$\model$, being the union of all
\cop-bisimulation relations, is a \cop-bisimulation relation
itself, actually the largest bisimulation relation. Moreover,
as a consequence of Theorem~\ref{thm:CoPabisEqIcrleq} below, \cop-bisimilarity $\copbis$ is an equivalence
relation.

\begin{lem}
  \label{lem-single-step-bisimilarity-suffices}
  Let $\model = (X, \closure, \peval)$ be a \qdcm\ and let
  $B \subseteq {X \times X}$ be a symmetric relation such that for all
  $x, y \in X$ satisfying $x \mopB y$ it holds that
  \begin{enumerate}
  \item $x \eqV y$;
  \item 
   if $x' \in \closure(x)$ for some $x' \in X$, then a path $(y_j)_{j{=}0}^m$
   from~$y$ exists for some $m \geq 0$ such that $y_j \, B \,x$ for all $j \in [0,m)$ and $y_m \, B \, x'$;
  \item %\label{bwd-transfer}
   if $x \in \closure(x')$ for some $x' \in X$, then a path $(y_j)_{j{=}0}^m$
   to~$y$ exists for some $m \geq 0$ such that $y_j \, B \,x$ for all $j \in (0,m]$ and $y_0 \, B \, x'$;
  \end{enumerate}
  Then $B$~is a \cop-bisimulation for~$\model$. 
\end{lem}

\begin{proof}
  We verify conditions (\ref{fwd-transfer}) and~(\ref{bwd-transfer})
  of Definition~\ref{def:CoPabisimilarity}.

  For the forward transfer condition (\ref{fwd-transfer}), assume
  $x \mopB y$ and let $\pi = ( x_i)_{i{=}0}^\ell$ be a path
  from~$x$. We show by induction on~$\ell$ that a $B$-compatible
  path~$\varrho$ from~$y$ exists. Basis, $\ell = 0$: The single point
  path $\varrho = (y)$ from~$y$ is $B$-compatible
  with~$\pi$. Induction step, $\ell > 0$: Consider the path
  $\pi_0 = (x_0,x_1)$ of length~$1$ from~$x_0$ and the path
  $\pi_1 = ( x_i )_{i{=}1}^\ell$ of length~$\ell{-}1$
  from~$x_1$. Because $x_0 \mopB y$ and $x_1 \in
    \closure(x)$,
  by property~(2) of~$B$, a path
  $\varrho_0 = ( y_j)_{j{=}0}^m$ from~$y$ for some~$m \geqslant 0$
  exists with
  $y_j \mopB x_0$ for $j \in [0,m)$ and $y_m \mopB x_1$, i.e.,
    the path~$\varrho_0$
  is $B$-compatible with~$\pi_0$. As
  $x_1 \mopB y_m$, by induction hypothesis, a path
  $\varrho_1 = (y_j)_{j{=}m}^{m{+}k}$ from~$y_m$ for
  some~$k \geqslant 0$ exists that is $B$-compatible
  with~$\pi_1$. Then the path $\varrho = ( y_j )_{j{=}0}^{m{+}k}$
  from~$y$ is $B$-compatible with~$\pi = ( x_i)_{i{=}0}^\ell$ as can
  be straightforwardly verified.
  
  For the backward transfer condition~(\ref{bwd-transfer}), $x \mopB y$ and let
  $\pi=( x_i)_{i{=}0^\ell}$ be a path to~$x$. We show by induction on~$\ell$ that a $B$-compatible path~$\varrho$
  to~$y$ exists.
  Basis, $\ell = 0$: The single point path $\varrho = (y)$ to~$y$ is $B$-compatible with~$\pi$.
  Induction step, $\ell > 0$: Consider the path $\pi_0 = (x_{\ell{-}1},x_{\ell})$ of length~$1$ to~$x$ 
  and the path $\pi_1 = ( x_i )_{i{=}0}^{\ell{-}1}$ of length~$\ell{-}1$ to~$x_{\ell{-}1}$.
  Because by hypothesis $x_{\ell} \mopB y$ and $x_{\ell} \in \closure(x_{\ell{-}1})$, since 
  since $\pi$ is a path, we have, by property~(3) of~$B$, that a path  $\varrho_0 = ( y_j)_{j{=}0}^m$ to~$y$ exists for
  some $m\geq 0$ such that $y_0 \mopB x_{\ell{-}1}$ and  $y_j \mopB x_{\ell}$ for all $j \in (0,m]$.
  Thus, path~$\varrho_0$ is $B$-compatible with~$\pi_0$. 
  As $x_{\ell{-}1} \mopB y_o$, by induction hypothesis, a
  path $\varrho_1 = (y'_j)_{j{=}0}^{k}$ to~$y_0$ exists, for
  some~$k \geqslant 0$, that is $B$-compatible with~$\pi_1$.
  Then, the path $\varrho$
  $$
  \varrho(j) = 
  \left\{
  \begin{array}{l}
  y'_j, \mbox{ if } j \in [0,k], \\\\
  y_{j-k}, \mbox{ if } j \in  (k,k+m].
  \end{array}
  \right.
  $$
 to $y$ is $B$-compatible with~$\pi$ as can be straightforwardly verified.
%  For the backward transfer condition~(\ref{bwd-transfer}), assume
%  $x \mopB y$ and let $\pi = ( x_i)_{i{=}\ell}^0$ be a path to~$x$. We
%  show by induction on~$\ell$ that a $B$-compatible path~$\varrho$
%  to~$y$ exists. Basis, $\ell = 0$: The single point path
%  $\varrho = (y)$ to~$y$ is $B$-compatible with~$\pi$. Induction step,
%  $\ell > 0$: Consider the path $\pi_0 = (x_1,x_0)$ of length~$1$
%  to~$x_0$ and the path $\pi_1 = ( x_i )_{i{=}\ell}^1$ of
%  length~$\ell{-}1$ to~$x_1$. Because $x_0 \mopB y$ \evadded{and $x_1
%    \in \closure(x_0)$}, by property~(3)
%  of~$B$, a path $\varrho_0 = ( y_j)_{j{=}m}^0$ to~$y$ for
%  some~$m \geqslant 0$ exists \evadded{such} that \evremoved{is $B$-compatible with~$\pi_0$. We can assume} $x_1 \mopB y_m$ and $x_0 \mopB y_j$ for
%  $m > j \geqslant 0$.
%  \evadded{Thus, path~$\varrho_0$ is $B$-compatible with~$\pi_0$.}
%%
%  As $x_1 \mopB y_m$, by induction hypothesis, a
%  path $\varrho_1 = (y_j)_{j{=}m{+}k}^{m}$ to~$y_m$ for
%  some~$k \geqslant 0$ exists that is $B$-compatible
%  with~$\pi_1$. Then the path $\varrho = ( y_j )_{j{=}m{+}k}^{0}$
%  to~$y$ is $B$-compatible with~$\pi = ( x_i)_{i{=}\ell}^0$ as can be
%  straightforwardly verified.
\end{proof}

\noindent
The logic~\icrl\ introduced in~\cite{CLMdV22,CLMdV25} provides a
logical characterisation of \cop-bisimilarity.
Besides atomic propositions, negations, and conjunctions, the logic
features two modalities, the forward modality~$\lstothru$ and the
backward modality~$\lsfromthru$.
Formulas~$\form$ are given by
\begin{displaymath}
  \form ::= p \mid \lneg \mkern2mu \form \mid 
  \textstyle{\liand_{i {\in} I}} \: \form_i \mid 
  \lstothru \mkern2mu \form_1 [\form_2] \mid
  \lsfromthru \mkern2mu \form_1 [\form_2]
\end{displaymath}
Here, $p$~ranges over~$\ap$, and $I$~ranges over a collection of finite
and countably infinite index sets.
The satisfaction relation of~\icrl\ in a point~$x$ of a \qdcm~$\model$
is
\begin{displaymath}
  \def\arraystretch{1.2}
  \begin{array}{r@{\,}c@{\,}l c l}
  \model,x & \models_{\icrl} & p & \Leftrightarrow &
  x  \in \peval(p)
\\
  \model,x & \models_{\icrl} & \lneg \mkern2mu \form & \Leftrightarrow &
  \text{$\model,x  \models_{\icrl} \form$ does not hold}
\\
  \model,x & \models_{\icrl} & \liand_{i {\in} I} \form_i & \Leftrightarrow &
  \text{$\model,x  \models_{\irl} \form_i$ for all $i \in I$}
\\
  \model,x & \models_{\icrl} & \lstothru \mkern2mu \form_1 [\form_2] & 
  \Leftrightarrow &
  \text{a path~$\pi$ and an index~$\ell$ exist such that $\pi(0) = x$},              \\ &&&&
  \text{$\model, \pi(\ell) \models_{\icrl} \form_1$, and 
  $\model, \pi(j) \models_{\icrl} \form_2$ for $0 \leqslant j < \ell$}
\\
  \model,x & \models_{\icrl} & \lsfromthru \mkern2mu \form_1 [\form_2] & 
  \Leftrightarrow &
  \text{a path~$\pi$ and an index~$\ell$ exist such that $\pi(\ell) =  x$},
\\ &&&&
  \text{$\model,\pi(0) \models_{\icrl} \form_1$, and
  $\model, \pi(j) \models_{\icrl} \form_2$ for $0 < j  \leqslant \ell$}
  \end{array}
  \def\arraystretch{1.0}
\end{displaymath}

\noindent
Logical equivalence~ $\icrleq$ for~\icrl{} is defined as expected:
With respect to a \qdcm~$\model$, we have that points $x$ and~$y$ are
logically equivalent for~\icrl, notation $x \icrleq y$, iff
${\model,x \models_{\icrl} \form} \Leftrightarrow {\model,y
  \models_{\icrl} \form}$.  The following result, proven
in~\cite{CLMdV25}, establishes the relationship between
\cop-bisimilarity and \icrl-equivalence.
\begin{thm}
  \label{thm:CoPabisEqIcrleq}
  For every \qdcm~$\model$ it holds that \icrl-equivalence~$\icrleq$
  coincides with \cop-bisimilarity~$\copbis$.
  \qed
\end{thm}

\noindent
In the remainder of the paper, since we are concerned with finite
models only, we confine to the finitary fragment of \icrl, i.e.\ the
part where $I$~is a finite index set.

\bigskip

\noindent
In this work, given a \qdcm\ $\model=(X , \closure, \peval)$, we aim at running the spatial model checking algorithm of~\cite{Ci+16} on the quotient of~$\model$ with respect to~$\copbis$. 
The remainder of this paper is devoted to explain how to compute this quotient. It is a natural question at this point, whether the minimal model exists in the class of \qdcm{s}. 
In other words, one needs to show that the set of equivalence classes of~$\copbis$ \emph{can} be endowed with a quasi-discrete closure operator, in such a way that logical truth is preserved and reflected. 
We do so in Theorem~\ref{prop:logical-minimization} below.
With respect to~$\model$, for a point~$x \in X$, the notation~$[x]$ is
used to the denote the equivalence class of~$x$ modulo~$\copbis$

\begin{thm}
\label{prop:logical-minimization}
  Given \qdcs~$\model = (X, \closure, \peval)$, let
  $\wtM = (\wtX, \wtC, \wtV)$ where $\wtX = X / {\copbis}$,
  $\wtC([x]) = \ZET{[y']}{\exists \mkern1mu y \in [x] \colon y' \in
    \closure(y)}$ for~$x \in X$, and 
  $\wtV(p)= \ZET{[x]}{x \in \peval(p)}$ for~$p \in \ap$.  Then
  $\wtM$~is the smallest \qdcs\ such that
  ${\model,x \models \form} \Leftrightarrow {\wtM,[x] \models \form}$
  for $x \in X$, $\form \in \icrl$.
\end{thm}

\begin{proof}
  Clearly, $\wtM$~is a \qdcm. We verify by induction on the structure
  of a \icrl-formula~$\Phi$ that
  ${\model,x \models \form} \Leftrightarrow {\wtM,[x] \models \form}$.

  \begin{enumerate} [align=left]
  \item [Case~$p$:] For $p \in \ap$, we have $\model,x \models p$ iff
    $x \in \peval(p)$ iff $[x] \in \wtV(p)$ iff $\wtM,[x] \models p$.

  \item [Cases~$\neg \form$ and $\form \land \Psi$:] These cases follow directly
    from the induction hypothesis.

  \item [Case~{$\lstothru \mkern2mu \form [\Psi]$}:] Suppose
    $\model,x \models \lstothru \mkern2mu \form [\Psi]$. Let the path
    $\pi = (x_i)_{i{=}0}^\ell$ be such that $x_0 = x$,
    $\model,x_\ell \models \form$, and $\model,x_i \models \Psi$ for
    $0 \leqslant i < \ell$. Let $\wtpi = ([x_i])_{i{=}0}^\ell$. Then
    $[x_{i{+}1}] \in \wtC([x_i])$ since $x_{i{+}1} \in \closure(x_i)$
    for $0 \leqslant i < \ell$. So, $\wtpi$~is a path
    in~$\wtX$. Moverover, $[x_0] = [x]$,
    $\wtM,[x_\ell] \models \form$, and $\wtM,[x_i] \models \Psi$ for
    $0 \leqslant i < \ell$ by induction hypothesis for $\form$
    and~$\Psi$, respectively. 

    Reversely, suppose $\wtM,[x] \models \lstothru \form [\Psi]$. Let
    the path $\wtpi = ( [\xbar_i] )_{i{=}0}^\ell$, where $\xbar_i \in
    X$ for $0 \leqslant i \leqslant \ell$, be such that
    $[\xbar_0] = [x]$, $\wtM,[\xbar_\ell] \models \form$, and
    $\wtM,[\xbar_i] \models \Psi$ for $0 \leqslant i < \ell$. We
    proceed by induction on~$\ell$ to show that, for
    some~$m, k \geqslant 0$, a path $\pi = ( y_j )_{j{=}0}^{m{+}k}$
    exists with $y_0 = x$, $\model,y_{m{+}k} \models \form$, and
    $\model,y_j \models \Psi$ for $0 \leqslant j <
    m{+}k$. Basis,~$\ell = 0$:
    Apparently, $\wtM,[x] \models \Phi$. Hence, $\model,x
      \models \Phi$ by induction hypothesis for~$\Phi$. So,
    the single point path~$(x)$ fulfills the requirements.

    Induction step,~$\ell > 0$: We have
    $[\xbar_1] \in \wtC([\xbar_0])$. So, $\xbarp_0, \xbarp_1 \in X$
    exist with $\xbarp_0 \copbis \xbar_0 \copbis x$,
    $\xbarp_1 \in \closure(\xbarp_0)$, and $\xbarp_1 \copbis
    \xbar_1$.
  In view of the path $(\xbarp_0,\xbarp_1)$ and $\xbarp_0
      \copbis \xbar_0$,   let the path $\pi_0 = (y_j)_{j{=}0}^m$ be such that
    $y_0 = x$, $y_j \copbis \xbarp_0$ for $0 \leqslant j < m$, and
    $y_m \copbis \xbarp_1$. Hence $y_j \copbis \xbar_0$ for
    $0 \leqslant j < m$. Because $\wtM,[\xbar_0] \models \Psi$, it follows that
      $\model,\xbar_0 \models \Psi$ by induction hypothesis
      for~$\Psi$. So $\model, y_j \models \Psi$ for $0 \leqslant j < m$ by
    Theorem~\ref{thm:CoPabisEqIcrleq}. Note, we have
    $\wtM,[\xbar_1] \models \lstothru \form [\Psi]$ as witnessed by
    the path $([\xbar_i])_{i{=}1}^\ell$ of length~$\ell{-}1$. So,
    $\model, \xbar_1 \models \lstothru \form [\Psi]$ by induction
    hypothesis for~$\ell$. Therefore,
    $\model, y_m \models \lstothru \form [\Psi]$ because
    $y_m \copbis \xbarp_1 \copbis \xbar_1$ and again
    Theorem~\ref{thm:CoPabisEqIcrleq}. Let the path
    $\pi_1 = (y_j)_{j{=}m}^{m{+}k}$ from~$y_m$ be such that
    $\model, y_{m{+}k} \models \form$ and $\model, y_j \models \Psi$
    for $m \leqslant j < m{+}k$. Then the combined path
    $\pi = (y_j)_{j{=}0}^{m{+}k}$ from~$x$ satisfies the required, and
    $\model,x \models \lstothru \form [\Psi]$.

  \item [Case~{$\lsfromthru \mkern2mu \form [\Psi]$}:] Suppose
    $\model,x \models \lsfromthru \mkern2mu \form [\Psi]$. Let the path
    $\pi = (x_i)_{i{=}0}^\ell$ be such that $x_\ell = x$,
    $\model,x_0 \models \form$, and $\model,x_i \models \Psi$ for
    $0 < i \leqslant \ell$. Put $\wtpi = ([x_i])_{i{=}0}^\ell$. Then
    $[x_{i{+}1}] \in \wtC([x_i])$ since $x_{i{+}1} \in \closure(x_i)$
    for $0 \leqslant i < \ell$. Thus, $\wtpi$~is a path
    in~$\wtX$. Moverover, $[x_\ell] = [x]$,
    $\wtM,[x_0] \models \form$, and $\wtM,[x_i] \models \Psi$ for
    $0 < i \leqslant \ell$ by induction hypothesis for $\form$
    and~$\Psi$, respectively.

    For the other direction, suppose
    $\wtM,[x] \models \lsfromthru \form [\Psi]$. Let the path
    $\wtpi = ( [\xbar_i] )_{i{=}0}^\ell$, where $\xbar_i \in X$ for
    $0 \leqslant i \leqslant \ell$, be such that $[\xbar_\ell] = [x]$,
    $\wtM,[\xbar_0] \models \form$, and $\wtM,[\xbar_i] \models \Psi$
    for $0 < i \leqslant \ell$.

    We proceed by induction on~$\ell$ to
    verify that, for some~$k,m \geqslant 0$, a path
    $\pi = ( y_j )_{j{=}0}^{k+m}$ to~$x$ exists with
    $\model,y_{0} \models \form$ and $\model,y_j \models \Psi$ for
    $0 < j \leqslant k{+}m$. Basis,~$\ell = 0$: The single point
    path~$(x)$ fulfills the requirements, as apparently,
    $\wtM,[x] \models \Phi$ and therefore $\model,x \models \Phi$ by
    induction hypothesis for~$\Phi$.

    Induction step,~$\ell > 0$: We have
    $[\xbar_\ell] \in \wtC([\xbar_{\ell{-}1}])$. Therefore, we can
    find $\xbarp_{\ell{-}1}, \xbarp_\ell \in X$  with
    $\xbarp_{\ell{-}1} \copbis \xbar_{\ell{-}1}$, 
    $\xbarp_\ell \in \closure(\xbarp_{\ell{-}1})$, and
    $\xbarp_\ell \copbis \xbar_\ell$. In view of the path
    $(\xbarp_{\ell{-}1},\xbarp_\ell)$ and
    $\xbarp_\ell \copbis \xbar_\ell$, let the path
    $\pi'' = (y''_j)_{j{=}0}^m$ to $\xbar_\ell $ be such that
    $y''_0 \copbis \xbarp_{\ell{-}1}$, $y''_j \copbis \xbarp_\ell$ for
    $0 < j \leqslant m$.  Hence
    $y''_j \copbis \xbar_\ell$ for $0 < j \leqslant m$. Because
    $\wtM,[\xbar_\ell] \models \Psi$, we have
    $\model,\xbar_\ell \models \Psi$ by induction hypothesis
    for~$\Psi$ and $\model, y''_j \models \Psi$ for $0 < j \leqslant m$
    by Theorem~\ref{thm:CoPabisEqIcrleq}. Note, we have
    $\wtM,[\xbar_{\ell{-}1}] \models \lsfromthru \form [\Psi]$ as
    witnessed by the path $([\xbar_i])_{i{=}0}^{\ell{-}1}$ of
    length~$\ell{-}1$. So,
    $\model, \xbar_{\ell{-}1} \models \lsfromthru \form [\Psi]$ by
    induction hypothesis for~$\ell$. Therefore,
    $\model, y''_0 \models \lsfromthru \form [\Psi]$ because
    $y''_0 \copbis \xbarp_{\ell{-}1} \copbis \xbar_{\ell{-}1}$ and
    another application of Theorem~\ref{thm:CoPabisEqIcrleq}. Let the
    path $\pi' = (y'_j)_{j{=}0}^{k}$ to~$y''_0$ be such that
    $\model, y'_{0} \models \form$ and $\model, y'_j \models \Psi$
    for $0 < j \leqslant k$. Let the combined path
    $\pi = (y_j)_{j{=}0}^{k{+}m}$ be given by $y_j = y'_j$ for $0
      \leqslant j \leqslant k$ and $y_j = y''_{j-k}$ for $0 \leqslant j
      \leqslant m$. Note, $y'_k = y''_0$.
   Then the path~$\pi$ is as required: It holds that
      $\model,y_{0} \models \form$ as $y_0 = y'_0$.
      Moreover,
      $\model,y_j \models \Psi$ for $0 < j \leqslant k$ and
      $k < j \leqslant k{+}m$ as $\model,y'_j \models \Psi$ for
      $0 < j \leqslant k$ and $\model,y''_{j{-}k} \models \Psi$ for
      $k < j \leqslant k{+}m$, respectively.
      Finally, $y_{k{+}m} = y''_m = \xbar_\ell = x$. Thus, path~$\pi$
      is a path to~$x$ that witnesses
      $\model,x \models \lsfromthru \form [\Psi]$, as was to be shown.
  \end{enumerate}
  As to $\wtM$ being the \qdcm\ with the smallest number of elements,
  suppose $\mathcal{N} = (Y, \mathcal{D}, \mathcal{W})$ is a \qdcm\
  and $f : X \to Y$ is a surjection such that
  ${\model,x \models \form} \Leftrightarrow {\mathcal{N},f(x) \models
    \form}$ for $x \in X$, $\form \in \icrl$. Define the mapping
  $g : Y \to \wtX$ by $g(y) = [x]$ if~$f(x) = y$. This is
  well-defined: If $f(x) = f(x')$, then $x \icrleq x'$ by the assumed
  property of~$\mathcal{N}$. Thus, $x \copbis x'$ by
  Theorem~\ref{thm:CoPabisEqIcrleq} and $[x] = [x']$. By surjectivity
  of~$f$, also $g$~is surjective. Therefore, $|Y| \geqslant |\wtX|$.
\end{proof}

\begin{exa}
Consider the simple  \qdcm\ in Figure~\ref{fig:QdCM1}. It has five
equivalence classes, namely $C_1 = \{x_1,x_2\}, C_2 = \{x_3\},
C_3 = \{x_4\}, C_4 = \{x_5\}$, and $C_5 = \{x_6\}$. 
To verify this, it is in view of
  Theorem~\ref{thm:CoPabisEqIcrleq} convenient to use logical equivalence rather than \cop-bisimilarity.
The points $x_1$ and~$x_2$ are in the same class, but $x_5$~is not in that class. 
This is so because from $x_1$ and~$x_2$, that satisfy the same atomic proposition~$r$, one can reach point~$x_3$, labelled by~$g$ that, in turn, can be \emph{reached from}~$x_4$ which is labelled by~$b$. 
This does not hold for~$x_5$. 
The formula that distinguishes $x_1$ (and~$x_2$) from~$x_5$ is $\lstothru \, (\lsfromthru \, b [g]) [r] $. 
In fact, from $x_4$ there is a path $\pi$, with $\pi(0)=x_4$, labelled by~$b$, and with $\pi(1) = x_3$, labelled by~$g$, so $x_3$~satisfies $\lsfromthru \, b [g]$. 
This formula is not satisfied by~$x_6$ since there is no path, starting in a point with label~$b$ and then going to~$x_6$. 
Furthermore, there is a path $\pi'$, with $\pi'(0)=x_1$, $\pi'(1)=x_2$ and $\pi'(2)=x_3$, so $x_1$ (and~$x_2$) satisfy $\lstothru \, (\lsfromthru \, b [g]) [r] $.
Also $x_3$ and~$x_6$, both labelled by~$g$ are not in the same class because, as we have seen, $x_3$~can be reached from a point satisfying~$b$, but $x_6$~cannot. 
The minimal model, applying Theorem~\ref{prop:logical-minimization}, is shown in Figure~\ref{fig:QdCMmin}.
\end{exa}

\begin{figure}
    \centering 
    \begin{tikzpicture}[%
      every state/.append style={inner sep=1pt, minimum size=7mm}]
      
    \node[state, label=above:{$\SET{r}$}] (s1) {$x_1$};
    \node[state, label=above:{$\SET{r}$}, right of=s1] (s2) {$x_2$};
    \node[state, label=above:{$\SET{g}$}, right of=s2] (s3) {$x_3$};
    \node[state, label=above:{$\SET{b}$}, right of=s3] (s4) {$x_4$};
    \node[state, label=above:{$\SET{r}$}, right of=s4, xshift=5mm] (s5) {$x_5$};
    \node[state, label=above:{$\SET{g}$}, right of=s5] (s6) {$x_6$};
    \draw   
            (s1) edge [thick] node{} (s2)
            (s2) edge [thick] node{} (s3)
            (s4) edge [thick] node{} (s3)            
            (s5) edge [thick] node{} (s6);            
    \end{tikzpicture}
    \caption{A finite \qdcm{}}
    \label{fig:QdCM1}
\end{figure}

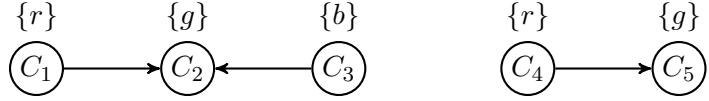
\begin{figure}
    \centering 
    \begin{tikzpicture}[%
      every state/.append style={inner sep=1pt, minimum size=7mm}]
      
    \node[state, label=above:{$\SET{r}$}] (s7) {$C_1$};
    %\node[state, label=above:{$\SET{r}$}, right of=s1] (s2) {$s_2$};
    \node[state, label=above:{$\SET{g}$}, right of=s7] (s8) {$C_2$};
    \node[state, label=above:{$\SET{b}$}, right of=s8] (s9) {$C_3$};
    \node[state, label=above:{$\SET{r}$}, right of=s9, xshift=5mm] (s10) {$C_4$};
    \node[state, label=above:{$\SET{g}$}, right of=s10] (s11) {$C_5$};
    \draw   
            (s7) edge [thick] node{} (s8)
 %           (s8) edge [thick] node{} (s9)
            (s9) edge [thick] node{} (s8)            
            (s10) edge [thick] node{} (s11);            
    \end{tikzpicture}
    \caption{Minimal \qdcm{} of the \qdcm\ of Figure~\ref{fig:QdCM1}.}
    \label{fig:QdCMmin}
\end{figure}

\noindent
We close this section remarking that if a \qdcm\ $\model$ is
symmetric, i.e., the relation~$R$ underlying the
 closure operator of~$\model$ is a symmetric relation, then $\model,x\,   \models_{\icrl}\,  \lstothru \, \form_1 [\form_2]$ if and only if $\model,x\, \models_{\icrl}\, \lsfromthru \, \form_1 [\form_2]$, for all $x \in \model$ and \icrl\ formulas $\form_1$ and $\form_2$. 

%% file: Translation.tex
\section{From \qdcm{s} to Labelled Transition Systems}
\label{sec:translation}

In this section we show how a finite \qdcm{} can be encoded as an
\lts{} in such a way that \cop-bisimilarity in the \qdcm{} is
preserved and reflected by branching bisimilarity in the LTS: two
points in the \qdcm{} are mapped to branching bisimilar states in the
\lts{} precisely when the two points are \cop-bisimilar in the first
place. We first present the encoding for general \qdcs{s}, i.e.,
\qdcs{s} that may or may not be symmetric. Next, the encoding or
symmetric \qdcs{s} is given. Because of the symmetry in these models,
a more compact encoding is possible.

\subsection{\lts\ encoding of general \qdcs{s}}

 \mbox{} \smallskip

\noindent
The idea of the encoding of a \qdcm\ $\model = (X,\closure,\peval)$
that is not necessarily symmetric is as follows: (i)~each element~$x \in X$
corresponds to exactly two states in the \lts, a forward
state~$\vec{x}$ and a backward state~$\cev{x}$ that have transitions
back and forth between them labelled $\eC$ and~$\eD$,
standing for ``converse'' and ``direct'', 
(ii)~for each proposition letter $p$ that is satisfied in the \qdcm{} by $x$ there is a selfloop
in the forward state~$\vec{x}$ labelled by $p$, (iii)
there are $\tau$-transitions between two forward states $\vec{x}$ and~$\vecp{x}$
and between two backward states $\cev{x}$ and~$\cevp{x}$ if $x$ and~$x'$
satisfy the same proposition letters, and (iv)~ between two forward states $\vec{x}$ and~$\vecp{x}$ --- 
respectively, two backward states $\cev{x}$ and~$\cevp{x}$ --- there are transitions with a special label~$\ch$, indicating change,
if $x$ and~$x'$ do not satisfy the same proposition letters.

\begin{defi}[Encoding, general case]
\label{def:QdCM2LTS}
Let $\model = (X,\closure,\peval)$ be a finite \cm\@.
Define labelled transition system $\ltsT(\model) = (S, \act, {\rightarrow})$ by
\begin{enumerate} [label=(\roman*)]
\item $S = \ZET{\vec{x} ,\, \cev{x} }{x \in X}$;
\item $\act = \ap \cup \SET{\eC, \eD, \ch, \tau}$ with $\eC$, $\eD$,
  $\ch$, $\tau$ fresh labels;
\item $ {\rightarrow} \subseteq {S \times \act \times S}$ such that
  \vspace*{-1.0\baselineskip}
  \def\arraystretch{1.2}
\begin{displaymath}
\begin{array}{r@{,\:}ll}
\multicolumn{2}{c}{\vec{x} \pijl{p} \vec{x}} 
& \text{for $x \in X$, $p \in \ap$ such that $x \in \peval(p)$} 
\\
\vec{x} \pijl{\eC} \cev{x} 
& \cev{x} \pijl{\eD} \vec{x} 
& \text{for $x \in X$} 
\\
\vec{x} \pijl{\tau} \vecp{x} 
& \cevp{x} \pijl{\tau} \cev{x} 
& \text{for $x,x' \in X$ such that $x' \in \closure(x)$ and $x \eqV x'$}
\\
\vec{x} \pijl{\ch} \vecp{x} 
& \cevp{x} \pijl{\ch} \cev{x} 
& \text{for $x,x' \in X$ such that $x' \in \closure(x)$ and $x \neqV x'$}
\\
\end{array}
\end{displaymath}
  \def\arraystretch{1.0}
\end{enumerate}
\closedefi
\end{defi}

\begin{figure}
  \centering
  \small
  \input{fig-encoding-example}
  \caption{(a) \qdcm{} of Figure~\ref{fig:QdCM1} repeated; (b)~its
    encoding as \lts\@, where $\tau$-selfloops are omitted for readability reasons.}
  \label{fig:LTSQdCM1}
\end{figure}
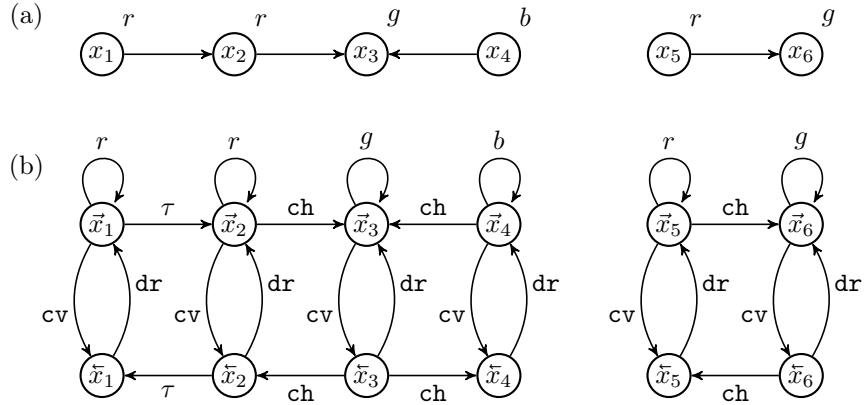

\begin{exa}
  \label{exe-encoding}
  Figure~\ref{fig:LTSQdCM1} depicts at the top the \qdcm{} of
  Figure~\ref{fig:QdCM1} and at the bottom its encoding as \lts\@.
  Its six elements $x_1, \ldots, x_6$ give rise to the states
  $\vec{x}_1, \ldots, \vec{x}_6$ and $\cev{x}_1 , \ldots, \cev{x}_6$
  with transitions with label $\eC$ and~$\eD$ between them.
  Since the elements $x_1$, $x_2$, and~$x_5$ satisfy proposition
  letter~$r$, the states $\vec{x}_1$, $\vec{x}_2$, and~$\vec{x}_5$
  have a selfloop with label~$r$. Similarly, the elements $x_3$
  and~$x_6$ with proposition letter~$g$ yield the selfloops with
  label~$g$ for $\vec{x}_3$ and~$\vec{x}_6$, the element~$x_4$ with
  proposition letter~$b$ yields the selfloop with label~$b$.
  As indicated by the arrow, we have that $x_1 \in
  \closure(x_2)$. Moreover, both $x_1$ and~$x_2$ satisfy proposition
  letter~$r$ only. So, in the \lts{} we have the transitions
  $\vec{x}_1 \pijltau \vec{x}_2$ and 
  $\cev{x}_2 \pijltau \cev{x}_1$.
  Since $x_2, x_4 \in \closure(x_3)$ and $x_5 \in \closure(x_6)$ and
  different proposition letters are satisfied by the pairs of
  elements, we have $\ch$-transitions $\vec{x}_2 \pijlch \vec{x}_3$,
  $\vec{x}_4 \pijlch \vec{x}_3$, and $\vec{x}_5 \pijlch \vec{x}_6$ and
  their counterparts $\cev{x}_3 \pijlch \cev{x}_2$,
  $\cev{x}_3 \pijlch \cev{x}_4$, and $\cev{x}_6 \pijlch \cev{x}_5$.
\end{exa}

\noindent
Below, Theorem~\ref{theo:CoPaBran} states that \cop-bisimilarity of
two elements of a closure model~$\model$ coincides with branching
bisimilarity of their corresponding states
in $\ltsT(\model)$. In preparation of the proof the theorem
we first establish two lemmas regarding the structure
of~$\ltsT(\model)$.
In turn, in aid of proving the first lemma, we introduce the notion of
the depth of a state in an~\lts\ and note a number of properties of it.

Let a finite \lts\ $\LTS = (S,\act,{\rightarrow})$ and an
action~$\tau \in \act$ be given. For~$s \in S$, we define the
depth~$|s|$ of~$s$ in~$\LTS$ by
\begin{displaymath}
  |s| = \max \ZET{n \in \mathbb{N}}{\exists \mkern1mu s_0, \ldots, s_n
    \in S \colon \forall \mkern1mu i ,\, 0 \leqslant i < n \colon (s_i
    \pijlstertau s_{i{+}1}) \land \neg \mkern1mu (s_{i{+}1} \pijlstertau s_i)}
\end{displaymath}
By finiteness of~$S$, the above is well-defined; in the sequence $s_0,
\ldots, s_n$ each state~$s_i$ can only occur once.

For $s,t \in S$ we define $s \equiv_\tau t$ iff
$(s \pijlstertau t) \land (t \pijlstertau s)$.  A straightforward
verification shows that the relation
${\equiv_\tau} \subseteq {S \times S}$ is a branching bisimulation
relation on~$\LTS$. Therefore, (i)~$s \equiv_\tau t$
implies~$s \bis_b t$. Also, (ii)~if $|s| = |t|$ and $s \pijlstertau t$
then we have $s \equiv_\tau t$. For, if $s \not\equiv_\tau t$, we have
$\neg \mkern1mu (t \pijlstertau s)$ and $|s| > |t|$ would follow.
Thus, (iii)~if $|s| = 0$ and~$s \pijlstertau t$ it follows
that~$s \equiv_\tau t$.  Finally, from the definition of depth one
directly obtains that (iv)~$s \pijlstertau t$ implies
$|s| \geqslant |t|$.

\begin{lem} 
\label{la:factI}
Let $\model = (X,\closure,\peval)$ be a finite \cm\@. 
For $x \in X$, it holds in $\ltsT(\model)$ that
\begin{displaymath}
\text{if $\vec{x} \pijltau \vecp{x}$ and $\vec{x} \bis_b \vecp{x}$ then $\cevp{x} \pijltau \cev{x}$ and $\cev{x}\bis_b \cevp{x}$}
\end{displaymath}
\end{lem}

\begin{proof}
  For $x,x' \in X$, if $\vec{x} \pijltau \vecp{x}$ then
  $\cevp{x} \pijltau \cev{x}$ by construction of~$\ltsT(\model)$.
  We verify the conclusion $\cev{x} \bis_b \cevp{x}$ of the lemma by
  induction on the depth~$\depth{\vec{x}}$.

  Basis, $\depth{\vec{x}}=0$: 
  Because $\depth{\vec{x}} = 0$ and $\vec{x} \pijltau \vecp{x}$, we
  have by~(iii) that $\vec{x} \equiv_\tau \vecp{x}$.
  By construction of~$\ltsT(\model)$, $\cevp{x} \equiv_\tau \cev{x}$,
  and therefore $\cevp{x} \bis_b \cev{x}$ by~(i).

  Induction step, $\depth{\vec{x}}>0$: 
  Assume $\vec{x} \pijltau \vecp{x}$ and $\vec{x} \bis_b \vecp{x}$.
  We have $\vec{x} \pijl{\eC} \cev{x}$ by construction
  of~$\ltsT(\model)$.
  Because $\vec{x} \bis_b \vecp{x}$, a matching computation
  for~$\vecxp$ exists,
  $\vecp{x} = \vec{u}_0 \pijltau {} \cdots \pijltau \vec{u}_n
  \pijl{\eC} \cev{u}_n$ say, where $n \geqslant 0$ and
  $u_0, \ldots, u_n \in X$ are such that $\vec{u}_{i} \bis_b \vec{x}$
  for $0 \leqslant i \leqslant n$ and $\cev{u}_n \bis_b \cev{x}$.
  (Note that $\vec{u}_0, \ldots, \vec{u}_n$ rather than
  $\cev{u}_0, \ldots, \cev{u}_n$ are involved, since
  $\tau$-transitions preserve direction, \ie, if $s \pijltau t$
  in~$\ltsT(\model)$ then either $s = \vec{x}, t = \vec{y}$ or
  $s = \cev{x}, t = \cev{y}$, for suitable $x,y \in X$.)
  Since $\vec{u}_0 \pijltau {} \cdots \pijltau \vec{u}_n$, it holds by~(iv) that
  $\depth{\vec{u}_0} \geqslant \cdots \geqslant \depth{\vec{u}_n}$.
  We distinguish two cases.

  Case~I, $\depth{\vec{u}_0} = \depth{\vec{u}_n}$:
  We have $\cev{u}_0 \equiv_\tau \cev{u}_n$ by~(ii).
  So,~$\vec{u}_0 \bis_b \vec{u}_n$ by~(i) and hence $\cev{u}_0 \bis_b
  \cev{u}_n$ by construction of~$\ltsT(\model)$.
  Using $\cevp{x} = \cev{u}_0$ and $\cev{u}_n \bis_b \cev{x}$ it
  follows that
  $\cevp{x} \bis_b \cev{u}_0 \bis_b \cev{u}_n \bis_b \cev{x}$.

  Case~II, $\depth{\vec{u}_0} \neq \depth{\vec{u}_n}$:
  We have $\depth{\vec{u}_0} > \depth{\vec{u}_n}$ by~(iv).
  Let~$k$, $0 \leqslant k < n$, be the minimal index such that
  $\depth{\vec{u}_k} > \depth{\vec{u}_{k{+}1}}$.
  By hypothesis $\vecx \pijltau \vecxp$.
  So,
  $\depth{\vecx} \geqslant \depth{\vecxp} = \depth{\vec{u}_0} =
  \depth{\vec{u}_k} > \depth{\vec{u}_\ell}$ for
  $k{+}1 \leqslant \ell \leqslant n$.
  By choice of~$u_n$ we have
  \begin{equation}
    \label{eq:41}
    \cevx \bis_b \cev{u}_n
  \end{equation}
  Because $\vec{u}_{k{+}1} \pijltau {} \cdots {} \pijltau \vec{u}_n$
  and $\vec{u}_{k{+}1} \bis_b {} \cdots {} \bis_b \vec{u}_n$, it
  follows by induction hypothesis that $\cev{u}_{k{+}1} \bis_b {}
  \cdots \bis_b {} \cev{u}_n$.
  Hence,
  \begin{equation}
    \label{eq:42}
    \cev{u}_n \bis_b \cev{u}_{k{+}1}
  \end{equation}
  Since $\vec{u}_k \bis_b \vec{u}_{k{+}1}$ and
  $\vec{u}_k \pijl{\eD} \cev{u}_k$, a matching sequence
  for~$\vec{u}_{k{+}1}$ exists, which is, by definition
  of~$\ltsT(\model)$, of the form
  $\vec{u}_{k{+}1} = \vec{v}_0 \pijltau {} \cdots {} \pijltau
  \vec{v}_m \pijl{\eD} \cev{v}_m$ for suitable $m \geqslant 0$ and
  $v_0, \ldots, v_m \in X$ such that
  $\vec{v}_0 , \ldots, \vec{v}_m \bis_b \vec{u}_k$ and
  $\cev{v}_m \bis_b \cev{u}_k$.
  From $\vec{u}_{k{+}1} \pijlstertau \vec{v}_j$ we get, by~(iv) that 
  $\depth{\vec{u}_{k{+}1}} \geqslant \depth{\vec{v}_j}$ for $ 0 \leqslant j \leqslant m$.
  Moreover, we know that 
  $\depth{\vec{x}} \geqslant \depth{\vec{u_0}} = \depth{\vec{u_k}} > \depth{\vec{u_{k{+}1}}}$.
 Thus, we get that
  $\depth{\vecx} \geqslant \depth{\vec{u}_0} = \depth{\vec{u}_k} >
  \depth{\vec{u}_{k{+}1}} \geqslant \depth{\vec{v}_j}$ for
  $0 \leqslant j \leqslant m$.
  By induction hypothesis we obtain $\cev{u}_{k{+}1} \bis_b \cev{v}_0
  , \ldots, \cev{v}_m$.
  In particular,
  \begin{equation}
    \label{eq:43}
    \cev{u}_{k{+}1} \bis_b \cev{v}_m
  \end{equation}
  By choice of~$v_m$ we have
  \begin{equation}
    \label{eq:44}
    \cev{v}_m \bis_b \cev{u}_k
  \end{equation}
  Finally, from
  $\depth{\vecxp} = \depth{\vec{u}_0} = \depth{\vec{u}_k}$ and
  $\vecxp = \vec{u}_0 \pijlstertau \vec{u}_k$ one derives
  $\vecxp \equiv_\tau \vec{u}_k$ using~(ii).
  Therefore, $\cevxp \equiv_\tau \cev{u}_k$, which implies
  \begin{equation}
    \label{eq:45}
    \cevxp \bis_b \cev{u}_k
  \end{equation}
  by~(i).
  Combining equations (\ref{eq:41}) to~(\ref{eq:45}) gives
  $\cevx \bis_b \cev{u}_n \bis_b \cev{u}_{k{+}1} \bis_b \cev{v}_m
  \bis_b \cev{u}_k \bis_b \cevxp$. So~$\cevx \bis_b \cevxp$, as was to
  be shown.
\end{proof}

\noindent
The next lemma states branching bisimilarity of pairs of forward
states and of backward states in an encoding \lts{} are related.

\begin{lem}
  \label{la:forwardbackwardequivalence}
  Let $\model = (X, \closure,\peval)$ be a finite \cm. Then
  $\vecx \bis_b \vecy \Leftrightarrow \cevx \bis_b \cevy$ in
  $\ltsT(\model)$ for $x, y \in X$.
\end{lem}
  
\begin{proof}
  ($\Rightarrow$) Let $x,y \in X$ such that $\vec{x} \bis_b \vec{y}$.
  For the transition $\vecx \pijl{\eC} \cevx$ of~$\vecx$, a matching
  computation of~$\vec{y}$ exists, which is of the form
  $\vecy = \vec{u}_0 \pijltau \vec{u}_1 \pijltau {} \cdots {} \pijltau
  \vec{u}_n \pijl{\eC} \cev{u}_n$ by definition of~$\ltsT(\model)$ and
  satisfies $\vec{u}_0, \ldots, \vec{u}_n \bis_b \vecx$ and
  $\cev{u}_n\bis_b \cevx$. Repeated application of
  Lemma~\ref{la:factI} yields
  $\cev{u}_n \pijltau {} \cdots {} \pijltau \cev{u}_0 = \cevy$ and
  $\cev{u}_n \bis_b \cevy$. Thus $\cevx \bis_b\cev{u}_n \bis_b \cevy$,
  as was to be shown.
  ($\Leftarrow$) Similar to the above.
\end{proof}

\noindent
The encoding of Definition~\ref{def:QdCM2LTS} preserves
\cop-bisimilarity of a \qdcm~$\model$ and reflects branching
bisimilarity of the \lts~$\ltsT(\model)$.

\begin{thm}
  \label{theo:CoPaBran}
  Let $\model = (X,\closure,\peval)$ be a finite \cm. For $x, y \in
  X$, it holds that
  \begin{displaymath}
    \text{$x \copbis y$ in~$\model$ iff
      $\vecx \bis_b \vecy$ in~$\ltsT(\model)$}
  \end{displaymath}
\end{thm}

\begin{proof}
  ($\Leftarrow$) Define the relation $B \subseteq {X \times X}$ by
  $x \mopB y$ if $\vecx \bisb \vecy$ for $x,y \in X$.  We verify
  that the relation~$B$ is a \cop-bisimulation
 using Lemma~\ref{lem-single-step-bisimilarity-suffices}.
  Suppose $x,y \in X$
  satisfy $x \mopB y$. Thus~$\vecx \bisb \vecy$.

%%% proposition condition %%%
  
  (i)~Let $p \in \ap$. If $x \in \peval(p)$, then $\vecx \pijlp \vecx$
  in~$\ltsT(\model)$.
  Since $\vecx \bisb \vecy$ in~$\ltsT(\model)$,
  exist $s,t \in S$
  such that $\vecy \pijlstertau s \pijlp t$ with $s,t \bisb
  \vec{x}$. By definition of~$\ltsT(\model)$, it must be that
  $s,t = \vecyp$ for some $y' \in \peval(p)$. Moreover,
  $\vecy \pijlstertau \vecyp$ implies, again by definition
  of~$\ltsT(\model)$, that $y =_\peval y'$. Hence~$y \in
  \peval(p)$. Symmetrically, $y \in \peval(p)$ implies $x \in
  \peval(p)$. So, we conclude $x \eqV y$.
  
%%% forward transfer condition %%%
  
  (ii)~Suppose $x' \in \closure(x)$ for some~$x' \in X$.
It suffices to show that a path
  $(y_i)_{i{=}0}^n$ from~$y$ exists such that $x \mopB y_i$ for $0
  \leqslant i < n$ and~$x' \mopB y_n$. The assumption
  $x' \in \closure(x)$ implies that $\ltsT(\model)$ has either the
  transition $\vecx \pijltau \vecxp$ or the transition $\vecx \pijlch
  \vecxp$.
  In the case of a $\tau$-transition,
  by branching bisimilarity of $\vecx$ and~$\vecy$,
  a computation $\vecy_0 \pijltau {} \cdots {} \pijltau \vecy_n$
  from~$\vecy$ with~$n \geqslant 0$ exists such that
  $\vecx \bis_b \vecy_0, \ldots, \vecy_{n{-}1}$
  and~$\vecxp \bis_b \vecy_n$.
  If~$n > 0$, by definition of~$\ltsT(\model)$, it holds that
  $y_{i{+}1} \in \closure(y_i)$ for $0 \leqslant i < n$.
  So, $(y_i)_{i{=}0}^n$ is a
  path from~$y$ in~$X$ as required.
  If~$n = 0$, we have both $x \mopB y$ and~$x' \mopB y$, and then the two
  elements path $(y,y)$ is a path from~$y$ in~$X$ as required.
  In the case of a $\ch$-transition,
  by branching bisimilarity of $\vecx$
  and~$\vecy$, a computation
  $y_0 \pijltau {} \cdots {} \pijltau y_{n{-}1} \pijlch y_n$
  from~$\vecy$ with~$n \geqslant 1$ exists such that
  $\vecx \bis_b y_0, \ldots, y_{n{-}1}$ and~$\vecxp \bis_b y_n$.
  Also here by definition of~$\ltsT(\model)$, $(y_i)_{i{=}0}^n$ is a
  path from~$y$ in~$X$ as required.

%%% backward transfer condition %%%
  
  (iii)~Now suppose $x \in \closure(x')$ for some~$x' \in X$.
We have to show that a path $(y_i)_{i{=}0}^n$ to~$y$ exists such that
  $x' \mopB y_0$ and $x \mopB y_i$ for $0 < i \leqslant n$.
  We will use that~$\cevx \bis_b \cevy$, which follows from~$\vecx
  \bis_b \vecy$ by Lemma~\ref{la:forwardbackwardequivalence}. It holds that in
  $\ltsT(\model)$ either the transition $\cevx \pijltau \cevxp$ or the
  transition $\cevx \pijlch \cevxp$ exists.
  In the case of a $\tau$-transition,
  by branching bisimilarity of $\cevx$ and~$\cevy$,
  $\ltsT(\model)$ has a computation
  $\cevy_0 \pijltau {} \cdots {} \pijltau \cevy_n$ from~$\cevy$
  with~$n \geqslant 0$ such that
  $\cevx \bis_b \cevy_0, \ldots, \cevy_{n{-}1}$ and
  $\cevxp \bis_b \cevy_n$.
  Lemma~\ref{la:forwardbackwardequivalence} yields
  $\vecx \bis_b \vecy_0, \ldots, \vecy_{n{-}1}$ and
  $\vecxp \bis_b \vecy_n$, thus $x \mopB y_i$ for $0 \leqslant i < n$ and
  $x' \mopB y_n$.
  Note, by definition of~$\ltsT(\model)$, $y_{i} \in
  \closure(y_{i{+}1})$ for $0 \leqslant i < n$.
  Therefore, if~$n > 0$ then, in reverse order, $(y_{n{-}i})_{i{=}0}^n$ is a
  path to~$y$ in~$X$ as required.
  If~$n = 0$, we have both $x' \mopB y$ and~$x \mopB y$.
  Consequently, the path $(y,y)$ is a path to~$y$ in~$X$ satisfying
  what is required.
  In the case of a $\ch$-transition,
  by branching bisimilarity of $\cevx$
  and~$\cevy$, a computation
  $\cevy_0 \pijltau {} \cdots {} \pijltau \cevy_{n{-}1} \pijlch
  \cevy_n$ from~$\cevy$ with~$n \geqslant 1$ exists such that
  $\cevx \bis_b \cevy_0, \ldots, \cevy_{n{-}1}$ and
  $\cevxp \bis_b y_n$.
  Again we obtain $\vecx \bis_b \vecy_0, \ldots, \vecy_{n{-}1}$ and
  $\vecxp \bis_b \vecy_n$, thus $x \mopB y_i$ for $0 \leqslant i < n$ and
  $x' \mopB y_n$ by application of Lemma~\ref{la:forwardbackwardequivalence}.
  Moreover, $y_{i} \in \closure(y_{i{+}1})$ for $0 \leqslant i < n$ by
  definition of~$\ltsT(\model)$.
  Thus, we have that, again in reversed order, $(y_{n{-}i})_{i{=}0}^n$
  is a path to~$y$ in~$X$ as required.

  \bigskip

  ($\Leftarrow$)
  We verify that the relation $B \subseteq S \times S$ given by
  $B = \ZET{ (\vecx, \vecy), (\cevx, \cevy)}{x \copbis y}$ is a
  branching bisimulation relation on~$\ltsT(\model)$. Suppose
  $x,y \in X$ satisfy $x \copbis y$.

  We check that each
  transition of~$\vecx$ can be matched by~$\vecy$.
  (i)~If $\vecx \pijl{p} \vecx$ for some proposition $p\in \ap$, then
  $x \in \peval(p)$. Hence, by \cop-bisimilarity, $y \in \peval(p)$.
  So, $\vecy \pijl{p} \vecy$, and the single-step computation
  $\vecy \pijl{p} \vecy$ in $\ltsT(\model)$ matches the transition
  $\vecx \pijl{p} \vecx$.
  (ii)~The transition $\vecx \pijl{\eC} \cevx$ is matched by the
  single-step computation $\vecy \pijl{\eC} \cevy$, by construction of~$\ltsT(\model)$.
  (iii)~If $\vecx \pijltau \vecxp$ for some~$x' \in X$, then
  $x' \in \closure(x)$ and $x \eqV x'$. Thus, $(x,x')$ is a
  path from~$x$ in~$X$. By \cop-bisimilarity of $x$ and~$y$, it
  follows that in~$X$ a path~$(y_i)_{i{=}0}^n$ from~$y$ exists
  such that $x \copbis y_0, \ldots, y_{n{-}1}$ and~$x' \copbis
  y_n$. Then we have $x \eqV y_0, \ldots, y_{n{-}1}$
  and also $x \eqV x' \eqV y_n$. Thus, both
  $y_{i{+}1} \in \closure(y_i)$ and $y_i \eqV y_{i{+}1}$ for
  $0 \leqslant i < n$. Hence, in $\ltsT(\model)$ we have a computation
  $\vecy_0 \pijltau {} \cdots {} \pijltau \vecy_n$ such that
  $\vecx \mopB \vecy_i$ for $0 \leqslant i < n$ and
  $\vecxp \mopB \vecy_n$. Therefore, $\vecy$~matches with this
  computation the $\tau$-transition of~$\vecx$.
  (iv)~Similarly, if $\vecx \pijl{\ch} \vecxp$ for some~$x' \in X$,
  then we find in~$\ltsT(\model)$ a computation
  $\vecy_0 \pijltau {} \cdots {} \pijltau \vecy_{n{-}1} \pijl{\ch}
  \vecy_n$ from~$\vecy$ which matches the $\ch$-transition of~$\vecx$.

  Next we check that each transition of~$\cevx$ can be matched
  by a computation of~$\cevy$.
  (i)~The transition $\cevx \pijl{\eD} \vecx$ is matched by the
  single-step computation $\cevy \pijl{\eD} \vecy$.
  (ii)~If $\cevx \pijltau \cevxp$ for some~$x' \in X$, a reasoning
  similar to the corresponding case for~$\vecx$ applies. By definition
  of~$\ltsT(\model)$, it holds that $x \in \closure(x')$ and
  $x \eqV x'$. Thus, $(x',x)$ is a path to~$x$ in~$X$. By
  \cop-bisimilarity of $x$ and~$y$, it follows that in~$X$ a
  path~$(y_i)_{i{=}0}^n$ to~$y$ exists that satisfies $x' \copbis y_0$
  and~$x \copbis y_1, \ldots, y_n$. Then we have $x \eqV x' \eqV y_0$
  and also $x \eqV y_i$ for $0 < i \leqslant n$. Thus, both
  $y_{i{+}1} \in \closure(y_{i})$ and $y_i \eqV y_{i{+}1}$ for
  $0 \leqslant i < n$. Hence, by definition of~$\ltsT(\model)$, a
  computation $\cevy_n \pijltau {} \cdots {} \pijltau \cevy_0$
  of~$\cevy$ exists which satisfies $\cevx \mopB \cevy_i$ for
  $n \geqslant i > 0$ and $\cevxp \mopB \cevy_0$. Therefore, this
  computation of~$\cevy$ matches the $\tau$-transition of~$\cevx$.
  (iii)~Finally, a transition $\cevx \pijl{\ch} \cevxp$ for
  some~$x' \in X$ is matched by~$\cevy$.
  In fact, if $\cevx \pijl{\ch} \cevxp$ for some~$x' \in X$, then it holds that
  $x \in \closure(x')$ and~$x \neqV x'$. Thus $(x',x)$ is a path
  to~$x$ in~$\model$. Let by \cop-bisimilarity of $x$ and~$y$, the
  path $( y_i)_{i{=}0}^n$ be a matching path to~$y$.
  Note~$n \geqslant 1$ as $x \not\copbis x'$. We have
  $x' \copbis y_0$ and $x \copbis y_1, \ldots, y_n$. Thus,
  $y_0 \eqV x' \neqV x \eqV y_1$ and $y_1 \eqV \cdots \eqV y_n$. So,
  in~$\ltsT(\model)$ the computation $\cevy_n \pijltau {} \cdots
  \pijltau \cevy_1 \pijl{\ch} \cevy_0$ exists that
  satisfies $x \mopB y_i$ for $n \geqslant i > 0$
  and~$x' \mopB y_0$. We see, this computation of~$\cevy$ matches the
  $\ch$-transition of~$\cevx$ as was to be shown.
\end{proof}

%\begin{figure}
%  \centering
%  \small
%  \input{fig-encoding-example-minimized}
%  \caption{Minimized \lts{} for the \lts~of Figure~\ref{fig:LTSQdCM1}b; 
% % (b)~the corresponding~\qdcm\@.
%  }
%  \label{fig:LTSMin}
%\end{figure}

\subsection{\lts\ encoding  of symmetric \qdcm{s}}
\label{sec:OptimisedEncoding}

\mbox{} \smallskip

\noindent
For finite \cm{s} that are symmetric a simplified version of the
encoding can be given. These \cm{s} naturally arise as representations
of digital images where points are related via an adjacency relation
as discussed in Section~\ref{sec:introduction}. Because of the
relevance these instances we discuss below a concise version of the encoding
of the previous subsection.

\begin{defi}[Encoding, symmetric case]
\label{def:SQdCM2LTS}
Let $\model = (X,\closure,\peval)$ be a symmetric \cm.
Define the labelled transition system
$\ltsTS(\model) = (X,\calA,{\rightarrow})$ by
\begin{enumerate} [label=(\roman*)]
%\item $S=\ZET{\vecev{s}}{s\in X}$;
\item $\act = \ap \cup \SET{\tau,\ch}$;
\item the transition relation~${\rightarrow} \subseteq {X \times \act
    \times X}$ contains exactly the
  following transitions:
\def\arraystretch{1.2}
\begin{displaymath}
\begin{array}{cl}
x \pijl{p} x & \text{for $x \in X$, $p \in \ap$ with $x \in \peval(p)$} \\
x \pijl{\tau} x' & \text{if $x' \in \closure(x) \backslash
                   \SET{x}$ and $x \eqV x'$} \\
x \pijl{\ch} x' & \text{if $x' \in \closure(x)$ and $x \neqV x'$} \\
\end{array}
\end{displaymath}
\def\arraystretch{1.0}
\end{enumerate}
\closedefi
\end{defi}

\input{fig-QdCM2}

\noindent
As an example, consider the symmetric finite \qdcm{} of Figure~\ref{fig:QdCM2} and its LTS encoding in Figure~\ref{fig:LTSQdCM2}, obtained with the encoding given in Definition~\ref{def:SQdCM2LTS}. 
The minimised model of the symmetric \qdcm\ of Figure~\ref{fig:QdCM2} is shown in Figure~\ref{fig:QdCM2min}. The equivalence classes are the same as in the example of Figure~\ref{fig:QdCM1} and for the same reasons.
  
\input{fig-LTS-QdCM2}
\input{fig-QdCM2min}
  
\begin{thm}\label{theo:SCoPaBran} 
  Let $\model =(X,\closure,\peval)$ be a finite and symmetric \cm. For 
  $x, y \in X$ it holds that
  \begin{displaymath}
    \text{$x \copbis y$ in~$\model$ iff
      $x \bis_b y$ in~$\ltsTS(\model)$}
  \end{displaymath}
\end{thm}

\begin{proof}
  We prove the statement $\vec{x} \bis_b\vec{y}$ in~$\ltsT(\model)$
  iff $x \bis_b y$ in~$\ltsTS(\model)$ for~$x,y \in X$. Together with
  Theorem~\ref{theo:CoPaBran}, leads to the assertion.

  ($\Leftarrow$) We verify that the relation~$B$ given by
  \begin{math}
    B = \ZET{
      \langle \vec{x},\vec{y} \,\rangle
      ,\, \langle \cev{x},\cev{y} \,\rangle} 
    {x \bis_b y} 
  \end{math}
  is a branching bisimulation relation, where branching bisimilarity
  of $x$ and~$y$ is considered in~$\ltsTS(\model)$.
  Let $x, y \in X$ be such that $x \bis_b y$. Then $\vecx \mopB \vecy$
  and $\cevx \mopB \cevy$. We first analyze the transitions of~$\vecx$
  in~$\ltsT(\model)$.
  \begin{enumerate}[label=(\roman*)]
  \item Regarding the transition $\vecx \pijl{\eC} x'$ in
    $\ltsT(\model)$ we note that $\vecy \pijl{\eC} \cevy$ and that
    $\cevx \mopB \cevy$.
  \item If $\vecx \pijl{p} \vecy$ in $\ltsT(\model)$ then
    $x \in \peval(p)$ for $p \in \ap$.  Also, $x \pijl{p} x$ in
    $\ltsTS(\model)$.  Because $x \bis_b y$, we have
    $y = y_0 \pijl{\tau} \cdots \pijl{\tau} y_{n{-}1} \pijl{p} y_n$
    for some $n > 0$ and $y_0, \ldots, y_n \in X$ with
    $x \bis_b y_0, \ldots, y_{n}$.
      Because $y_{n{-}1} \in \peval(p)$ it follows from the
      construction of~$\ltsTS(\model)$ that $y_{n{-}1} = y_n$ and
      $y_{n{-}1} \in \peval(p)$. As $y_0 \pijl{\tau} \cdots
      \pijl{\tau} y_{n{-}1}$ it follows from the
      construction of~$\ltsTS(\model)$ as well that $y_0 \eqV \cdots
      \eqV y_{n{-}1}$. 
    We conclude that $y = y_0 \in \peval(p)$ and therefore
    $\vecy \pijl{p} \vecy$ by construction of $\ltsT(\model)$.
  \item If $\vecx \pijl{\tau} \vecxp$ in $\ltsT(\model)$ for
    some~$x'\in X$, then it holds that $x = x'$, or $x' \in
    \closure(x) \backslash \SET{x}$ and $x \eqV x'$. 
    In the first case, i.e.\ $x = x'$, we have
    $\vecxp \mathop{=} \vecx \mopB \vecy$. 
    In the second case, $x \pijl{\tau} x'$ in
    $\ltsTS(\model)$. As~$x \bis_b y$, a computation
    $y = y_0 \pijl{\tau} \cdots \pijl{\tau} y_{n{-}1} \pijl{\tau} y_n$
    exists in $\ltsTS(\model)$ with $n \geqslant 0$,
    $y_0, \ldots, y_{n{-}1} \bis_b x$, and~$y_n \bis_b x'$. We have,
    without loss of generality,
    $y_{i{+}1} \in \closure(y_i) \backslash \SET{y_i}$ and
    $y_i \eqV y_{i{+}1}$ for $0 \leqslant i < n$. Thus,
    $\vecy = \vecy_0 \pijl{\tau} \cdots \pijl{\tau} \vecy_{n{-}1}
    \pijl{\tau} \vecy_n$ is a computation in $\ltsT(\model)$, and, by
    definition of~$B$, it holds that
    $\vecx \mopB \vecy_0, \ldots, \vecy_{n{-}1}$, and
    $\vecxp \mopB \vecy_n$.
    \item If $\vecx \pijl{\ch} \vecxp$ in $\ltsT(\model)$ for
    some~$x'\in X$, then it holds that $x' \in \closure(x)$ and
    $x \neqV x'$. So, $x \pijl{\ch} x'$ in $\ltsTS(\model)$.
    As~$x \bis_b y$, a computation
    $y = y_0 \pijl{\tau} \cdots \pijl{\tau} y_{n{-}1} \pijl{\ch} y_n$
    exists in $\ltsTS(\model)$ with $n > 0$,
    $y_0, \ldots, y_{n{-}1} \bis_b x$, and~$y_n \bis_b x'$. We have,
    without loss of generality,
    $y_{i{+}1} \in \closure(y_i) \backslash \SET{y_i}$ for
    $0 \leqslant i < n$ and
    $y_0 \eqV \cdots \eqV y_{n{-}1} \neqV y_n$. Thus,
    $\vecy = \vecy_0 \pijl{\tau} \cdots \pijl{\tau} \vecy_{n{-}1}
    \pijl{\ch} \vecy_n$ is a computation in $\ltsT(\model)$, and, by
    definition of~$B$, it holds that
    $\vecx \mopB \vecy_0, \ldots, \vecy_{n{-}1}$, and
    $\vecxp \mopB \vecy_n$.
  \end{enumerate}
  Next, we consider transitions of~$\cevx$ in $\ltsT(\model)$.
  \begin{enumerate}[label=(\roman*)]
  \item Regarding the transition $\cevx \pijl{\eD} \cevx$ in
    $\ltsT(\model)$ we note that $\cevy \pijl{\eD} \cevy$ and
    $\cevx \mopB \cevy$.
  \item If $\cevx \pijl{\tau} \cevxp$ in $\ltsT(\model)$ for
    some~$x' \in X$, then it holds that $x = x'$, or $x \in
    \closure(x') \backslash \{x\}$
    and~$x \eqV x'$. 
    Thus, $x = x'$ or $x \pijl{\tau} x'$ in $\ltsTS(\model)$,
    by symmetry of~$\model$. 
    In case $x = x'$ we are done since
      $\cevxp \mathop{=} \cevx \mopB \cevy$. In the other case, we
      reason as follows.
    Since $x \bis_b y$, a computation $y = y_0 \pijl{\tau} \cdots
    \pijl{\tau} y_{n{-}1} \pijl{\tau} y_n$ 
    exists in $\ltsTS(\model)$ with $n \geqslant 0$,
    $y_0, \ldots, y_{n{-}1} \bis_b x$, and $y_n \bis_b x'$.
    From $y_i \pijl{\tau} y_{i{+}1}$ in $\ltsTS(\model)$, for $0 \leqslant i < n$, we have, by construction of
    $\ltsTS(\model)$ (see Definition~\ref{def:SQdCM2LTS}), that $ y_{i{+}1} \in \closure(y_i)$. %,  y_{i{+}1} \not= y_i$ and $y_i \eqV y_{i{+}1}$.
    In addition, by symmetry of $\model$, we get that also $ y_{i} \in \closure(y_{i{+}1})$ holds for $0 \leqslant i < n$.
    Thus, by construction of $\ltsT(\model)$ (see Definition~\ref{def:QdCM2LTS}), we get that $\vecy_{i{+}1} \pijl{\tau} \vecy_i$ 
    is a transition of $\ltsT(\model)$, for $0 \leqslant i < n$ and then 
    $\vecy_n \pijl{\tau} \vecy_{n{-}1} \pijl{\tau} \cdots \pijl{\tau} \vecy_0$ is a computation
    of $\ltsT(\model)$. By construction of of $\ltsT(\model)$ we also get that 
     $\cevy = \cevy_0 \pijl{\tau} \cdots \pijl{\tau} \cevy_{n{-}1} \pijl{\tau} \cevy_n$ is a computation in $\ltsT(\model)$.
%    Since $x \bis_b y$, a computation
%    $y = y_0 \pijl{\tau} \cdots \pijl{\tau} y_{n{-}1} \pijl{\tau} y_n$
%    exists in $\ltsTS(\model)$ with $n \geqslant 0$,
%    $y_0, \ldots, y_{n{-}1} \bis_b x$, and $y_n \bis_b x'$.
%%
%    \evreplaced{Reversing the transitions}{Taking transitions in the
%      other direction}, as allowed by Definition~\ref{def:SQdCM2LTS},
%    yields that
%    $y_n \pijl{\tau} y_{n{-}1} \pijl{\tau} \cdots \pijl{\tau} y_0$ is
%    also a computation in $\ltsTS(\model)$.
%%
%    Hence,
%    $y_{i{-}1} \in \closure(y_i)$ and $y_i \eqV y_{i{-}1}$ for
%    $n \geqslant i > 0$. Thus
%    $\vecy_n \pijl{\tau} \vecy_{n{-}1} \pijl{\tau} \cdots \pijl{\tau}
%    \vecy_0$ is a computation in $\ltsT(\model)$. By again reversing
%    the transitions as allowed by Definition~\ref{def:QdCM2LTS}, we
%    obtain that
%    $\cevy = \cevy_0 \pijl{\tau} \cdots \pijl{\tau} \cevy_{n{-}1}
%    \pijl{\tau} \cevy_n$ is a computation in $\ltsT(\model)$. 
 %   
    Because $y_0, \ldots, y_{n{-}1} \bis_b x$ and $y_n \bis_b x'$, we have
    $\cevy_0, \ldots, \cevy_{n{-}1} \mopB \cevx$ and
    $\cevy_n \mopB \cevx'$ and we are done.   
    
  \item
    If $\cevx \pijl{\ch} \cevxp$ in $\ltsT(\model)$ for
    some~$x' \in X$, we proceed in a similar way as in the previous case.
    It holds that $x \in \closure(x')$, $x \neq x'$, and~$x \neqV x'$. 
    By symmetry of~$\model$ it follows
    that $x' \in \closure(x)$. Thus, $x \pijl{\ch} x'$ in
    $\ltsTS(\model)$. 
    
    Since $x \bis_b y$, a computation $y = y_0 \pijl{\tau} \cdots \pijl{\tau} y_{n{-}1} \pijl{\ch} y_n$
    exists in $\ltsTS(\model)$ with $n \geqslant 0$,
    $y_0, \ldots, y_{n{-}1} \bis_b x$, and $y_n \bis_b x'$.
    From $y_i \pijl{\tau} y_{i{+}1}$ in $\ltsTS(\model)$, for $0 \leqslant i < n$, we have, by construction of
    $\ltsTS(\model)$ (see Definition~\ref{def:SQdCM2LTS}), that $ y_{i{+}1} \in \closure(y_i)$. %,  y_{i{+}1} \not= y_i$ and $y_i \eqV y_{i{+}1}$.
    In addition, by symmetry of $\model$, we get that also $ y_{i} \in \closure(y_{i{+}1})$ holds for $0 \leqslant i < n$.
    Thus, by construction of $\ltsT(\model)$ (see Definition~\ref{def:QdCM2LTS}), 
    we get that $\vecy_{i{+}1} \pijl{\tau} \vecy_i$ 
    is a transition of $\ltsT(\model)$, for $0 \leqslant i < n-1$ as well as $\vecy_{n} \pijl{\ch} \vecy_{n{-}1}$.
    
    Consequently,
    $\vecy_n \pijl{\ch} \vecy_{n{-}1} \pijl{\tau} \cdots \pijl{\tau} \vecy_0$ is a computation
    of $\ltsT(\model)$. By construction of of $\ltsT(\model)$ we also get that 
     $\cevy = \cevy_0 \pijl{\tau} \cdots \pijl{\tau} \cevy_{n{-}1} \pijl{\ch} \cevy_n$ is a computation in $\ltsT(\model)$.  
    Because $y_0, \ldots, y_{n{-}1} \bis_b x$ and $y_n \bis_b x'$, we have
    $\cevy_0, \ldots, \cevy_{n{-}1} \mopB \cevx$ and
    $\cevy_n \mopB \cevx'$ and we are done.
%  
%    If $\cevx \pijl{\ch} \cevxp$ in $\ltsT(\model)$ for
%    some~$x' \in X$, we reason similarly to the previous case. It
%    holds that $x \in \closure(x')$, $x \neq x'$, and~$x \neqV x'$. By
%    symmetry of~$\model$ it follows that $x' \in \closure(x)$. Thus,
%    $x \pijl{\ch} x'$ in $\ltsTS(\model)$. Since $x \bis_b y$, a
%    computation
%    $y = y_0 \pijl{\tau} \cdots \pijl{\tau} y_{n{-}1} \pijl{\ch} y_n$
%    exists in $\ltsTS(\model)$ with $n > 0$,
%    $y_0, \ldots, y_{n{-}1} \bis_b x$, and $y_n \bis_b x'$.
%%
%    \evreplaced{Reversing the transitions}{Taking transitions in the
%      other direction}, as allowed by Definition~\ref{def:SQdCM2LTS},
%    yields that
%    $y_n \pijl{\ch} y_{n{-}1} \pijl{\tau} \cdots \pijl{\tau} y_0$ is a
%    computation in $\ltsTS(\model)$.
%%
%    Hence, it holds that
%    $y_{i{-}1} \in \closure(y_i)$ for $n \geqslant i > 0$,
%    $y_n \neqV y_{n{-}1}$, and $y_i \eqV y_{i{-}1}$ for $n > i >
%    0$. Thus
%    $\vecy_n \pijl{\tau} \vecy_{n{-}1} \pijl{\tau} \cdots \pijl{\tau}
%    \vecy_0$ is a computation in $\ltsT(\model)$. Now reversing
%    transitions as supported by Definition~\ref{def:QdCM2LTS}, we
%    obtain that
%    $\cevy = \cevy_0 \pijl{\tau} \cdots \pijl{\tau} \cevy_{n{-}1}
%    \pijl{\ch} \cevy_n$ is a computation in $\ltsT(\model)$. Because
%    $y_0, \ldots, y_{n{-}1} \bis_b x$ and $y_n \bis_b x'$, we also have
%    here that $\cevy_0, \ldots, \vecy_{n{-}1} \bis_b \cevx$ and
%    $\cevy_n \bis_b \cevx'$ and we are done.
  \end{enumerate}
  This finishes the proof from right to left. 

  ($\Rightarrow$) We verify that the relation
  $B = \ZET{\langle x, y \rangle}{\text{$\vecx \bis_b \vecy$ in
      $\ltsT(\model)$}}$ is a branching bisimulation for
  $\ltsTS(\model)$. So, let $x,y \in X$ such that $x \mopB y$.
  \begin{enumerate} [label=(\roman*)]
  \item If $x \pijl{p} x$ in $\ltsTS(\model)$ for $p \in \ap$, then
    $x \in \peval(p)$. Therefore, $\vecx \pijl{p} \vecx$ in
    $\ltsT(\model)$. We have $\vecx \bis_b \vecy$ in
    $\ltsT(\model)$. So, we can find a matching computation
    $\vecy = \vecy_0 \pijl{\tau} \cdots \pijl{\tau} \vecy_{n{-}1}
    \pijl{p} \vecy_n$ for~$\vecy$ with $n \geqslant 0$ and
    $\vecy_0, \ldots, \vecy_n \bis_b \vecx$. Apparently in view of
    Definition~\ref{def:QdCM2LTS}, $\vecy_{n{-}1} = \vecy_n$, and
    therefore $y_{n{-}1} \in \peval(p)$. As above, an inductive
    argument shows that $y_{n{-}1}, \ldots, y_0 \in \peval(p)$. Hence,
    $y = y_0 \in \peval(p)$. We conclude that $y \pijl{p} y$ in
    $\ltsTS(\model)$ and this computation for~$y$ matches the
    transition $x \pijl{p} x$.
  \item If $x \pijl{\tau} x'$ in $\ltsTS(\model)$ for some~$x' \in X$,
    then $x' \in \closure(x)$ and $x \eqV x'$. Thus, according to
    Definition~\ref{def:QdCM2LTS}, $\vecx \pijl{\tau} \vecxp$ is a
    transition of $\ltsT(\model)$. Let
    $\vecy = \vecy_0 \pijl{\tau} \cdots \pijl{\tau} \vecy_{n{-}1}
    \pijl{\tau} \vecy_n$ for~$\vecy$ with $n \geqslant 0$,
    $\vecy_0, \ldots, \vecy_{n{-}1} \bis_b \vecx$, and
    $\vecy_n \bis_b \vecxp$ be a matching computation of~$\vecy$ in
    $\ltsT(\model)$. Thus, $y_{i{+}1} \in \closure(y_i)$ and
    $y_i \eqV y_{i{+}1}$ for $0 \leqslant i < n$. Following
    Definition~\ref{def:SQdCM2LTS} we conclude that
    $y = y_0 \pijl{\tau} \cdots \pijl{\tau} y_{n{-}1} \pijl{\tau} y_n$
    is a computation of~$y$ in $\ltsTS(\model)$. Moreover,
    $y_0, \ldots , y_{n{-}1} \mopB x$ and~$y_n \mopB x'$. So, this
    computation of~$y$ matches the transition $x \pijl{\tau} x'$
    of~$x$.
  \item If $x \pijl{\ch} x'$ in $\ltsTS(\model)$ for some~$x' \in X$,
    then $x' \in \closure(x)$ and $x \neqV x'$. Therefore
    $\vecx \pijl{\ch} \vecxp$ is a transition of~$\vecx$ in
    $\ltsT(\model)$. Let
    $\vecy = \vecy_0 \pijl{\tau} \cdots \pijl{\tau} \vecy_{n{-}1}
    \pijl{\ch} \vecy_n$ for~$\vecy$ with $n \geqslant 0$,
    $\vecy_0, \ldots, \vecy_{n{-}1} \bis_b \vecx$, and
    $\vecy_n \bis_b \vecxp$ be a matching computation of~$\vecy$ in
    $\ltsT(\model)$. Then it holds that $y_{i{+}1} \in \closure(y_i)$ and
    for $0 \leqslant i < n$, $y_i \eqV y_{i{+}1}$ for $0 \leqslant i <
    n{-}1$, and~$y_{n{-}1} \neqV y_n$. Thus, according to
    Definition~\ref{def:SQdCM2LTS}, we have that
    $y = y_0 \pijl{\tau} \cdots \pijl{\tau} y_{n{-}1} \pijl{\ch} y_n$
    is a computation of~$y$ in $\ltsTS(\model)$. Also we have that
    $y_0, \ldots , y_{n{-}1} \mopB x$ and~$y_n \mopB x'$. Thus, the
    computation of~$y$ is matching the transition $x \pijl{\tau} x'$
    of~$x$ as was to be shown.
  \end{enumerate}
\end{proof}

\begin{figure}[t!]
        \centering
        \subfloat[][]
        {
                \begin{minipage}{2in}
                        \raisebox{0.385in}{\includegraphics[height=2in]{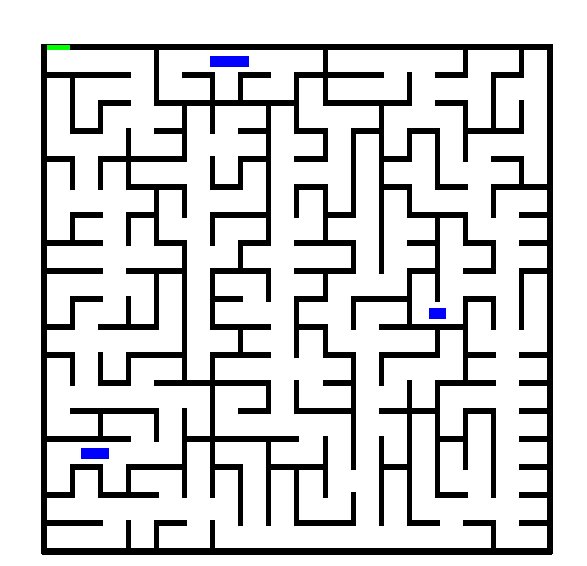}}
                \end{minipage}
                \label{fig:Maze}
        }\quad\quad
        \centering
        \subfloat[][]
        {
                \begin{minipage}{2in}                        
                        \begin{tikzpicture}
                                \tikzset{<->}
                                \tikzstyle{point}=[circle,draw=black,fill=white,inner sep=0pt,minimum width=4mm]
                                \begin{scope}[scale=0.6, transform shape]
                                        \node (p1)[point,draw=black,fill=blue] at (0,0) {};
                                        \node (p2)[point,draw=black,fill=white] at (2,0) {};
                                        \node (p3)[point,draw=black,fill=black] at (4,0) {};
                                        \node (p4)[point,draw=black,fill=white] at (6,0) {};
                                        \node (p5)[point,draw=black,fill=blue] at (8,0) {};
                                        \node (p6)[point,draw=black,fill=green] at (3,2.5) {};
                                        \node (p7)[point,draw=black,fill=white] at (4,-2) {};

                                        \node (p1c)[point,draw=black,fill=lightgray] at (1,1) {};
                                        \node (p2c)[point,draw=black,fill=lightgray] at (3,1) {};
                                        \node (p3c)[point,draw=black,fill=lightgray] at (5,1) {};
                                        \node (p4c)[point,draw=black,fill=lightgray] at (7,1) {};
                                        \node (p5c)[point,draw=black,fill=lightgray] at (9,1) {};
                                        \node (p6c)[point,draw=black,fill=lightgray] at (4,3.5) {};
                                        \node (p7c)[point,draw=black,fill=lightgray] at (5,-1) {};
                                        \draw [thick] (p1) edge[below,thick] node{$\ch$} (p2);
                                        %\draw [thick] (p1) edge[below,thick] node{$cn$} (p1c); 
                                        %\draw [thick] (p2) edge[below,thick,bend left] node{$ch$} (p1); 

                                        \draw [thick] (p2) edge[below,thick] node{$\ch$} (p3);
                                        %\draw [thick] (p3) edge[below,thick,bend left] node{$ch$} (p2); 

                                        \draw [thick] (p3) edge[below,thick,pos=0.70] node{$\ch$} (p4);
                                        %\draw [thick] (p4) edge[below,thick,bend left] node{$ch$} (p3); 

                                        \draw [thick] (p4) edge[below,thick] node{$\ch$} (p5);
                                        %\draw [thick] (p5) edge[below,thick,bend left] node{$ch$} (p4);

                                        \draw [thick, bend left] (p2) edge [left,thick, pos=0.80] node{$\ch$} (p6);
                                        %\draw [thick] (p6) edge[right,thick,bend left] node{$ch$} (p2);

                                        \draw [thick] (p3) edge[left,thick] node{$\ch$} (p7);
                                        %\draw [thick] (p7) edge[left,thick,bend left] node{$ch$} (p3);

                                        \draw [thick, bend right=20] (p3) edge[right,thick,pos=0.90] node{$\ch$} (p6);
                                        %\draw [thick] (p6) edge[right,thick,bend left] node{$ch$} (p3);

                                        \draw [thick] (p1c) edge[above,thick,pos=0.30] node{$\ch$} (p2c);
                                        \draw [thick] (p2c) edge[above,thick,pos=0.70] node{$\ch$} (p3c);
                                        \draw [thick] (p3c) edge[above,thick] node{$\ch$} (p4c);
                                        \draw [thick] (p4c) edge[above,thick] node{$\ch$} (p5c);

                                        \draw [thick,bend right] (p2c) edge[left,thick,pos=0.30] node{$\ch$} (p6c);
                                        \draw [thick,bend right=20] (p3c) edge[right,thick] node{$\ch$} (p6c);

                                        \draw [thick] (p3c) edge[left,thick,pos=0.70] node{$\ch$} (p7c);

                                        \node (p6d) [point,draw=black,fill=green] at (3,2.5) {};
                                        
                                        \tikzset{->}

                                        \draw [thick] (p1c) edge[right,thick, bend left] node{$\eD$} (p1);
                                        \draw [thick] (p1) edge[left,thick, bend left] node{$\eC$} (p1c);

                                        \draw [thick] (p2c) edge[right,thick, bend left] node{$\eD$} (p2);
                                        \draw [thick] (p2) edge[left,thick, bend left,pos=0.70] node{$\eC$} (p2c);

                                        \draw [thick] (p3c) edge[left,thick, bend left,pos=0.30] node{$\eD$} (p3);
                                        \draw [thick] (p3) edge[left,thick, bend left,pos=0.70] node{$\eC$} (p3c);

                                        \draw [thick] (p4c) edge[right,thick, bend left] node{$\eD$} (p4);
                                        \draw [thick] (p4) edge[left,thick, bend left] node{$\eC$} (p4c);

                                        \draw [thick] (p5c) edge[right,thick, bend left] node{$\eD$} (p5);
                                        \draw [thick] (p5) edge[left,thick, bend left] node{$\eC$} (p5c);

                                        \draw [thick] (p6c) edge[right,thick, bend left, pos=0.80] node{$\,\eD$} (p6);
                                        \draw [thick] (p6) edge[left,thick, bend left] node{$\eC$} (p6c);

                                        \draw [thick] (p7c) edge[right,thick, bend left] node{$\eD$} (p7);
                                        \draw [thick] (p7) edge[left,thick, bend left,pos=0.80] node{$\eC$} (p7c);
                                \end{scope}

                        \end{tikzpicture}  \\
                        \begin{tikzpicture}
                                \tikzset{<->}
                                \tikzstyle{point}=[circle,draw=black,fill=white,inner
                                sep=0pt, minimum width=4mm]
                                \begin{scope}[scale=0.6, transform shape]
                                        \node (p1)[point,draw=black,fill=blue] at (0,0) {};
                                        \node (p2)[point,draw=black,fill=white] at (2,0) {};
                                        \node (p3)[point,draw=black,fill=black] at (4,0) {};
                                        \node (p4)[point,draw=black,fill=white] at (6,0) {};
                                        \node (p5)[point,draw=black,fill=blue] at (8,0) {};
                                        \node (p6)[point,draw=black,fill=green] at (3,2) {};
                                        \node (p7)[point,draw=black,fill=white] at (4,-2) {};
                                        \draw [thick] (p1) edge[below,thick] node{$\ch$} (p2);
                                        % \draw [thick] (p2) edge[below,thick,bend left] node{$ch$} (p1); 

                                        \draw [thick] (p2) edge[below,thick] node{$\ch$} (p3);
                                        %\draw [thick] (p3) edge[below,thick,bend left] node{$ch$} (p2); 

                                        \draw [thick] (p3) edge[below,thick] node{$\ch$} (p4);
                                        %\draw [thick] (p4) edge[below,thick,bend left] node{$ch$} (p3); 

                                        \draw [thick] (p4) edge[below,thick] node{$\ch$} (p5);
                                        % \draw [thick] (p5) edge[below,thick,bend left] node{$ch$} (p4);

                                        \draw [thick,bend left] (p2) edge[left,thick] node{$\ch$} (p6);
                                        % \draw [thick] (p6) edge[right,thick,bend left] node{$ch$} (p2);

                                        \draw [thick] (p3) edge[left,thick] node{$\ch$} (p7);
                                        %\draw [thick] (p7) edge[left,thick,bend left] node{$ch$} (p3);

                                        \draw [thick,bend right] (p3) edge[right,thick] node{$\ch$} (p6);
                                        % \draw [thick] (p6) edge[right,thick,bend left] node{$ch$} (p3);
                                \end{scope}
                        \end{tikzpicture}
                \end{minipage}
                \label{fig:MiniMaze}
        }
        \caption{An image of a 2D maze, with green exit (upper-left), blue starting points, black walls and white walking areas (Figure~\ref{fig:Maze}),
                its minimal \lts{} using the general encoding of Definition~\ref{def:QdCM2LTS} (Figure~\ref{fig:MiniMaze} - top), and that obtained using the optimised encoding of Definition~\ref{def:SQdCM2LTS}
                (Figure~\ref{fig:MiniMaze} - bottom). 
                For readability, self-loops labelled by atomic propositions are not shown; the corresponding states are shown in the colour represented by the omitted label; symmetric transition pairs are drawn as doubly-headed arrows.
        }
        \label{fig:MazeMinimisation}
\end{figure}

\noindent
The 2D maze in Figure~\ref{fig:Maze}, which will also be part of our feasibility study in Section~\ref{sec:FisStudy}, exemplifies the significance of \cop-minimization on images. Each state of the \lts{} represents an area of interest in the image: exit (green), walls (black), walking areas (white) and starting points (blue). The three white states, as an example, represent three different kinds of white walking areas: the ones from which neither an exit nor a starting point can be reached (without crossing walls), the ones from which a starting point can be reached (but not the exit), and the ones from which a starting point and the exit can be reached.

%% file: fig-encoding-example.tex
  \begin{tikzpicture}[%
    every state/.append style={label distance=-1mm},
    every loop/.append style={out=120, in=60, looseness=8} ]

    \node [draw=none] at (0,1.5) {(a)} ;
    
    \node [state, label=60:{$r$}] (s1) at (1,1) {$x_1$} ;
    \node [state, label=60:{$r$}] (s2) at (2.75,1) {$x_2$} ;
    \node [state, label=60:{$g$}] (s3) at (4.5,1) {$x_3$} ;
    \node [state, label=60:{$b$}] (s4) at (6.25,1) {$x_4$} ;
    \node [state, label=60:{$r$}] (s5) at (8.5,1) {$x_5$} ;
    \node [state, label=60:{$g$}] (s6) at (10.25,1) {$x_6$} ;
    \draw (s1) edge (s2) ;
    \draw (s2) edge (s3) ;
    \draw (s4) edge (s3) ;
    \draw (s5) edge (s6) ;            
    \begin{scope}[shift={(0,-4.25)},
      every state/.append style={inner sep=1pt}]

      \node [draw=none] at (0,3.75) {(b)} ;

      \node[state] (s1D) at (1,3) {$\vec{x}_1$} ;
      \node[state] (s2D) at (2.75,3) {$\vec{x}_2$} ;
      \node[state] (s3D) at (4.5,3) {$\vec{x}_3$} ;
      \node[state] (s4D) at (6.25,3) {$\vec{x}_4$} ;
      \node[state] (s5D) at (8.5,3) {$\vec{x}_5$} ;
      \node[state] (s6D) at (10.25,3) {$\vec{x}_6$} ;
      \node[state, below of=s1D] (s1C) {$\cev{x}_1$} ;
      \node[state, below of=s2D] (s2C) {$\cev{x}_2$} ;
      \node[state, below of=s3D] (s3C) {$\cev{x}_3$} ;
      \node[state, below of=s4D] (s4C) {$\cev{x}_4$} ;
      \node[state, below of=s5D] (s5C) {$\cev{x}_5$} ;
      \node[state, below of=s6D] (s6C) {$\cev{x}_6$} ;
      \draw (s1D) edge [loop] node [above] {$r$} (s1D) ;
      \draw (s2D) edge [loop] node [above] {$r$} (s2D) ;
      \draw (s3D) edge [loop] node [above] {$g$} (s3D) ;
      \draw (s4D) edge [loop] node [above] {$b$} (s4D) ;
      \draw (s5D) edge [loop] node [above] {$r$} (s5D) ;
      \draw (s6D) edge [loop] node [above] {$g$} (s6D) ;
      \draw (s1D) edge node [above] {$\tau$} (s2D) ;
      \draw (s2D) edge node [above] {$\ch$} (s3D) ;
      \draw (s4D) edge node [above] {$\ch$} (s3D) ;
      \draw (s5D) edge node [above] {$\ch$} (s6D) ;
      \draw (s2C) edge node [below] {$\tau$} (s1C) ;
      \draw (s3C) edge node [below] {$\ch$} (s2C) ;
      \draw (s3C) edge node [below] {$\ch$} (s4C) ;
      \draw (s6C) edge node [below] {$\ch$} (s5C) ;
      \tikzset{bend right, pos=0.65}
      \draw (s1D) edge node [left=-1pt] {$\eC$} (s1C) ;
      \draw (s2D) edge node [left=-1pt] {$\eC$} (s2C) ;
      \draw (s3D) edge node [left=-1pt] {$\eC$} (s3C) ;
      \draw (s4D) edge node [left=-1pt] {$\eC$} (s4C) ;
      \draw (s5D) edge node [left=-1pt] {$\eC$} (s5C) ;
      \draw (s6D) edge node [left=-1pt] {$\eC$} (s6C) ;
      \draw (s1C) edge node [right=-1pt] {$\eD$} (s1D) ;
      \draw (s2C) edge node [right=-1pt] {$\eD$} (s2D) ;
      \draw (s3C) edge node [right=-1pt] {$\eD$} (s3D) ;
      \draw (s4C) edge node [right=-1pt] {$\eD$} (s4D) ;
      \draw (s5C) edge node [right=-1pt] {$\eD$} (s5D) ;
      \draw (s6C) edge node [right=-1pt] {$\eD$} (s6D) ;
    \end{scope}
  \end{tikzpicture}

%% file: fig-QdCM2.tex
\begin{figure}[tbh]
  \centering
  \begin{tikzpicture}[%scale=0.7,
    every state/.append style={inner sep=1pt, minimum size=6mm},
    every loop/.style={min distance=10mm, out=120, in=60, looseness=5}]

    \node[state, label=above:{$\SET{r}$}] (s1) {$x_1$};
    \node[state, label=above:{$\SET{r}$}, right of=s1] (s2) {$x_2$};
    \node[state, label=above:{$\SET{g}$}, right of=s2] (s3) {$x_3$};
    \node[state, label=above:{$\SET{b}$}, right of=s3] (s4) {$x_4$};
    \node[state, label=above:{$\SET{r}$}, right of=s4, xshift=5mm] (s5) {$x_5$};
    \node[state, label=above:{$\SET{g}$}, right of=s5] (s6) {$x_6$};
    \draw   
    (s1) edge [thick,bend right] node{} (s2)
    (s2) edge [thick,bend right] node{} (s3)
    (s4) edge [thick,bend right] node{} (s3)            
    (s5) edge [thick,bend right] node{} (s6)
    (s2) edge [thick,bend right] node{} (s1)
    (s3) edge [thick,bend right] node{} (s2)
    (s3) edge [thick,bend right] node{} (s4)            
    (s6) edge [thick,bend right] node{} (s5);            
  \end{tikzpicture}
  \caption{A \emph{symmetric} \cm.}
  \label{fig:QdCM2}
\end{figure}

%% file: fig-LTS-QdCM2.tex
\begin{figure}[tbh]
  \centering 
  \begin{tikzpicture}[% {scale=0.7}
    every state/.append style={inner sep=1pt, minimum size=6mm},
    every loop/.style={min distance=10mm, out=120, in=60, looseness=5}]

    \node[state] (s1D) {$x_1$};
    \node[state, right of=s1D] (s2D) {$x_2$};
    \node[state, right of=s2D] (s3D) {$x_3$};
    \node[state, right of=s3D] (s4D) {$x_4$};
    \node[state, right of=s4D, xshift=5mm] (s5D) {$x_5$};
    \node[state, right of=s5D] (s6D) {$x_6$};

    \draw   
    (s1D) edge[thick,loop above] node{$r$} (s1D)
    (s1D) edge[thick,bend right,below] node{$\tau$} (s2D)
    (s2D) edge[thick,bend right,above] node{$\tau$} (s1D)
    (s2D) edge[thick,loop above] node{$r$} (s2D)
    (s2D) edge[thick,bend right,below] node{$\ch$} (s3D)
    (s3D) edge[thick,bend right,above] node{$\ch$} (s2D)
    (s3D) edge[thick,loop above] node{$g$} (s3D)
    (s4D) edge[thick,bend right,above] node{$\ch$} (s3D)
    (s3D) edge[thick,bend right,below] node{$\ch$} (s4D)
    (s4D) edge[thick,loop above] node{$b$} (s4D)
    (s5D) edge[thick,loop above] node{$r$} (s5D)
    (s5D) edge[thick,bend right,below] node{$\ch$} (s6D)
    (s6D) edge[thick,bend right,above] node{$\ch$} (s5D)
    (s6D) edge[thick,loop above] node{$g$} (s6D)
    ;            
  \end{tikzpicture}
  \caption{\lts{} resulting from the application of the encoding in  Definition~\ref{def:SQdCM2LTS} to the symmetric \qdcm{} of Figure~\ref{fig:QdCM2}.}
  \label{fig:LTSQdCM2}
\end{figure}

%% file: fig-QdCM2min.tex
\begin{figure}[tbh]
  \centering
  \begin{tikzpicture}[%
    every state/.append style={inner sep=1pt, minimum size=6mm},
    every loop/.style={min distance=10mm, out=120, in=60, looseness=5}]

    \node[state, label=above:{$\SET{r}$}] (s1) {$C_1$};
    \node[state, label=above:{$\SET{g}$}, right of=s1] (s3) {$C_2$};
    \node[state, label=above:{$\SET{b}$}, right of=s3] (s4) {$C_3$};
    \node[state, label=above:{$\SET{r}$}, right of=s4, xshift=5mm] (s5) {$C_4$};
    \node[state, label=above:{$\SET{g}$}, right of=s5] (s6) {$C_5$};
    \draw   
    (s1) edge [thick,bend right] node{} (s3)
    % (s2) edge [thick,bend right] node{} (s3)
    (s4) edge [thick,bend right] node{} (s3)            
    (s5) edge [thick,bend right] node{} (s6)
    % (s2) edge [thick,bend right] node{} (s1)
    (s3) edge [thick,bend right] node{} (s1)
    (s3) edge [thick,bend right] node{} (s4)            
    (s6) edge [thick,bend right] node{} (s5);            
  \end{tikzpicture}
  \caption{The minimal \emph{symmetric} \cm{} of that in Fig.~\ref{fig:QdCM2}}
  \label{fig:QdCM2min}
\end{figure}

%% file: FeasibilityStudy.tex
\section{Spatial Model Checking of Digital Images via Minimised Models} %\section{Feasibility Study}
\label{sec:FisStudy}

In this section we present the \voxminx\ toolchain for spatial model checking of digital images, that exploits the minimisation procedure described in Section~\ref{sec:translation}.  We first present a high-level overview of the toolchain, followed by a more detailed discussion in subsequent subsections. We then describe an experimental evaluation showing the feasibility of the proposed toolchain and compare its results with that of the spatial model checker \voxlogica\ for a selected set of benchmark images.

%% tikz-toolchainext-ev-2.tex

\newcommand*\circled[1]{\tikz[baseline=(char.base)]{
    \node[shape=circle, draw, minimum size=4.25mm, inner sep=0.25pt]
    (char) {\small \textbf{#1}};}}

\tikzstyle {data} = [chamfered rectangle, draw, minimum height=12mm, minimum width=24mm]
\tikzstyle {comp} = [rectangle, draw, minimum height=12mm, minimum width=24mm]

\begin{figure}[t!]
%\centering
\noindent
\scalebox{0.80}{%
\begin{tikzpicture}[every text node part/.style={align=center}, line width=1pt]
  \node [data, fill=Dandelion!25, draw=Dandelion] (formula) at (1,5) {input \\ formula} ;
  \node [data, fill=Dandelion!25, draw=Dandelion] (image) at (1,3) {input \\ image} ;
  \node [comp, fill=Thistle!25, draw=Thistle] (encoding) at (4,3) {encoding} ;
  \node [comp, rounded corners=2mm, fill=NavyBlue!25, draw=NavyBlue] (ltsconv) at (7,3) {\texttt{ltsconvert}} ;
  \node [comp, rounded corners=2mm, fill=NavyBlue!25, draw=NavyBlue] (ltsinfo) at (7,1) {\texttt{ltsinfo}} ;
  \node [comp, fill=Thistle!25, draw=Thistle] (lts2graph) at (10,3) {lts2graph} ; % shade of purple
  \node [data, fill=orange!25, draw=orange] (graphmodel) at (13,3) {graph\,model} ;
  \node [data, fill=RedOrange!35, draw=RedOrange] (decgraphmodel) at (16,5) {decorated \\ CS\,model} ;
  \node [data, fill=orange!25, draw=orange] (classpixelrel) at (11.5,1) {class/pixel \\ relation} ;
  \node [comp, fill=Thistle!25, draw=Thistle] (graphlogica) at (13,5) {GraphLogica} ;
  \node [comp, fill=Thistle!25, draw=Thistle] (presenter) at (16,1) {presenter} ;
  \node [data, fill=red!25, draw=red] (output) at (16,-1) {output \\ image} ;
  \tikzset{line width=0.5pt, >=stealth'}
  \draw [->] (image) -- (encoding) ;
  \draw [->] (encoding) -- (ltsconv) ;
  \draw [->] (ltsconv) -- (lts2graph) ;
  \draw [->] (ltsconv) -- (ltsinfo) ;
  \draw [->] (lts2graph) -- (graphmodel) ;
  \draw [->] (graphmodel) -- (graphlogica) ;
  \draw [->] (formula) -- (graphlogica) ;
  \draw [->] (graphlogica) -- (decgraphmodel) ;
  \draw [->] (ltsinfo) -- (classpixelrel) ;
  \draw [->] (decgraphmodel) -- (presenter) ;
  \draw [->] (classpixelrel) -- (presenter) ;
  \draw [->] (presenter) -- (output) ;
  \node at ([xshift=-14mm,yshift=+7mm]formula) {\circled{6}} ;
  \node at ([xshift=-15mm,yshift=+7mm]graphlogica) {\circled{7}} ;
  \node at ([xshift=-14mm,yshift=+7mm]decgraphmodel) {\circled{8}} ;
  \node at ([xshift=-13mm,yshift=+7mm]image) {\circled{1}} ;
  \node at ([xshift=-15mm,yshift=+7mm]encoding) {\circled{2}} ;
  \node at ([xshift=-15mm,yshift=+7mm]ltsconv) {\circled{3}} ;
  \node at ([xshift=-15mm,yshift=+7mm]lts2graph) {\circled{4}} ;
  \node at ([xshift=-14mm,yshift=+7mm]graphmodel) {\circled{5}} ;
  \node at ([xshift=-15mm,yshift=+7mm]ltsinfo) {\circled{9}} ;
  \node at ([xshift=-14mm,yshift=+7mm]classpixelrel) {\circled{10}} ;
  \node at ([xshift=-15mm,yshift=+7mm]presenter) {\circled{11}} ;
  \node at ([xshift=-14mm,yshift=+7mm]output) {\circled{12}} ;
\end{tikzpicture}
} %% scalebox
%}
\caption{\voxminx\ toolchain for model checking via model minimisation and projection of results onto the original image. Parts 3 and 9 in blue
  are command line operations of the \mCRLT{} toolset. Parts 2, 4, 7 and 11 in purple are developed in Python in the context of the current
  paper. The orange parts 5, 8 and 10 are intermediate data structures. Parts 1 and 6 in yellow are the input and part 12 in red is the output of the procedure.}\label{fig:toolchainext} 
\end{figure}
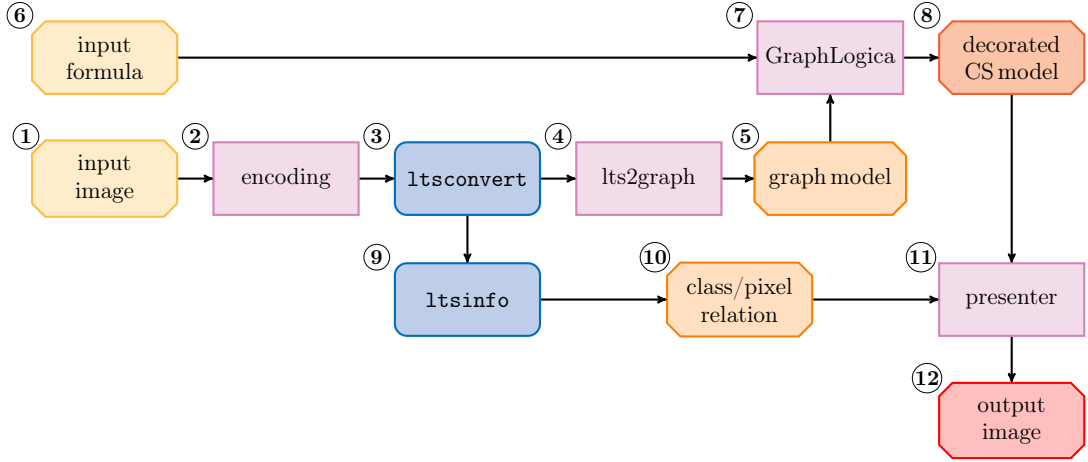

\subsection{The \voxminx\ toolchain} \mbox{} \smallskip

\noindent
The major components of the \voxminx\ toolchain are shown in Figure~\ref{fig:toolchainext}. As input \voxminx\ takes a digital image (box~1) and a spatial logic specification based on \icrl\ (box 6). As output, an adaptation of the original image is provided in which the pixels that satisfy the specification are highlighted and the others are shaded. Figure~\ref{fig:mazef1results} shows an example for the verification of an image of a maze of 4096x4096 pixels against the \icrl\ formula
\begin{displaymath}
 \phi^{maze}_1=\lstothru\,\mathit{blue}[\mathit{white}] \land \lstothru\,\mathit{green}[\mathit{white}].
\end{displaymath}
Formula $\phi^{maze}_1$ expresses the fact that a point (pixel) in the image should be white and should be the starting point both of a path passing by white points only reaching a blue point and of a path, passing by white points only, reaching a green point (situated in the upper-left corner of the maze) that represents the exit of the maze. In other words, the formula holds for all white pixels on a path from a blue starting point to the green exit.
 
The original image of the maze is shown at the left in Figure~\ref{fig:mazef1results}; the adapted image resulting from model checking is shown in the middle. In the latter, the white pixels satisfy the formula, whereas the black pixels  do not. At the right in Figure~\ref{fig:mazef1results} an alternative visualisation of the same results are presented in which the pixels satisfying the formula are highlighted and those that do not satisfy the formula are shaded (i.e.\ including the shaded black, green and blue ones). These are only two examples of how the results can be presented. Depending on the application area, other presentations could be more common or adequate. For example, in medical imaging, regions of interest (ROI) in images  
are often shown as a semi-transparent overlay, in a colour of choice, overlapping the original image. 
An example is shown in Figure~\ref{fig:brain}, where the pixels satisfying the logic specification 
that identifies tumour and oedema tissue, are shown in semi-transparent green, whereas the pixels in red are part of the ground truth provided by human experts. Note that the green and red overlays are mostly overlapping and that the differences between the two tumour segmentations are shown as (a small number of) bright red and bright green pixels. The latter feature greatly facilitates the assessment of the quality of the segmentation obtained by spatial model checking with respect to the ground truth. 
Such alternative visual presentations can be easily accommodated by adjusting the back-end software of the toolchain or via the user interface settings of the visualisation tool that has been developed for \voxlogica{}~\cite{strippoli2025}.

%% EV was here, 26 May 2026

\begin{figure}[t!]
        \centering
        \subfloat[][]
        {
                \begin{minipage}{0.3\textwidth}
                        \raisebox{0.385in}{\includegraphics[width=\textwidth]{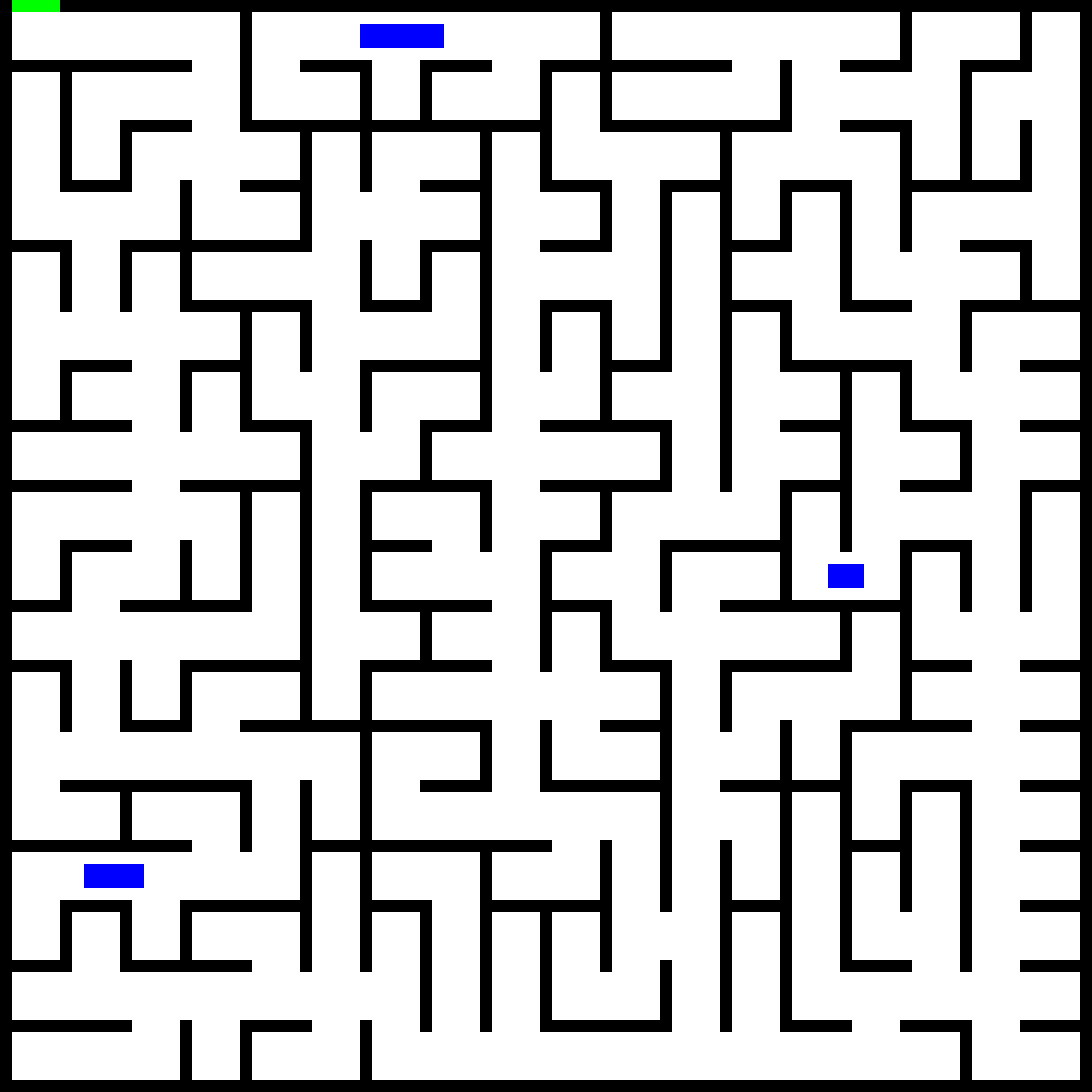}}
                \end{minipage}
                \label{fig:maze4096}
        }\quad
        \centering
        \subfloat[][]
        {\begin{minipage}{0.3\textwidth}
                        \raisebox{0.385in}{\includegraphics[width=\textwidth]{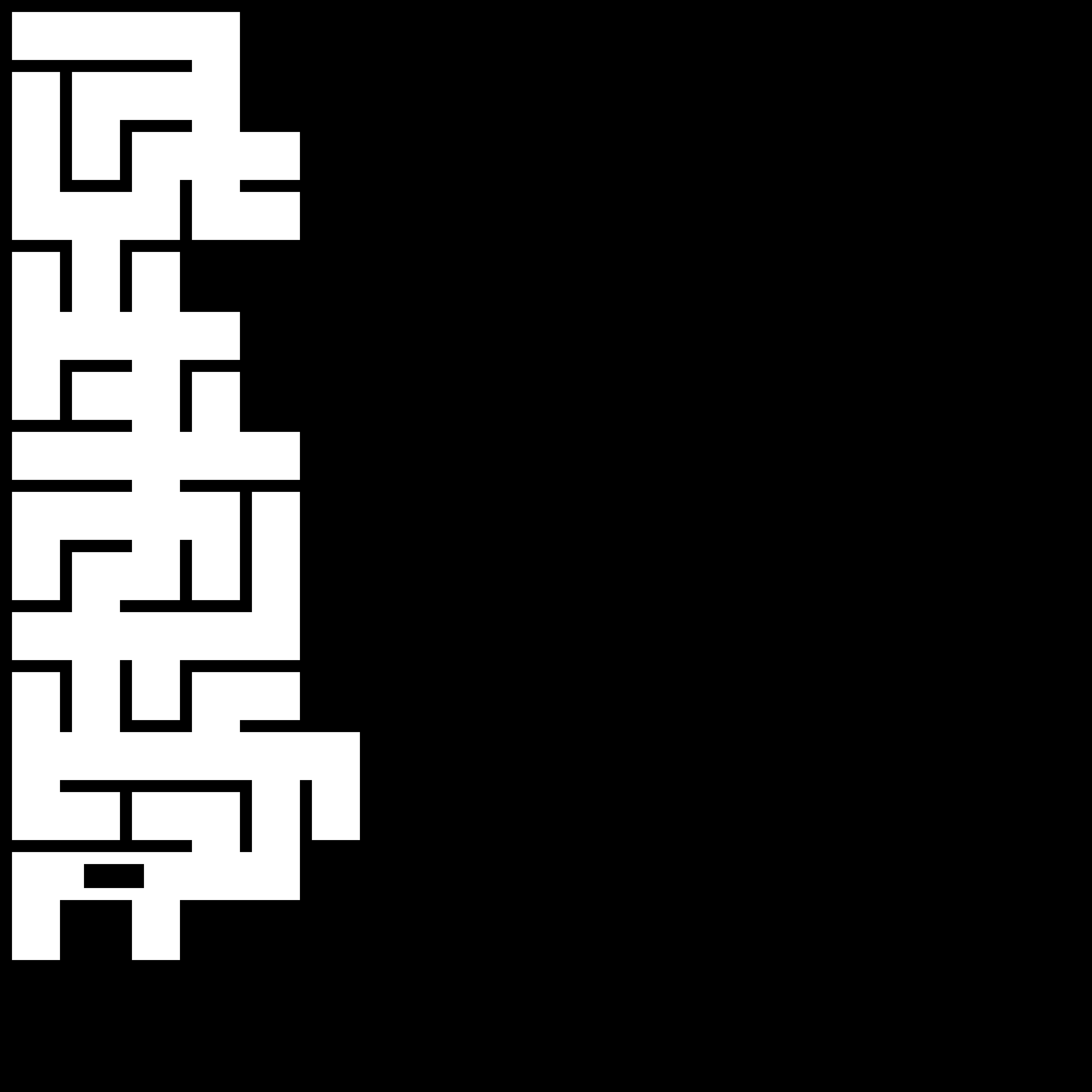}}
                \end{minipage}
                \label{fig:maze4096form1bw}
        }\quad
        \centering
        \subfloat[][]
        {\begin{minipage}{0.3\textwidth}
                        \raisebox{0.385in}{\includegraphics[width=\textwidth]{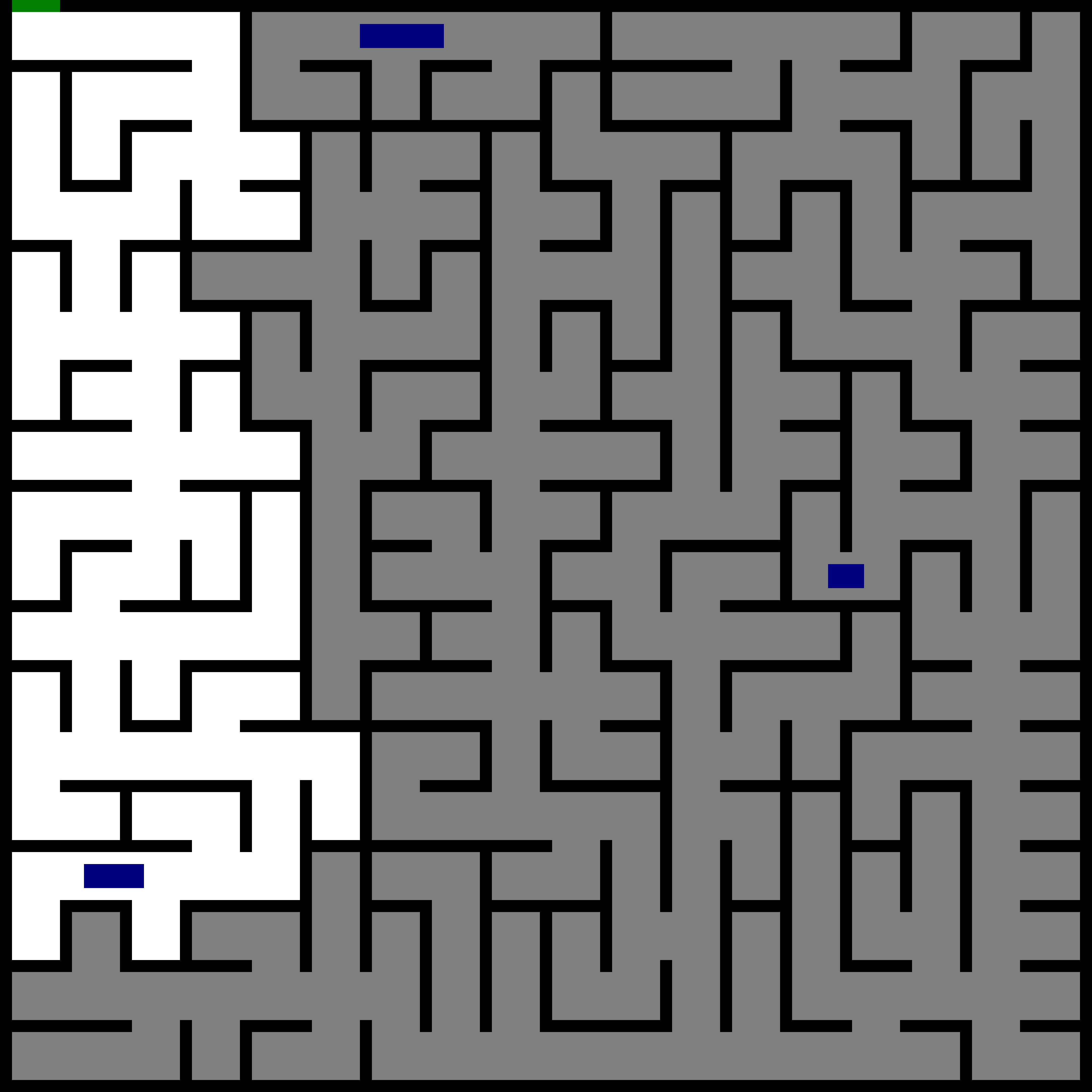}}
                \end{minipage}
                \label{fig:maze4096form1high}}
\caption{Maze example (4096x4096). Figure~\ref{fig:maze4096} original image. Figure~\ref{fig:maze4096form1bw} is showing the model checking results in which pixels satisfying $\phi^{maze}_1$ are shown in white, others are shown in black. Figure~\ref{fig:maze4096form1high} shows an alternative visualisation of the same results where pixels satisfying formula $\phi^{maze}_1$ are highlighted (white), others are grey-shaded. The green exit is situated in the top left corner of the maze in Figure~\ref{fig:maze4096}.}\label{fig:mazef1results}
\end{figure}

{In Figure~\ref{fig:pacman200ghostsresults} a more colourful and complex image is shown of a scene of the well-known video game Pac-Man\footnote{Pac-Man Official Website -- History: \url{ https://pacman.com/en/history/}. Accessed on September 2, 2025.} that is checked against the nested \icrl\ formula 
\begin{equation}
\lstothru\, (\lstothru\, \mathit{blue}
[\mathit{white}])\,[\neg(\mathit{black} \lor \mathit{blue})]
\label{eqn:eyes-formula}
\end{equation}
This formula holds for all the so-called Ghost Monster figures in this picture. Observe that the eyes of the ghosts are composed of a blue and a white part and that each Ghost Monster has at least one eye in which the blue part does not touch the black background but does touch the white part of the eye instead. Such a spatial configuration only occurs in the ghosts, and therefore the formula uniquely identifies them. The original image of the Pac-Man scene is shown on the left; the image in the middle is showing model checking results where pixels satisfying the formula are shown in white and the others in black. An alternative visualisation is shown on the right, where the highlighted pixels (identifying the ghosts) satisfy the formula, and the shaded pixels do not.

\begin{figure}[h!]
        \centering
        \subfloat[][]
        {
                \begin{minipage}{0.3\textwidth}
                        \raisebox{0.385in}{\includegraphics[width=\textwidth]{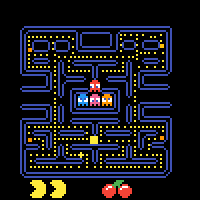}}
                \end{minipage}
                \label{fig:pacman200}
        }\quad
        \centering
        \subfloat[][]
        {\begin{minipage}{0.3\textwidth}
                        \raisebox{0.385in}{\includegraphics[width=\textwidth]{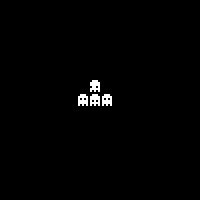}}
                \end{minipage}
                \label{fig:pacman200ghostsbw}
        }\quad
        \centering
        \subfloat[][]
        {\begin{minipage}{0.3\textwidth}
                        \raisebox{0.385in}{\includegraphics[width=\textwidth]{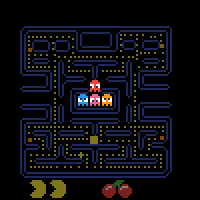}}
                \end{minipage}
                \label{fig:pacman200ghostshigh}}
\caption{Pacman example (200x200 pixels). Figure~\ref{fig:pacman200} original image. Figure~\ref{fig:pacman200ghostsbw} shows pixels satisfying formula {\sf ghosts} in white, others are shown in black. Figure~\ref{fig:pacman200ghostshigh} pixels satisfying formula {\sf ghosts} as highlighted colours, others are shown in shaded colours.}\label{fig:pacman200ghostsresults}
\end{figure}

\bigskip

\noindent Next, let us have a closer look at the other parts of the toolchain in Figure~\ref{fig:toolchainext}. First the input image (box~1) is converted into a \qdcm\ representation and fed through an encoder (box~2), which converts the image into an LTS, implementing the encoding defined in Definition~\ref{def:SQdCM2LTS}. This LTS has one state for each pixel of the image with a self-loop recording the colour of the pixel in its label. The LTS is subsequently minimised with respect to branching bisimilarity using an existing minimisation tool available from the \mCRLT\ tool suite~\cite{Bu+19} (box~3). Care is taken that the resulting equivalence classes, represented by the states of the minimal LTS, remain associated to their respective set of pixels in the original image. The minimal LTS is then transformed into an annotated closure space model representing the minimal \qdcm\ modulo \cop-bisimilarity (box~5). This transformation is the implementation of what is established by Theorem~\ref{prop:logical-minimization}. The closure space model (box~5) and the input formula (box~6) are provided as input to the \graphlogica\ model checker (box~7, discussed below). The model checking result (box~8) is an annotated version of the input closure space model in which the points that satisfy the formula are indicated by a specific label. These points are representatives of equivalence classes of pixels present in the input image. This outcome of the model checker (box~8) is then combined with the information linking the equivalences classes to the sets of pixels they represent (box~10) by the presenter (box~11) and used to produce the adapted input image as output (box~12). We refer the reader interested in further details on each of the steps in the toolchain, and the particular file formats involved, to the small running example presented in Appendix~\ref{apx:running}.

In the previous section, in Figure~\ref{fig:Maze} we have already seen the example of the maze and the minimal models that are obtained applying the general encoding (see Definition~\ref{def:QdCM2LTS}) and the optimised encoding (see Definition~\ref{def:SQdCM2LTS}) followed by branching bisimulation minimisation. Let us here consider a more complex example using \voxminx.
Figure~\ref{fig:pacmin} shows the minimal LTS, obtained after an optimised encoding, for the scene of the Pac-Man game shown in Figure~\ref{fig:pacman200}. This minimal model has 35~states, each representing an equivalence class (of pixels) modulo \icrl. It is not difficult to recognise some groups of states representing specific parts of the image. For example, at the left one observes a group of two white, two red and a green state (states 1, 5, 9, 33,~34). This group represents the two cherries at the bottom of the original Pac-Man image. The single black node in the middle (state~4) represents all pixels of the black background. Note that there is only one such state because \voxminx, like \voxlogica, constructs a regular graph model from a digital image based on an 8-adjacency relationship\footnote{In digital image processing the basic relationships between pixels are the 4-adjacency and the 8-adjacency relationship. The former considers four neighbours of each pixel in the adjacency relationship, namely its left and right neighbours and its neighbours directly above and below it. The latter considers also the four diagonal neighbours of a pixel. This holds for all pixels except for those at the border of the image that have a reduced number of neighbours depending on their position.} between pixels, so also those pixels touching diagonally are considered adjacent in the graph.  Therefore all black pixels are adjacent. The four groups of states at the top of the figure represent the four ghosts: On the left, states 10, 24, 29, 30, 31, and~32 represent the red ghost; next of that group on the right the states 3, 13, and~14 represent the cyan ghost; next of the latter group the states 2, 6, and~8 represent the pink ghost, and on the right the states 7, 11, and~25 represent the orange ghost. Note that each ghost has indeed a blue state (representing blue pixels of the eye) that has no transition to the black state (state~4). We used this observation to find a formula to identify the ghosts. With the red ghost three blue states are associated. Two of them have a transition to the black state (state~4) but they represent different equivalence classes of pixels. State 30 represents blue pixels that are not part of an eye but are actually part of the blue stripes marking the sides of paths in the Pac-Man game that happen to touch the red part of the ghost (situated on top of its head). State~29, instead, represents pixels of the only eye of the red ghost (its right eye) that has a blue part touching the black background. At the bottom of Figure~\ref{fig:pacmin} we can observe the states 12, 21 and~22, that represent a pale orange cookie, coloured slightly different from the other three orange cookies, in the top-right of the original image. That specific cookie does touch a blue pixel (belonging to the border of a path) that, in turn, touches a yellow pellet (diagonally). 
%Recall that \voxmin\ works with an 8-adjacency relationship between pixels. 
The single yellow state (state~0) represents both the pixels of the yellow Pac-Man items at the bottom of the image, and the pixels of most of the yellow pellets. Note that these objects are all surrounded by black pixels belonging to the background. 

\definecolor{myred}{HTML}{F44336}
\definecolor{mypink}{HTML}{F06292}
\definecolor{myviolet}{HTML}{3F51B5}
\definecolor{myorange}{HTML}{FFC107}
\definecolor{myorange2}{HTML}{FF9800}
\definecolor{mygreen}{HTML}{4CAF50}
\definecolor{myyellow}{HTML}{FFEB3B}
%\definecolor{mycursor}{HTML}{c3F51B5}
\definecolor{mycursor}{HTML}{3F51B5}
\definecolor{myblue}{HTML}{2196F3}

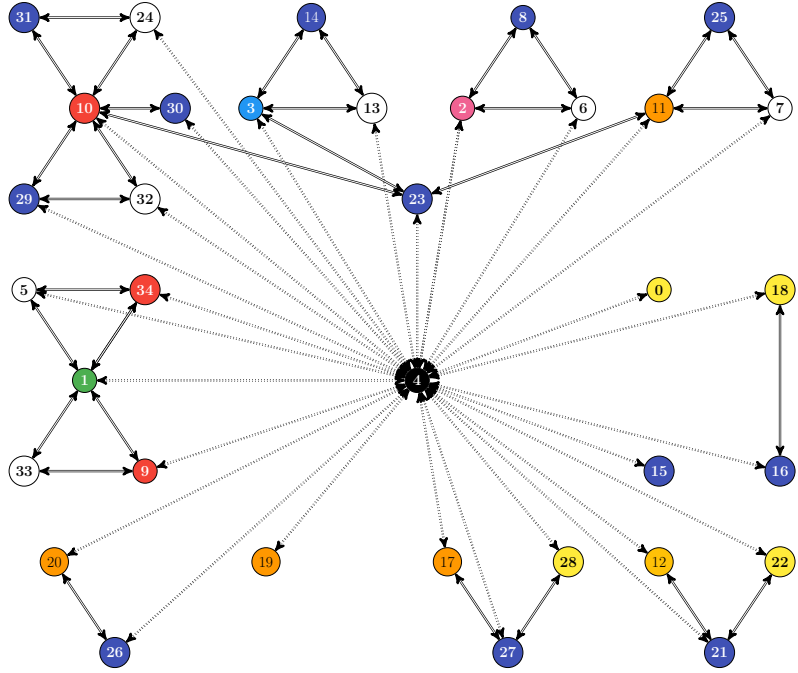
\begin{figure}
\scalebox{0.4}{%
\begin{tikzpicture}
  \tikzstyle{state}=[circle, draw, minimum size=8mm, font=\Large]
  \tikzstyle{stateviolet}=[circle, draw, minimum size=8mm, fill=myviolet, text=white, font=\bfseries\Large]
  \tikzstyle{stategreen}=[circle, draw, minimum size=8mm, fill=mygreen, text=white, font=\bfseries\Large]
  \tikzstyle{stateblue}=[circle, draw, minimum size=8mm, fill=myblue, text=white, font=\bfseries\Large]
  \tikzstyle{statepink}=[circle, draw, minimum size=8mm, fill=mypink, text=white, font=\bfseries\Large]
  \tikzstyle{statered}=[circle, draw, minimum size=8mm, fill=myred, text=white, font=\bfseries\Large]
  \tikzstyle{stateyellow}=[circle, draw, minimum size=8mm, fill=myyellow, text=black, font=\bfseries\Large]
  \tikzstyle{statewhite}=[circle, draw, minimum size=8mm, fill=white, text=black, font=\bfseries\Large]

  \tikzstyle{edge}=[thick, <->, double, mark=at position 1 with {\arrow[scale=3,>=stealth]{>}} ]
  \tikzstyle{edgedotted}=[thick, dotted,<->, double, mark=at position 1 with {\arrow[scale=3,>=stealth]{>}}]
%  \tikzstyle{edge}=[thick, <->, decoration={markings,mark=at position 1 with
%    {\arrow[scale=3,>=stealth]{>},\arrow[scale=3,stealth-<]{<}}},postaction={decorate}]
%    \tikzstyle{edge}=[thick, <->, decoration={markings,mark=at position 1 with
%    {\arrow[scale=3,>=stealth]{>}}},postaction={decorate}]
%  \tikzstyle{edgedotted}=[thick, dotted]
  
  % \tikzstyle{initstate}=[state,fill=green]
  % \tikzstyle{inactive}=[dashed]
  % \tikzstyle{prob}=[line width = 2.5pt]
  % \tikzstyle{transition}=[->,arrows={-Stealth[scale=2,inset=1pt]}]
   
%% background black, initial state
\node [state, fill=black, text=white, font=\bfseries\Large] at (-4,10) (state4) {4};

\begin{scope}[shift={(-2,0)}]
%% isolated wall bricks
\node [stateviolet] at (6,7) (state15) {15};
\draw [edgedotted] (state4) -- (state15) ;

%% cursor
\node [stateviolet] at (10,7) (state16) {16};
\node [stateyellow] at (10,13) (state18) {18};

\draw [edgedotted] (state4) -- (state16) ;
\draw [edgedotted] (state4) -- (state18) ;
\draw [edge] (state16) -- (state18) ;

%% pac-man
\node [stateyellow] at (6,13) (state0) {0};
\draw [edgedotted] (state4) -- (state0) ;
\end{scope}

%% cherries
\begin{scope}[shift={(1,0)}]
\node [stategreen] at (-16,10)  (state1) {1};
%% white and red of left cherry (not connected to black)
\node [statewhite] at (-18,7)  (state33) {33};
\node [statered] at (-14,7) (state9) {9};
%% white and red of right cherry (connected to black)
\node [statewhite] at (-18,13)  (state5) {5};
\node [statered] at (-14,13)  (state34) {34};

\draw [edgedotted] (state4) -- (state1) ;
\draw [edgedotted] (state4) -- (state5) ;
\draw [edgedotted] (state4) -- (state9) ;
\draw [edgedotted] (state4) -- (state34) ;
\draw [edge] (state1) -- (state34) ;
\draw [edge] (state1) -- (state5) ;
\draw [edge] (state1) -- (state9) ;
\draw [edge] (state1) -- (state33) ;
\draw [edge] (state1) -- (state34) ;
\draw [edge] (state5) -- (state34) ;
\draw [edge] (state9) -- (state33) ;
\end{scope}

%% wall bricks (connected to ghosts)
\node [stateviolet] (state23) at (-4,16) {23} ;

%% red ghost
\begin{scope}[shift={(1,-1)}]
\node [statered] at (-16,20) (state10) {10};
\node [statewhite] at (-14,23) (state24) {24};
\node [stateviolet] at (-18,17) (state29) {29};
\node [stateviolet] at (-13,20)  (state30) {30};
\node [stateviolet] at (-18,23) (state31) {31};
\node [statewhite] at (-14,17)  (state32) {32};

\draw [edgedotted] (state4) -- (state10) ;
\draw [edgedotted] (state4) -- (state23) ;
\draw [edgedotted] (state4) -- (state24) ;
\draw [edgedotted] (state4) -- (state29) ;
\draw [edgedotted] (state4) -- (state30) ;
\draw [edgedotted] (state4) -- (state32) ;
\draw [edge] (state23) -- (state10) ;
\draw [edge] (state10) -- (state24) ;
\draw [edge] (state10) -- (state29) ;
\draw [edge] (state10) -- (state30) ;
\draw [edge] (state10) -- (state31) ;
\draw [edge] (state10) -- (state32) ;
\draw [edge] (state24) -- (state31) ;
\draw [edge] (state29) -- (state32) ;
\end{scope}

%% blue ghost
\begin{scope}[shift={(0.5,-1)}]
\node [stateblue] at (-10,20) (state3) {3};
\node [stateviolet, state] at (-8,23) (state14) {14};
\node [statewhite] at (-6,20) (state13) {13};

\draw [edgedotted] (state4) -- (state3) ;
\draw [edgedotted] (state4) -- (state13) ;
\draw [edge] (state23) -- (state3) ;
\draw [edge] (state3) -- (state13) ;
\draw [edge] (state3) -- (state14) ;
\draw [edge] (state13) -- (state14) ;
\end{scope}

%% pink ghost
\begin{scope}[shift={(-0.5,-1)}]
\node [statepink] at (-2,20) (state2) {2};
\node [statewhite] (state6) at (2,20) {6};
\node [stateviolet] at (0,23) (state8) {8};

\draw [edgedotted] (state4) -- (state2);
\draw [edgedotted] (state4) -- (state6);
\draw [edgedotted] (state4) -- (state2);
\draw [edge] (state2) -- (state6);
\draw [edge] (state2) -- (state8);
\draw [edge] (state6) -- (state8);
\end{scope}

%% orange ghost
\begin{scope}[shift={(-2,-1)}]
\node [fill=myorange2, state] at (6,20) (state11) {11};
\node [stateviolet] at (8,23) (state25) {25};
\node [statewhite] at (10,20)  (state7) {7};

\draw [edgedotted] (state4) -- (state7) ;
\draw [edgedotted] (state4) -- (state11) ;
\draw [edge] (state23) -- (state11) ;
\draw [edge] (state11) -- (state7) ;
\draw [edge] (state11) -- (state25) ;
\draw [edge] (state7) -- (state25) ;
\end{scope}

%% orange block (top-left)
\begin{scope}[shift={(2,1)}]
\node [state, fill=myorange2] at (-18,3) (state20) {20} ;
\node [stateviolet] at (-16,0) (state26) {26};

\draw [edgedotted] (state4) -- (state20) ;
\draw [edgedotted] (state4) -- (state26) ;
\draw [edge] (state20) -- (state26) ;
\end{scope}

%% isolated orange block (bottom-left)
\begin{scope}[shift={(1,1)}]
\node [fill=myorange2, state] at (-10,3) (state19) {19} ;
\draw [edgedotted] (state4) -- (state19) ;
\end{scope}

%% orange block (bottom-right)
\begin{scope}[shift={(-1,1)}]
\node [state, fill=myorange2] at (-2,3) (state17) {17} ;
\node [stateviolet] at (0,0) (state27) {27} ;
\node [stateyellow] at (2,3) (state28) {28} ;

\draw [edgedotted] (state4) -- (state17) ;
\draw [edgedotted] (state4) -- (state27) ;
\draw [edgedotted] (state4) -- (state28) ;
\draw [edge] (state17) -- (state27) ;
\draw [edge] (state27) -- (state28) ;
\end{scope}

%% light orange block top-right
\begin{scope}[shift={(-2,1)}]
\node at (6,3) [state, fill=myorange] (state12) {12};
\node at (8,0) [stateviolet] (state21) {21};
\node [stateyellow]at (10,3)  (state22) {22};

\draw [edgedotted] (state4) -- (state12) ;
\draw [edgedotted] (state4) -- (state21) ;
\draw [edgedotted] (state4) -- (state22) ;
\draw [edge] (state12) -- (state21) ;
\draw [edge] (state21) -- (state22) ;
\end{scope}

\end{tikzpicture}
} %% scalebox
\caption{Minimal LTS of the Pac-Man image shown in Figure~\ref{fig:pacman200}. State colours correspond to colours in the image and states are numbered for convenience of reference. Each state corresponds to an equivalence class. There are 35 classes. The class in the centre (state 4) is black and corresponds to the background of the image. The label denoting `change' at all transitions, as well as all self-loops and related labels have been omitted. The latter labels have been replaced by colouring the states themselves instead. Furthermore, the states have been numbered for easy reference.}\label{fig:pacmin}
\end{figure}

\subsection{Experimental Evaluation on Digital Images}
\mbox{} \smallskip \\
\noindent
In this subsection we present an experimental evaluation of \voxminx\
on three sets of digital images to provide insight in the feasibility and performance of the spatial model checking approach exploiting label transition systems and their minimization. 
In particular, we first provide a more detailed description of the tools that form part of the \voxminx\ toolchain and of others that are used in the experiments. 
Then we present the setup of the experimental evaluation of \voxminx, including the equipment used, the test images, and the logical specifications used for each test image. 
Section~\ref{sec:vmx} presents the performance results, also comparing those obtained with \voxminx\ and those obtained with \voxlogica.

\smallskip

The following tools are involved in the experiments: 
On the one hand \voxminx, which includes \graphlogica\ and \mCRLT,  and on the other hand \voxlogica. 
The spatial model checker \graphlogica\ is a model checker that can be
applied to general graphs, but it is less suitable for analysis of digital images as it has not been optimised for their specific models. 
\voxlogica\ is a state-of-the-art spatial model checker, used to perform spatial model checking of full digital images, i.e.\ without use of minimisation.  
Therefore we use \voxlogica\ to compare the model checking results with those of \voxminx. 
Differently from \graphlogica, the \voxlogica\ model checker is not dealing with general graphs but specialised for images. 
We briefly describe some of the main features of these tools below. 

\emph{The \graphlogica\ model checker.}
In \voxminx\ we employ the spatial model checker \graphlogica. 
It implements a basic spatial model checking algorithm for the finitary fragment of \icrl\ on general graphs, based on the analysis of connected components for \cm{s} which are not necessarily symmetric. 
The reachability operators of \icrl\ are provided explicitly by the specification language \grql\ of \graphlogica\ (\grql\ is short for Graph Query Language). 
In~\grql, one or more logical formulas can be specified that are to be checked against the input model. 
When multiple formulas are given, they are checked one after the other in the same model checking session. 
Moreover, \graphlogica\ provides the implementation of the encodings presented in Section~\ref{sec:translation}.

\emph{The \mCRLT\ toolset.} The \mCRLT\ toolset~\cite{GrooteM2014,Bu+19} is a collection of tools for the analysis of models of distributed and concurrent systems. 
Its process language is based on the Algebra of Communicating Processes~\cite{BaW1990}, its property specification language is based on the modal $\mu$-calculus~\cite{BS07:handbook}}. 
The toolset includes a temporal model checker to formally verify and analyse models of complex software and protocols. 
It also includes operations for the minimisation of models based on various types of bisimilarity equivalences, in particular, the branching bisimilarity minimisation using the optimized algorithm proposed in~\cite{Gr+17}.  

\emph{The \voxlogica\ model checker.}
In~\cite{Be+19,BelmonteCM25} the \voxlogica\ model checker an
implementation of the model checking algorithm is presented. 
\voxlogica\ is specialised for digital images.
In particular, it employs efficient procedures\footnote{In particular, this implementation exploits very efficient algorithms from the Insight Tool Kit  (\textsf{ITK}) via the \textsf{SimpleITK} glue, see \url{https://itk.org} and \url{http://www.simpleitk.org} library~\cite{LowekampCIB13,McC+14} specially designed for (medical) digital image processing.} for the analysis of connected components. 
The tool \voxlogica\ implements, among others, a spatial model checking procedure for a logical operator for conditional reachability, included in \imgql, the specification language of \voxlogica. 

\smallskip

In order to measure the model checking speed-up of the \voxminx\ toolchain, with respect to direct spatial model checking of the full original images, we use \graphlogica\ for checking the minimal model, and \voxlogica\ to check the full model\footnote{We underline again that \voxlogica{} is inherently much faster than \graphlogica{} as it is specialised for images, exploiting state-of-the-art imaging libraries and automatic parallelisation. This poses a further challenge to the speed-up via minimisation and is the reason why we use~\voxlogica{} instead of \graphlogica{} for the full model. }.

The tests have been performed on a MacBook Pro equipped with an Apple M2~Pro processor and 32GB of~RAM running macOS Sequoia~15.7. 
The \mCRLT\ toolset version mcrl2-202507.0.66927898fc\_arm64 was used. 
In this version \mCRLT\ uses a non-recursive procedure to find strongly connected components that does not require a large stack. %~\cite{Wieger Wesselink et al}.
Full data, source code and tools needed to reproduce the maze and monoscope experiments can be found in the Zenodo repository~\cite{Zenodo}.\footnote{The results are currently available at a public repository on GitHub at: \\\url{https://github.com/VoxLogicA-Project/VoxMinX-Validation}\\ The Zenodo repository will be produced for the final version of the paper based on the GitHub repository.}

For the experimental evaluation we have used a benchmark consisting of three families of images: a family of maze images (see Figure~\ref{fig:Maze}), a family of monoscope image (see Figure~\ref{fig:monoscope}), and a family of images of a snapshot of the Pac-Man game (see Figure~\ref{fig:pacman200}). 
Each family consists of a benchmark image that has been rescaled at various resolutions. 
The names of the images in each family are composed of their name (\textsf{maze}, \textsf{mono},~\textsf{pm}) followed by a numeric indication of the vertical resolution of each image. 
The maze and the Pac-Man images are square, therefore their horizontal resolution coincides with their vertical one. 
For example, \textsf{maze-1024} is a {\sf png} image of the maze of 1024~pixels wide and 1024~pixels high. 
The monoscope image has a $16{:}9$~ratio, thus, e.g., the horizontal resolution of \textsf{mono-1080} is $1920$ pixels.

\begin{figure}[t!]
        \centering
        \includegraphics[width=.45\textwidth]{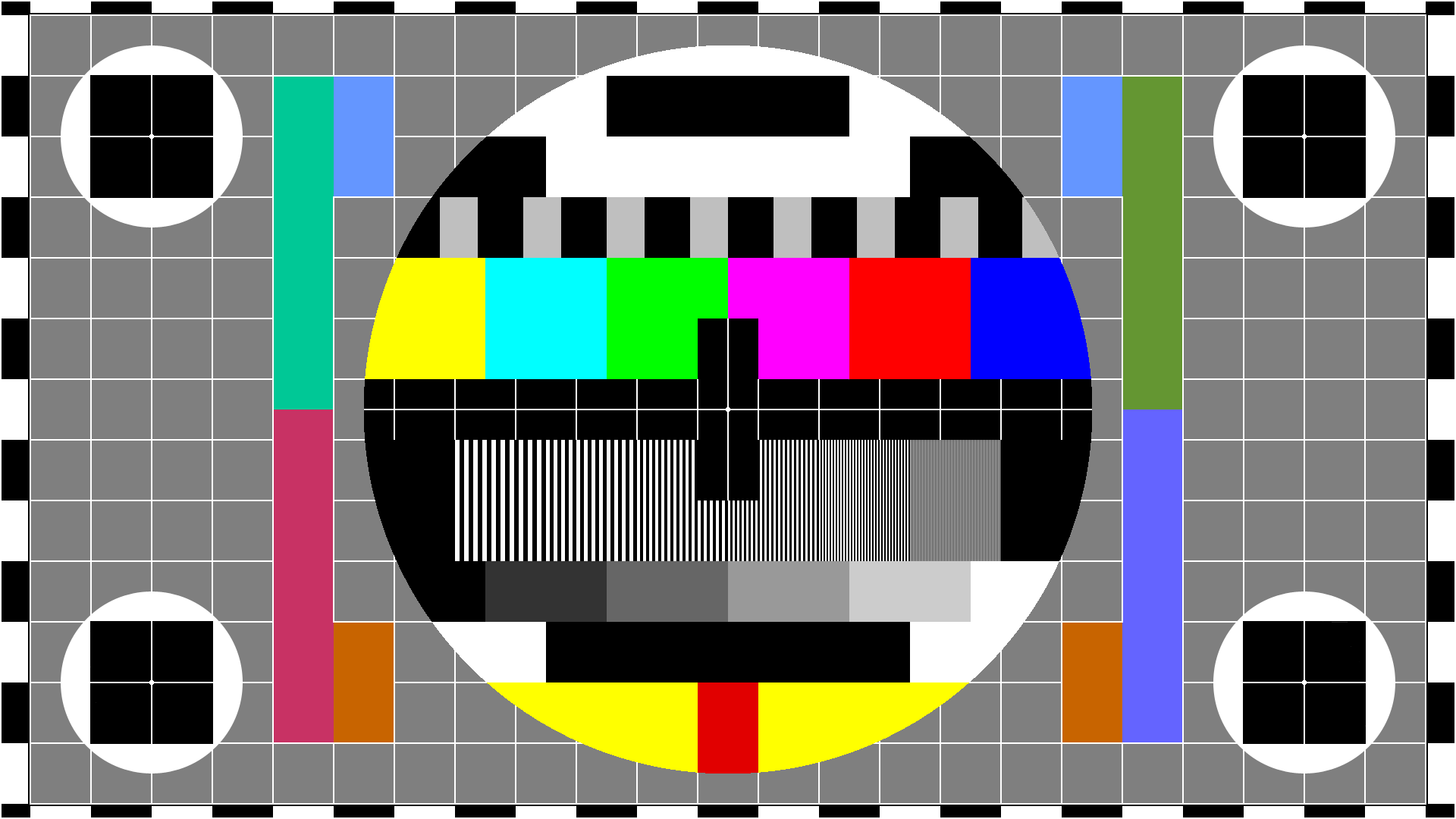}
        \caption{\label{fig:monoscope}Monoscope test pattern Philips PM5544}
\end{figure}
 
For the experiment with the \emph{maze image}, the property
specification consists of the three reachability formulas
that identify
\begin{enumerate}
\item the white points from which both a blue point and a green point can be reached (in other words, the white paths connecting blue points to the green exit)
  \begin{displaymath}
    \varphi^{\mkern1mu maze}_1=
    \lstothru\, \mathit{green}[\mathit{white}] \land
    \lstothru\, \mathit{blue} [\mathit{white}]
  \end{displaymath}
\item  the blue points from which there is \emph{no} white path to the green exit 
  \begin{displaymath}
    \varphi^{\mkern1mu maze}_2=
    \mathit{blue} \land
    \neg \lstothru\,(\lstothru\,\mathit{green}[\mathit{white}])[\mathit{blue}]
  \end{displaymath}
\item the blue points from which, instead, an exit \emph{can} be reached 
  \begin{displaymath}
    \varphi^{\mkern1mu maze}_3= \lstothru\, (
    \lstothru\, \mathit{green}[\mathit{white}] \land
    \lstothru\, \mathit{blue} [\mathit{white}]
    ) [\mathit{blue}] 
  \end{displaymath}
\end{enumerate}
Note that it holds that $\varphi^{\mkern1mu maze}_3 = \lstothru\, \varphi^{\mkern1mu maze}_1  [\mathit{blue}]$.

For the \emph{monoscope image}, the property specification is slightly artificial as it has been designed to be increasingly demanding in terms of computation time (caused by the nesting of sub-formulas). 
A single property~$\varphi^{\mkern1mu mono}$ is used in the experiments,
that characterises the points from which very specific paths start, crossing a number of different colours in a specific order, using $16$ nested reachability constraints, of the form 
\begin{displaymath}
  \varphi^{\mkern1mu  mono} =
  \lstothru ( \lstothru ( \lstothru( \ldots )
  [\mathit{green}]) [\mathit{cyan}]) [\mathit{yellow}]
\end{displaymath}

For the \emph{Pac-Man scene}, the logic specification illustrates how several objects of interest can be identified in the Pac-Man scene. 
These are the eyes of the ghosts, the cherries, the pellets, and the ghosts themselves, exploiting the various colours of the pixels and their relationship with other pixels. 
For the \emph{Pac-Man scene}, the property specification consists of
the reachability formulas $\varphi^{\mkern2mu pm}_1$ to~$\varphi^{\mkern2mu pm}_4$ below.
\begin{enumerate}
\item The eyes of the ghosts can be uniquely identified by white pixels through which blue pixels can be reached. 
  This can be expressed via the formula
  $$\varphi^{\mkern2mu pm}_1 =  \lstothru\,\mathit{blue}[\mathit{white}] $$
  Formula~$\varphi^{\mkern2mu pm}_1$ is also discussed in Equation~(\ref{eqn:eyes-formula}) above.
\item The cherries are uniquely identified by red, white, or green pixels through which green pixels can be reached. 
This can be expressed via formula
\begin{displaymath}
  \varphi^{\mkern2mu pm}_2 =
  \lstothru\,\mathit{green}
  [(\mathit{red} \lor \mathit{white} \lor \mathit{green})]
\end{displaymath}
\item Pellets are the four small orange squares.\footnote{One of the pellets in Figure~\ref{fig:pacman200} has a somewhat fainter orange colour. We will abstract from this here and consider them all the orange colour.} 
  They are not the only orange pixels, because there is also an orange ghost. To distinguish the orange pixels of pellets, we require that they are not those via which one can reach the eyes of the orange ghost. 
  This can be expressed via formula
  \begin{displaymath}
    \varphi^{\mkern2mu pm}_3 =
    \mathit{orange} \land
    \neg \lstothru\,\varphi^{\mkern2mu pm}_1[\mathit{orange}]
  \end{displaymath}
\item The ghosts themselves are uniquely identified by the colour of their pixels (one red, one cyan, one orange, and one pink, which is not black or blue) and the fact that they all have at least one white with blue eye that is not touching the black background. 
  This can be expressed via the formula
  \begin{displaymath}
    \varphi^{\mkern2mu pm}_4 =
    \lstothru\, \varphi^{\mkern2mu pm}_1
    [(\neg (\mathit{black} \lor \mathit{blue}))]
  \end{displaymath}
\end{enumerate}

% Maze specification:
% let form1 = through(N b,w) & through(N g,w)
% let form2 = through(N (w & !through(N g,w)),b) & through(N g,w)
% let form3 = through(N form1,b)
% Monoscope specification (the infix Z has swapped arguments...): 
% y Z c Z g Z m Z r Z b Z gr Z bl Z w Z gr Z bl Z y

\subsection{Performance Results}\label{sec:vmx}

Table~\ref{tab:resultsTable2mazeMacArm64scc}, Table~\ref{tab:resultsTable2monoMacArm64scc}, and Table~\ref{tab:resultsTable2PacManMacArm64scc} report the \voxminx\ results for each test image, for the logical properties specified earlier. 
We have run all the phases of our experiment for each image, also in the cases that produce the same minimal models, for simplicity of the set-up.

Each table shows the results for a specific image and its rescaled versions. Table~\ref{tab:resultsTable2mazeMacArm64scc} shows the results for the maze. 
The first column shows the name of the image, where the number attached to the name indicates the vertical resolution of the image in terms of the number of pixels. 
The second and third column show the time in seconds needed to perform the encoding, the former shows the pure computation time, whereas the latter includes the time for writing the result to file (i.e.\ including IO). 
The fourth, fifth, and sixth columns provide the number of states, transitions, and the size of the full model of the image (i.e.\ before minimisation), respectively. 
Columns seven to eleven provide, respectively, the minimisation time, without and with IO, the number of states of the minimal model, the number of transitions and the time to translate the results back to the original image. 
The last three columns of the table provide the pure model checking time of the full model (performed with \voxlogica), the model checking time of the minimal model (performed with \graphlogica), and the speed-up (gain) in model checking time obtained when using the minimal model compared to model checking the full model, respectively. 
All times are in seconds, rounded to two decimals.

\begin{table}%[h!]
       % \centering
        \scriptsize
        \hspace{-4em}
        %\rotatebox{90}{
        \begin{filecontents*}{experiments/macos_results_arm64_mCRL2sccnew/result_maze_newToolchain/results-table_new.csv}
label,conversionimg2aut,convWIOfull,pixels,transitions,autSizeH,min_computation,minWIO,statesMin,transitionsMin,convertBack,modelCheckingFull,modelCheckingMin,speedupMC
maze-128,0.18,0.21,16.00 K,142.50 K,2.47 MB,0.00,0.11,7,21,0.20,0.72,0.22,3.33
maze-256,0.18,0.24,64.00 K,573.00 K,10.35 MB,0.02,0.38,7,21,0.24,0.71,0.22,3.30
maze-512,0.18,0.39,256.00 K,2.24 M,44.55 MB,0.11,1.66,7,21,0.42,0.72,0.22,3.33
maze-1024,0.19,0.86,1.00 M,8.99 M,184.34 MB,0.48,7.71,7,21,1.13,0.75,0.22,3.48
maze-2048,0.24,2.88,4.00 M,35.98 M,793.73 MB,2.32,33.98,7,21,4.59,0.94,0.22,4.35
maze-4096,0.45,10.69,16.00 M,143.95 M,3.27 GB,9.89,145.63,7,21,19.79,2.46,0.22,11.03
maze-8192,0.97,42.35,64.00 M,575.91 M,13.63 GB,51.90,599.01,7,21,80.82,3.85,0.22,17.77
\end{filecontents*}
\pgfplotstabletypeset[
                % every even column/.style={
                %                 column type/.add={>{\columncolor[gray]{.8}}}{}
                % },
                columns/pixels/.style={column name = states},
                columns/transitions/.style={column name = trans.},
                columns/autSizeH/.style={column name = aut file size},
                %columns/conversion/.style={column name = time},
                %columns/img2aut/.style={column name = time},
                columns/conversionimg2aut/.style={column name = time},
                columns/convWIOfull/.style={column name = t.w.IO},
                columns/statesMin/.style={column name = stat.},
                columns/transitionsMin/.style={column name = trans.},
                columns/min_computation/.style={column name = time},
                columns/minWIO/.style={column name = t.w.IO},
                columns/convertBack/.style={column name = t.back},
                %columns/modelCheckingMin/.style={column name = ch. min},
                columns/modelCheckingMin/.style={column name = t.min},
                %columns/modelCheckingFull/.style={column name = ch. full},
                columns/modelCheckingFull/.style={column name = t.full},
                columns/computation/.style={column name = total},
                columns/speedupMC/.style={column name = gain},
                %columns/total/.style={column name = total},
                col sep=comma,
                string type,
                every column/.style={column type={|r}},
                every last column/.style={column type={|r|}},
                columns/label/.style={column type={|l},column name = name},
                %every last column/label/.style={column type={|l|}}
                every head row/.style={before row=\hline % \toprule
                                & 
                                \multicolumn{2}{c}{Encoding} & %{Conversion} &   
                                \multicolumn{3}{|c}{Full model} & 
                                \multicolumn{5}{|c}{Minimisation} &
                                \multicolumn{3}{|c|}{Model checking} 
                                % &
                                % \multicolumn{1}{|c|}{}                            
                                \\ \hline,
                                after row=\hline},
                every even row/.style={before row={\rowcolor{lightgray!30!white}}},                
                every last row/.style={after row=\hline},
                %outfile=tmp-results-table.tex
        ]{experiments/macos_results_arm64_mCRL2sccnew/result_maze_newToolchain/results-table_new.csv} 
        %} % rotate
        \vspace{6pt}
        \caption{%(Macos\_arm64\_NewSCC) 
        Results for the maze case study. All times are in seconds, rounded to two decimals. From left to right: encoding time for the image ({\sf  png} file) into automata format ({\sf aut} file) conversion, without and with I/O; number of states, transitions, and {\sf aut} file size of full model; minimisation time, without and with I/O; number of states and transitions of minimal \lts; time to convert the minimal ({\sf aut} file) model into a \qdcm\ ({\sf json} file); time for model checking the full model with \voxlogica{}, and the minimal model with \graphlogica{}; model checking speed-up.
        %; sum of times for conversion, minimization, backwards conversion, and model checking the minimal model.
        }\label{tab:resultsTable2mazeMacArm64scc}
\end{table}

The obtained speed-up (last column) is noteworthy, ranging from~3 to more than 17~times faster model checking when using minimised models, especially for the larger images, as shown in the last lines in the right-most column (gain) in the Table~\ref{tab:resultsTable2mazeMacArm64scc}. 
Note that the minimal model has the same number of states and transitions for each of the scaled images. 
This is as expected, as the size of each element in the images does not matter for \cop-bisimilarity, but only the (conditional) reachability between elements of different colour as established by the spatial logic. 
The processing times to produce the encoding of the full model increase with the size of the input image. 
Also the generation of the intermediate result files in the tool chain is rather large and requires time to be produced.
This is an aspect that can be overcome in a subsequent implementation of the toolchain in which intermediate results are stored internally.
The current experimental toolchain was built with the aim to first get insight in the potential gains of the model checking times. 
Note, however, that the experimental version of \voxminx\ is nevertheless already able to handle images of considerable size, i.e.\ 64M~pixels.
The time to translate the model checking results obtained from the minimal model back to the full image is not constant, but varies with the size of the input image. 
This is because this translation has to take the equivalence classes of pixels into account, which tend to be larger in larger images, thus requiring more processing time.

Similar observations can be made for the results for the more involving monoscope image shown in Table~\ref{tab:resultsTable2monoMacArm64scc}, that has the same structure as Table~\ref{tab:resultsTable2mazeMacArm64scc}. 
In this case, a single property is evaluated that is characterised by a deep nesting of conditional reachability operators. 
For the largest of this set of images the speed-up is more than 25~times that of the model checking time for the full (i.e.\ non-minimised) version. 
This more complex image leads to minimal models with a considerable number of states, ranging from~155 for the smallest image, to~945 for the largest image. 
The reason that for this image the minimal models are not all identical is that downscaling of the image reduces the number of details that can be distinguished (e.g.\ some thin lines simply disappear). 
This leads to a smaller number of equivalence classes for the smaller images. 
It also explains why for the largest model, with more details, the
model checking time of the full model is significantly higher, while there is an only limited increase in the model checking time of the minimal model.

\begin{table}%[h!]
       % \centering
        \scriptsize
        \hspace{-4em}
        %\rotatebox{90}{
        \begin{filecontents*}{experiments/macos_results_arm64_mCRL2sccnew/result_mono_newToolchain/results-table_new.csv}
label,conversionimg2aut,convWIOfull,pixels,transitions,autSizeH,min_computation,minWIO,statesMin,transitionsMin,convertBack,modelCheckingFull,modelCheckingMin,speedupMC
mono-130,0.18,0.21,30.47 K,272.05 K,4.83 MB,0.01,0.22,155,899,0.22,0.76,0.25,3.01
mono-260,0.26,0.28,121.88 K,1.07 M,20.27 MB,0.05,0.74,315,1841,0.40,1.06,0.28,3.78
mono-540,0.19,0.54,506.25 K,4.44 M,90.33 MB,0.20,3.31,460,2766,0.69,0.87,0.28,3.05
mono-1080,0.21,1.49,1.98 M,17.78 M,384.28 MB,0.97,14.57,945,6965,2.31,1.26,0.33,3.84
mono-2160,0.28,5.36,7.91 M,71.16 M,1.55 GB,3.79,63.66,945,6965,9.50,2.67,0.33,8.15
mono-4320,0.58,21.13,31.64 M,284.70 M,6.65 GB,15.71,258.92,945,6965,42.15,8.28,0.33,25.21
\end{filecontents*}
\pgfplotstabletypeset[
                % every even column/.style={
                %                 column type/.add={>{\columncolor[gray]{.8}}}{}
                % },
                columns/pixels/.style={column name = states},
                columns/transitions/.style={column name = trans.},
                columns/autSizeH/.style={column name = aut file size},
                %columns/conversion/.style={column name = time},
                %columns/img2aut/.style={column name = time},
                columns/conversionimg2aut/.style={column name = time},
                columns/convWIOfull/.style={column name = t.w.IO},
                columns/statesMin/.style={column name = stat.},
                columns/transitionsMin/.style={column name = trans.},
                columns/min_computation/.style={column name = time},
                columns/minWIO/.style={column name = t.w.IO},
                columns/convertBack/.style={column name = t.back},
                %columns/modelCheckingMin/.style={column name = ch. min},
                columns/modelCheckingMin/.style={column name = t.min},
                %columns/modelCheckingFull/.style={column name = ch. full},
                columns/modelCheckingFull/.style={column name = t.full},
                columns/computation/.style={column name = total},
                columns/speedupMC/.style={column name = gain},
                %columns/total/.style={column name = total},
                col sep=comma,
                string type,
                every column/.style={column type={|r}},
                every last column/.style={column type={|r|}},
                columns/label/.style={column type={|l},column name = name},
                %every last column/label/.style={column type={|l|}}
                every head row/.style={before row=\hline % \toprule
                                & 
                                \multicolumn{2}{c}{Encoding} & %{Conversion} &   
                                \multicolumn{3}{|c}{Full model} & 
                                \multicolumn{5}{|c}{Minimisation} &
                                \multicolumn{3}{|c|}{Model checking} 
                                % &
                                % \multicolumn{1}{|c|}{}                            
                                \\ \hline,
                                after row=\hline},
                every even row/.style={before row={\rowcolor{lightgray!30!white}}},                
                every last row/.style={after row=\hline},
                %outfile=tmp-results-table.tex
        ]{experiments/macos_results_arm64_mCRL2sccnew/result_mono_newToolchain/results-table_new.csv}
        %}%rotate
        \vspace{6pt}
        \caption{%(Macos\_arm64\_NewSCC)
        Results for the monoscope case study. All times are in seconds, rounded to two decimals. From left to right: encoding time for the image ({\sf  png} file) into automata format ({\sf aut} file) conversion, without and with I/O; number of states, transitions, and {\sf aut} file size of full model; minimisation time, without and with I/O; number of states and transitions of minimal \lts; time to convert the minimal ({\sf aut} file) model into a \qdcm\ ({\sf json} file); time for model checking the full model with \voxlogica{}, and the minimal model with \graphlogica{}; model checking speed-up.
        %; sum of times for conversion, minimization, backwards conversion, and model checking the minimal model.
        }\label{tab:resultsTable2monoMacArm64scc}
\end{table}

For what concerns the results for the Pac-Man scene, shown in Table~\ref{tab:resultsTable2PacManMacArm64scc}, rescaling of the image did not change the structure of the image. 
Therefore, in this case, the number of states (and transitions) of the minimal model are the same for all cases, i.e.\ 35~states and 155~transitions. 
The speed-up that is obtained is somewhat smaller compared to that found for the other two cases. 
However, it is still 7~times faster to perform model checking on the minimal model, compared to model checking the full model.

Regarding this experimental evaluation, it is noted that the \voxlogica\ model checker is highly optimised for images, whereas the \graphlogica\ variant is working on general graphs and more of an experimental, and less optimised, nature. 
So the results on speed-up of the model checking times are actually underestimated. Furthermore, once a minimal model has been generated, it can be used for many model checking sessions, which is increasing the advantage considering computation time. 

The plots in Figure~\ref{fig:plotAllVMX} provide a visual presentation of some of the numbers in the previous tables. 
Figure~\ref{fig:statesspeedVMX} shows how the speed-up is related to the size of an image, for the maze (blue), the monoscope (red), and the Pac-Man scene (brown), respectively. 
For the maze, the increase in speed-up and the increase in model checking time of the full model appears not to be linear. 
Presumably, full model checking is relatively fast for the larger models, probably due to some possibilities to reuse partial results in the optimisation of the model checking procedure of \voxlogica. 
Figure~\ref{fig:statesFullMCVMX} shows how the full model checking time relates to the size of the images. 
It shows a similar pattern as in Figure~\ref{fig:statesspeedVMX}. 
Finally, in Figure~\ref{fig:statesMinVMX} it is shown how the minimisation time relates to the size of the images. 
This seems to follow essentially a linear pattern for each case. 

\smallskip

\noindent
In this section we have studied spatial model checking for digital images using model minimisation based on \cop-bisimilarity. 
We have used \voxminx\ on several types of images, ranging from the synthetic 2D maze example to images from the `real world' such as the Pac-Man scene and the monoscope example. 
The analysis shows promising results for what concerns the speed-up in model checking time that can be obtained using minimal models. 
Moreover, it also shows that the approach is applicable to images of a size that is easily found in various application areas, for example, the largest maze consists of 64M pixels. 
At the same time, it also shows that a further integration of the tools that the \voxminx\ toolchain is composed of would be very beneficial to reduce resources needed for the various transformations. 
We envision that such improvements would facilitate the wider applicability of the method. 
In particular, we are interested in its application in the domain of medical imaging. \voxlogica\ has already been applied to that area, see for example~\cite{BelmonteCM25}, where the contouring of brain tumours and the segmentation of white and grey matter of the brain were addressed. 
However, the method could find its way to other medical applications and many other domains in which spatial analysis is of interest~\cite{CianciaLM25,BezhanishviliBCGJLMV26,ACLM2026,BelmonteCLM26}.

\begin{table}%[h!]
       % \centering
        \scriptsize
        \hspace{-4em}
        %\rotatebox{90}{
        \begin{filecontents*}{experiments/macos_results_arm64_mCRL2sccnew/result_pacman_newToolchain/results-table_new.csv}
label,conversionimg2aut,convWIOfull,pixels,transitions,autSizeH,min_computation,minWIO,statesMin,transitionsMin,convertBack,modelCheckingFull,modelCheckingMin,speedupMC
pm-200,0.18,0.22,39.06 K,349.22 K,6.29 MB,0.02,0.26,35,155,0.22,0.71,0.22,3.25
pm-400,0.18,0.30,156.25 K,1.37 M,26.56 MB,0.08,0.95,35,155,0.33,0.73,0.22,3.30
pm-600,0.18,0.46,351.56 K,3.08 M,61.97 MB,0.17,2.40,35,155,0.51,0.75,0.22,3.37
pm-800,0.18,0.63,625.00 K,5.48 M,111.49 MB,0.31,4.57,35,155,0.79,0.75,0.22,3.39
pm-1000,0.19,0.82,976.56 K,8.57 M,175.13 MB,0.46,7.21,35,155,1.12,0.77,0.22,3.49
pm-2000,0.23,2.68,3.81 M,34.31 M,756.49 MB,1.92,31.48,35,155,4.46,0.96,0.22,4.34
pm-4000,0.38,10.30,15.26 M,137.28 M,3.11 GB,7.93,134.44,35,155,18.83,1.62,0.22,7.34
\end{filecontents*}
\pgfplotstabletypeset[
                % every even column/.style={
                %                 column type/.add={>{\columncolor[gray]{.8}}}{}
                % },
                columns/pixels/.style={column name = states},
                columns/transitions/.style={column name = trans.},
                columns/autSizeH/.style={column name = aut file size},
                %columns/conversion/.style={column name = time},
                %columns/img2aut/.style={column name = time},
                columns/conversionimg2aut/.style={column name = time},
                columns/convWIOfull/.style={column name = t.w.IO},
                columns/statesMin/.style={column name = stat.},
                columns/transitionsMin/.style={column name = trans.},
                columns/min_computation/.style={column name = time},
                columns/minWIO/.style={column name = t.w.IO},
                columns/convertBack/.style={column name = t.back},
                %columns/modelCheckingMin/.style={column name = ch. min},
                columns/modelCheckingMin/.style={column name = t.min},
                %columns/modelCheckingFull/.style={column name = ch. full},
                columns/modelCheckingFull/.style={column name = t.full},
                columns/computation/.style={column name = total},
                columns/speedupMC/.style={column name = gain},
                %columns/total/.style={column name = total},
                col sep=comma,
                string type,
                every column/.style={column type={|r}},
                every last column/.style={column type={|r|}},
                columns/label/.style={column type={|l},column name = name},
                %every last column/label/.style={column type={|l|}}
                every head row/.style={before row=\hline % \toprule
                                & 
                                \multicolumn{2}{c}{Encoding} & %{Conversion} &   
                                \multicolumn{3}{|c}{Full model} & 
                                \multicolumn{5}{|c}{Minimisation} &
                                \multicolumn{3}{|c|}{Model checking} 
                                % &
                                % \multicolumn{1}{|c|}{}                            
                                \\ \hline,
                                after row=\hline},
                every even row/.style={before row={\rowcolor{lightgray!30!white}}},                
                every last row/.style={after row=\hline},
                %outfile=tmp-results-table.tex
        ]{experiments/macos_results_arm64_mCRL2sccnew/result_pacman_newToolchain/results-table_new.csv}
        %}%rotate
        \vspace{6pt}
        \caption{%(Macos\_arm64\_NewSCC) 
        Results for the Pac-Man case study. All times are in seconds, rounded to two decimals. From left to right: encoding time for the image ({\sf png} file) into automata format ({\sf aut} file) conversion, without and with I/O; number of states, transitions, and {\sf aut} file size of full model; minimisation time, without and with I/O; number of states and transitions of minimal \lts; time to convert the minimal ({\sf aut} file) model into a \qdcm\ ({\sf json} file); time for model checking the full model with \voxlogica{}, and the minimal model with \graphlogica{}; model checking speed-up.
        %; sum of times for conversion, minimization, backwards conversion, and model checking the minimal model.
        }\label{tab:resultsTable2PacManMacArm64scc}
\end{table}

\begin{figure}
\centering
        \subfloat[][]
        {
\resizebox{!}{1.4in}{
\begin{tikzpicture}
        \begin{axis}[
            xlabel={States},
            ylabel={Speed up},
            grid = major,
            line width = 1pt
        ]
            \addplot+ [sharp plot] coordinates {(16,3.33) (64,3.30) (256,3.33) (1000,3.48) (4000,4.35) (16000,11.03) (64000,17.77)};
            \addplot+ [sharp plot] coordinates {(30.47,3.01) (121.88,3.78) (506.25,3.05) (1980,3.84) (7910,8.15) (31640,25.21)};
            \addplot+ [sharp plot] coordinates {(39.06,3.25) (156.25,3.30) (351.56,3.37) (625,3.39) (976.56,3.49) (3810,4.34) (15260,7.34)};
            %\addplot table [x=convWIO, y=speedupMC, col sep=comma] {experiments/current/results-table1_maze_2025_08_26.csv};
        \end{axis}
\end{tikzpicture}
}\label{fig:statesspeedVMX}
}\quad
\centering
    \subfloat[][]
      {
\resizebox{!}{1.4in}{
\begin{tikzpicture}
        \begin{axis}[
            xlabel={States},
            ylabel={Full MC time [s]},
            grid = major,
            line width = 1pt
        ]
            \addplot+ [sharp plot] coordinates {(16,0.72) (64,0.71) (256,0.72) (1000,0.75) (4000,0.94) (16000,2.46) (64000,3.85)};
            \addplot+ [sharp plot] coordinates {(30.47,0.76) (121.88,1.06) (506.25,0.87) (1980,1.26) (7910,2.67) (31640,8.28)};
            \addplot+ [sharp plot] coordinates {(39.06,0.71) (156.25,0.73) (351.56,0.75) (625,0.75) (976.56,0.77) (3810,0.96) (15260,1.62)};
        \end{axis}
\end{tikzpicture}
}\label{fig:statesFullMCVMX}
}\quad
\centering
    \subfloat[][]
      {
\resizebox{!}{1.4in}{
\begin{tikzpicture}
        \begin{axis}[
            xlabel={States},
            ylabel={Minimisation time [s]},
            grid = major,
            line width = 1pt
        ]
            \addplot+ [sharp plot] coordinates {(16,0.0) (64,0.02) (256,0.11) (1000,0.48) (4000,2.32) (16000,9.89) (64000,51.90)};
            \addplot+ [sharp plot] coordinates {(30.47,0.01) (121.88,0.05) (506.25,0.20) (1980,0.97) (7910,3.79) (31640,15.71)};
            \addplot+ [sharp plot] coordinates {(39.06,0.02) (156.25,0.08) (351.56,0.17) (625,0.31) (976.56,0.46) (3810,1.92) (15260,7.93)};
        \end{axis}
\end{tikzpicture}
}\label{fig:statesMinVMX}
}
\caption{\voxminx{:} Model size in number of states vs. speed-up factor for the maze (blue), the monoscope (red) and the Pac-Man (brown) examples (\ref{fig:statesspeedVMX}). For the same models: number of states vs. model checking of full model with \voxlogica{} (\ref{fig:statesFullMCVMX}) and number of states vs.\ minimisation time without IO (\ref{fig:statesMinVMX}).}\label{fig:plotAllVMX}
\end{figure}
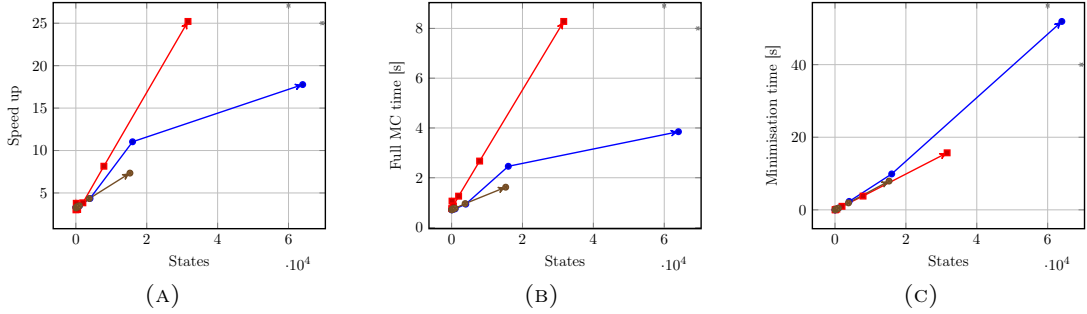

%% file: Conclusions.tex
\section{Conclusions and Future Work}

Traditional model checking is a widely used verification technique for ensuring that a model of system behaviour conforms to a logical specification of desired properties. {\em Spatial} model checking, where a model of space is checked against a spatial logic formula, expressing {\em spatial} properties, has proven a valuable verification technique. It has been successfully applied in various domains, in particular in medical image analysis.  One way to improve the performance of spatial model checking is by model reduction. In this paper we have shown that this can be obtained exploiting the logical characterisation of \cop-bisimilarity by the \icrl{} logic.
(In-)finitary Compatible Reachability Logic includes conjunction and two conditional reachability modalities: one forward and one backward modality.

A practical and feasible minimisation method has been proposed for \cop-bisimilarity for finite closure models. 
The latter are a convenient theoretical framework for model checking spatial logics. 
The method relies on an encoding of (finite) closure models into suitable \lts{s} such that an existing efficient algorithm for branching bisimilarity can be used to obtain a minimal model. 
The encoding has been proven correct, in the sense that two points in the closure models are \cop-bisimilar if and only if the states they are mapped to by the encoding are branching bisimilar in the corresponding \lts. 
The correctness proof exploits induction  based on the novel notion of depth of a state in an \lts. 

An implementation of the encoding has been developed for the special case of symmetric closure models representing digital images.
The implementation is part of a toolchain, \voxminx, that uses a branching bisimilarity minimisation procedure from the \mCRLT\ toolsuit for the minimisation of the \lts\ encoding the original image. 
It also translates the spatial model checking results for the minimal model back to the pixels in the original image that was taken as input. 
Doing so, model checking results can be directly made visible on the input image itself. 
For example, \voxminx\ can highlight the pixels that satisfy the spatial property of interest while showing the other pixels in a shaded way. 
This provides the user with an immediate and informative visual feedback of the results of spatial model checking on images.

A feasibility study has been performed for \voxminx.
This study uses a benchmark of images and formulas and provides insight in the potential of the minimisation method for its use in the analysis of, possibly large, 2D~images. 
The benchmark consists of three series of representative images. Each image has been evaluated at six or seven different resolutions to be able to evaluate the performance for increasing sizes of the images.
For each set of images, the results confirm that a very promising speed-up of spatial model checking can be obtained for single formulas, also for images with a huge, but realistic, size. This insight is particularly relevant for the envisioned use of spatial model checking in the medical domain. 

Of course, minimisation pays off more when multiple formulas are checked on the same model, which is common in formal verification. 
In such a scenario, the multiple model checking time for the full model can in general be expected to be substantially longer than the sum of the conversion, minimisation, backwards conversion, and multiple model checking time of the minimal model, even for the current prototype.
The advantage increases, of course, when multiple properties are checked on a single minimised model. Furthermore, bisimulation-based minimisation of images may also serve as a particular form of compression of large images. Such compression preserves all the spatial properties of the original image that can be expressed by the spatial logic, but requires, in general, much less space to be stored.

Ongoing work, also taking into account the results presented in~\cite{FZMasterThesis}, is devoted to translating spatial-logic properties to the language of \mCRLT\ in order to use its state-of-the-art model checking techniques to verify spatial properties of directed graphs, in order to leverage the obtained speed-up even further.
Future work aims at further optimisations of the representations of the models and an integration of the components present in the \voxminx\ toolchain. 
The basic ingredients for such a mapping, i.e. the sets of states in the equivalence classes of the bisimulation, are readily available using the \mCRLT\ tool suite~\cite{Bu+19}.

%% file: AppendixExample.tex
\newpage
\section{Running Example for \voxminx}\label{apx:running}

To illustrate the various intermediate steps of the \voxmin\ toolchain in Figure~\ref{fig:toolchainext} and the file formats that are involved, we show the toolchain at work on a micro example of an image of 2~by~2 pixels, composed of three white pixels and one blue, as shown in Figure~\ref{fig:running}.

\begin{figure}
\centering
\input{pict/img_grid2by2.tex}
\caption{Running example: 4 pixels, 1 blue and 3 white}\label{fig:running}
\end{figure}
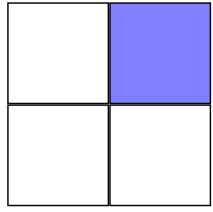

\paragraph{1) Encoding phase}
The encoding (Box 3 in Figure~\ref{fig:toolchainext}) of the image in
Figure ~\ref{fig:running} is produced with the convert-option of
\graphlogica\ that implements the encoding procedure of Section~\ref{sec:OptimisedEncoding}. 
The resulting \lts\ (in the \textsf{aut} format accepted by \mCRLT) is the following:
\begin{verbatim}
des (0,16,4)
(0, cFFFFFF,0)
(0,change,1)
(0,tau,2)
(0,tau,3)
(1, c3F51B5,1)
(1,change,0)
(1,change,2)
(1,change,3)
(2, cFFFFFF,2)
(2,tau,0)
(2,change,1)
(2,tau,3)
(3, cFFFFFF,3)
(3,tau,0)
(3,change,1)
(3,tau,2)
\end{verbatim}
This format reads as follows. The first line is a descriptor stating that the \lts\ has initial state~0, and it is composed of 16~transitions and 4~states. 
In the next 16~lines the transitions are listed as triples composed of the source state, the label and the target state.  
The labels are either \textsf{tau} or \textsf{change}, or one of the two colours in the image, \textsf{cFFFFFF} (denoting white) and \textsf{c3F51B5} (denoting the particular shade of blue). 
Let's call this file \textsf{pm\_1b3w.aut}.

\graphlogica\ also produces a frame file from the images if the {\sf --framefile} option is added. This file, {\sf pm\_1b3w.mcrl}, is generated in the \mCRLT\ model specification format: 
\begin{verbatim}
act 
cFFFFFF,
c3F51B5,
change;
init delta;
\end{verbatim}
The only purpose of this \mCRLT\ file is to define the three action
labels, \textsf{cFFFFFF},  \textsf{c3F51B5} and \textsf{change} (the \textsf{tau}-action is assumed to be present by default). 
The frame file does not specify any behaviour (as denoted by the \textsf{delta} process that stands for inaction).

\paragraph{2) Conversion of the \textsf{aut} file and frame file to the internal \textsf{lts} format}
%% \paragraph{2) Conversion of the {\sf .aut} and frame file into the internal {\sf \lts} format of \mCRLT}
In \voxminx\ the frame file {\sf pm\_1b3w.mcrl}, together with the file {\sf pm\_1b3w.aut}, are converted into an \lts\ in the {\sf lts} format of \mCRLT\ using the {\sf ltsconvert} function with the {\sf -l} option (see Box 4 in Figure~\ref{fig:toolchainext}). 

\paragraph{3) \lts\ minimisation}
The \textsf{pm\_1b3w.lts} file, produced in the previous step, is minimised using the  \textsf{ltsconvert --add-state-as-state-label -ebranching-bisim} operation (also this occurs in Box~4 in Figure~\ref{fig:toolchainext}). 
This results in a minimised \lts\ in \textsf{lts} format (\textsf{pm\_1b3w.min.lts}). 
From the latter, information on the relationship between the states in the full \lts\ and the minimised \lts\ can be obtained using the \textsf{ltsinfo} operator (see Box~9 in Figure~\ref{fig:toolchainext}) on the minimal \lts\ resulting in the following file:
\begin{verbatim}
Number of states: 2.
Number of action labels: 4 (including a tau label).
Number of transitions: 4.
Number of state labels: 2.
LTS is deterministic.
This lts has no probabilistic states.
The state labels of this labelled transition system:
0: (1).
1: (3).
1: (2).
1: (0).
\end{verbatim}
The above file shows that the minimised \lts\ has 2~states: state~0 and state~1.
These represent the two equivalence classes. 
The last four lines report how the states of the minimised \lts\ (states 0 and~1) are related to the original states of the encoded \lts\ (0, 1, 2,~3). 
In particular it says that state~1 of the original model in \textsf{pm\_1b3w.aut} is mapped to equivalence class~0 and the other three states (0, 2 and~3, between brackets) are mapped to the equivalence class~1.  
Note that state~1 in the encoded \lts\ in \textsf{pm\_1b3w.aut} has indeed a transition (self-loop) labeled with {\sf c3F51B5}, i.e.\ denoting blue, and states 0, 2 and~3 have self-loops with label \textsf{cFFFFFF}, i.e.\ denoting white.

\paragraph{4) Producing a graph model for model checking with \graphlogica}
The minimal \lts, in turn, is translated back (see Box~5 in Figure~\ref{fig:toolchainext}), first into the \textsf{aut} format, using the \mCRLT\ operation \textsf{ltsconvert -enone --in=lts pm\_1b3w.min.lts --out=aut pm\_1b3w.min.aut}, resulting in the following file:
\begin{verbatim}
des (1,4,2)
(1,"change",0)
(1,"cFFFFFF",1)
(0,"c3F51B5",0)
(0,"change",1)
\end{verbatim}
Next, this \textsf{aut} file is transformed into a \textsf{json} file, representing a finite closure model, accepted by the model checker \graphlogica\ shown below:
\begin{verbatim}
{
  "nodes": [
    {
      "id": "1",
      "atoms": [
        "cFFFFFF"
      ]
    },
    {
      "id": "0",
      "atoms": [
        "c3F51B5"
      ]
    }
  ],
  "arcs": [
    {
      "source": "1",
      "target": "0"
    },
    {
      "source": "0",
      "target": "1"
    }
  ]
}
\end{verbatim}
The above transformation is provided by the \graphlogica\ tool itself, using the following command \textsf{GraphLogicA --convert pm\_1b3w.min.aut pm\_1b3w.min.json}. 
The latter \textsf{json} file is used by \graphlogica\ for spatial model checking on the closure model. 
In particular, this file is called from within the logical specification file given as input to \graphlogica. 
An example is shown below.
\begin{verbatim}
load graph = "pm_1b3w.min.json"

let white = ap("cFFFFFF")
let blue = ap("c3F51B5")

// wtchb: property "white pixels that are touching blue", defined below
let wtchb = white & touch(white,blue)

save "white.json" white
save "blue.json" blue

save "wtchb.json" wtchb
\end{verbatim}
The first line loads the minimised model as a closure model with points in the model representing the equivalence classes. 
It then defines two atomic propositions, one for the white point (``node" with id=1)  and one for the blue point (``node" with id=0) in the model, and a simple property for white points touching the blue one. 
Finally, three results are saved (in \textsf{json} format). 
The model checking result \textsf{wtchb.json}, for the white points touching blue ones, is as follows.
\begin{verbatim}
{
  "nodes": [
    {
      "id": "1",
      "atoms": [
        "result",
        "cFFFFFF"
      ]
    },
    {
      "id": "0",
      "atoms": [
        "c3F51B5"
      ]
    }
  ],
  "arcs": [
    {
      "source": "1",
      "target": "0"
    },
    {
      "source": "0",
      "target": "1"
    }
  ]
}
\end{verbatim} 
Note that this result file looks very similar to the \textsf{json} model file that \graphlogica\ takes as input, but now the points (i.e.\ equivalence classes) that satisfy the property (\textsf{wtchb} in this case) have an additional label, \textsf{result}. 
In the above file the node with label "\textsf{cFFFFFF}" has this additional label.

\paragraph{5) Translating the result back to the original images}
Finally, we convert the decorated graph model in \textsf{json} format back into an \textsf{aut} file using \graphlogica\ and the convert option. 
We then extract, via the Python script \textsf{resultaut2rts.py}, from the \textsf{aut} file an \textsf{rts} file that lists all the nodes of the minimal model that satisfy the property \textsf{wtchb}. 
There is only one such class, and this is class~1.
\begin{verbatim}
[
  1
]
\end{verbatim}
This result can be projected (see Box~11 in Figure~\ref{fig:toolchainext}) onto the original image via another Python script \textsf{glresults.py} that takes as input the original image, the \textsf{pm\_1b3w.min.info} file and the \textsf{wtchb.rts} file. 
The resulting image has four pixels, the three white ones where the
property \textsf{wtchb} holds, and one black one denoting that
property \textsf{wtchb} does not hold in the blue pixel.
The result is shown in Figure~\ref{fig:runningfinal}.

\begin{figure}
\centering
\input{pict/img_grid2by2MCresult.tex}
\caption{Running example model checking result for atomic proposition {\sf white}: 4 pixels, 1 black and 3 white.}\label{fig:runningfinal}
\end{figure}
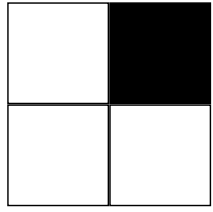

The procedure involves various file transformations and the writing and reading of intermediate files. 
This is because the current toolchain is a prototype and allowing to study the time each intermediate step takes. 
In future work we plan to integrate the various steps of the procedure keeping intermediate results in appropriate data structures avoiding time and memory lost in~I/O.

%% file: pict/img_grid2by2.tex
%% file img_grid2by2.tex used in Appendix %%

%\scriptsize
\scalebox{0.9}{
\begin{tikzpicture}[baseline=1,
  <->, >=stealth', semithick, shorten >=0.5pt, shorten <=0.5pt,
  every state/.style={
    %circle, minimum size=12pt, inner sep=0.5pt, draw},
    rectangle, minimum size=42pt, inner sep=0.5pt, draw},
  ]
  
  \node [state, fill=white!50] (x11) at (1.0,1.0) {\phantom{21}};
  \node [state, fill=white!50] (x12) at (1.0,2.5) {\phantom{16}};
%  \node [state, fill=green!50] (x13) at (1.0,4.0) {\phantom{11}};
%  \node [state, fill=green!50] (x14) at (1.0,5.5) {\phantom{6}};
%  \node [state, fill=green!50] (x15) at (1.0,7.0) {\phantom{1}};

  \node [state, fill=white!50] (x21) at (2.5,1.0) {\phantom{22}};
  \node [state, fill=blue!50  ] (x22) at (2.5,2.5) {\phantom{17}};
%  \node [state, fill=red!50  ] (x23) at (2.5,4.0) {\phantom{12}};
%  \node [state, fill=red!50  ] (x24) at (2.5,5.5) {\phantom{7}};
%  \node [state, fill=green!50] (x25) at (2.5,7.0) {\phantom{2}};

%  \node [state, fill=green!50] (x31) at (4.0,1.0) {\phantom{23}};
%  \node [state, fill=red!50  ] (x32) at (4.0,2.5) {\phantom{18}};
%  \node [state, fill=red!50  ] (x33) at (4.0,4.0) {\phantom{13}};
%  \node [state, fill=red!50  ] (x34) at (4.0,5.5) {\phantom{8}};
%  \node [state, fill=green!50] (x35) at (4.0,7.0) {\phantom{3}};
%
%  \node [state, fill=green!50] (x41) at (5.5,1.0) {\phantom{24}};
%  \node [state, fill=red!50  ] (x42) at (5.5,2.5) {\phantom{19}};
%  \node [state, fill=red!50  ] (x43) at (5.5,4.0) {\phantom{14}};
%  \node [state, fill=red!50  ] (x44) at (5.5,5.5) {\phantom{9}};
%  \node [state, fill=green!50] (x45) at (5.5,7.0) {\phantom{4}};
%
%  \node [state, fill=green!50] (x51) at (7.0,1.0) {\phantom{25}};
%  \node [state, fill=green!50] (x52) at (7.0,2.5) {\phantom{20}};
%  \node [state, fill=green!50] (x53) at (7.0,4.0) {\phantom{15}};
%  \node [state, fill=green!50] (x54) at (7.0,5.5) {\phantom{10}};
%  \node [state, fill=green!50] (x55) at (7.0,7.0) {\phantom{5}};

\end{tikzpicture}
}

%% file: pict/img_grid2by2MCresult.tex
%% file img_grid2by2.tex used in Appendix %%

%\scriptsize
\scalebox{0.9}{
\begin{tikzpicture}[baseline=1,
  <->, >=stealth', semithick, shorten >=0.5pt, shorten <=0.5pt,
  every state/.style={
    %circle, minimum size=12pt, inner sep=0.5pt, draw},
    rectangle, minimum size=42pt, inner sep=0.5pt, draw},
  ]
  
  \node [state, fill=white!50] (x11) at (1.0,1.0) {\phantom{21}};
  \node [state, fill=white!50] (x12) at (1.0,2.5) {\phantom{16}};
%  \node [state, fill=green!50] (x13) at (1.0,4.0) {\phantom{11}};
%  \node [state, fill=green!50] (x14) at (1.0,5.5) {\phantom{6}};
%  \node [state, fill=green!50] (x15) at (1.0,7.0) {\phantom{1}};

  \node [state, fill=white!50] (x21) at (2.5,1.0) {\phantom{22}};
  \node [state, fill=black ] (x22) at (2.5,2.5) {\phantom{17}};
%  \node [state, fill=red!50  ] (x23) at (2.5,4.0) {\phantom{12}};
%  \node [state, fill=red!50  ] (x24) at (2.5,5.5) {\phantom{7}};
%  \node [state, fill=green!50] (x25) at (2.5,7.0) {\phantom{2}};

%  \node [state, fill=green!50] (x31) at (4.0,1.0) {\phantom{23}};
%  \node [state, fill=red!50  ] (x32) at (4.0,2.5) {\phantom{18}};
%  \node [state, fill=red!50  ] (x33) at (4.0,4.0) {\phantom{13}};
%  \node [state, fill=red!50  ] (x34) at (4.0,5.5) {\phantom{8}};
%  \node [state, fill=green!50] (x35) at (4.0,7.0) {\phantom{3}};
%
%  \node [state, fill=green!50] (x41) at (5.5,1.0) {\phantom{24}};
%  \node [state, fill=red!50  ] (x42) at (5.5,2.5) {\phantom{19}};
%  \node [state, fill=red!50  ] (x43) at (5.5,4.0) {\phantom{14}};
%  \node [state, fill=red!50  ] (x44) at (5.5,5.5) {\phantom{9}};
%  \node [state, fill=green!50] (x45) at (5.5,7.0) {\phantom{4}};
%
%  \node [state, fill=green!50] (x51) at (7.0,1.0) {\phantom{25}};
%  \node [state, fill=green!50] (x52) at (7.0,2.5) {\phantom{20}};
%  \node [state, fill=green!50] (x53) at (7.0,4.0) {\phantom{15}};
%  \node [state, fill=green!50] (x54) at (7.0,5.5) {\phantom{10}};
%  \node [state, fill=green!50] (x55) at (7.0,7.0) {\phantom{5}};

\end{tikzpicture}
}